\documentclass[10pt,journal,compsoc]{IEEEtran}
\usepackage{cite}
\usepackage{amsmath,amssymb,amsfonts}
\usepackage{algorithmicx}
\usepackage[noend,ruled, linesnumbered, vlined]{algorithm2e}

\usepackage{array}
\usepackage{graphicx}
\usepackage{stfloats}
\usepackage{amsthm}
\usepackage{subfigure}
\usepackage{url}
\usepackage{verbatim}
\usepackage{algpseudocode}  
\usepackage{textcomp}
\usepackage{xcolor}
\usepackage{color}

\allowdisplaybreaks[4]
\def\BibTeX{{\rm B\kern-.05em{\sc i\kern-.025em b}\kern-.08em
    T\kern-.1667em\lower.7ex\hbox{E}\kern-.125emX}}
\newenvironment{sequation}{\begin{equation}\small}{\end{equation}}
\newenvironment{tequation}{\begin{equation}\footnotesize}{\end{equation}}

\newtheorem{lemma}{\textbf{Lemma}}
\newtheorem{theorem}{\textbf{Theorem}}

\newtheorem{definition}{\textbf{Definition}}
\newtheorem{corollary}{\textbf{Corollary}}
\usepackage{multirow}
\usepackage{bm}
\usepackage{booktabs}

\usepackage{epstopdf}
\usepackage{cases}
\usepackage{makecell,multirow,diagbox}
\usepackage{stfloats}
\usepackage{comment}
\usepackage{array}

\definecolor{b}{rgb}{0.0, 0, 1}
\definecolor{k}{rgb}{0, 0, 0}

\def\BibTeX{{\rm B\kern-.05em{\sc i\kern-.025em b}\kern-.08em
		T\kern-.1667em\lower.7ex\hbox{E}\kern-.125emX}}

\begin{document}

\title{TJCCT: A Two-timescale Approach for UAV-assisted Mobile Edge Computing}

\author{Zemin~Sun,~\IEEEmembership{}
        Geng~Sun,~\IEEEmembership{Member,~IEEE},
        Qingqing Wu,~\IEEEmembership{Senior Member,~IEEE}, 
        Long~He,
        Shuang Liang,
        Hongyang Pan,
        Dusit Niyato,~\IEEEmembership{Fellow,~IEEE},
        Chau Yuen,~\IEEEmembership{Fellow,~IEEE}, and
        Victor C. M. Leung,~\IEEEmembership{Life Fellow,~IEEE}
	\thanks{This study is supported in part by the National Natural Science Foundation of China (61872158, 62002133, 62172186, 62272194), and in part by the Science and Technology Development Plan Project of Jilin Province (20230201087GX). A part of this paper is accepted by IEEE INFOCOM 2024 (\textit{Corresponding author: Geng Sun.)}}
	\IEEEcompsocitemizethanks{\IEEEcompsocthanksitem  Zemin Sun, Geng Sun, Long He, and Hongyang Pan are with the College of Computer Science and Technology, Jilin University, Changchun 130012, China, and also with the Key Laboratory of Symbolic Computation and Knowledge Engineering of Ministry of Education, Jilin University, Changchun 130012, China.\protect\\
    E-mail:sunzemin@jlu.edu.cn, sungeng@jlu.edu.cn, helong0517@foxmail.com, panhongyang18@foxmail.com.
    \IEEEcompsocthanksitem Qingqing Wu is with the Department of Electronic Engineering, Shanghai Jiao Tong University, Shanghai, China.\protect\\
		E-mail: qingqingwu@sjtu.edu.cn.  
  	\IEEEcompsocthanksitem Shuang Liang is with the School of Information Science and Technology, Northeast Normal University, Changchun 130024, China.\protect\\
		E-mail: liangshuang@nenu.edu.cn.
        \IEEEcompsocthanksitem Dusit Niyato is with the School of Computer Science and Engineering, Nanyang Technological University, Singapore 639798.\protect\\
        E-mail: dniyato@ntu.edu.sg. 
        \IEEEcompsocthanksitem Chau Yuen is with the School of Electrical and Electronics Engineering, Nanyang Technological University, Singapore 639798.\protect\\
         E-mail: chau.yuen@ntu.edu.sg.
        \IEEEcompsocthanksitem Victor C. M. Leung is with the College of Computer Science and Software Engineering, Shenzhen University, Shenzhen 518060, China, and also with the Department of Electrical and Computer Engineering, University of British Columbia, Vancouver, BC V6T 1Z4, Canada. \protect\\
        E-mail: vleung@ieee.org.
       }}

\markboth{Journal of \LaTeX\ Class Files,~Vol.~, No.~, ~}%
{Shell \MakeLowercase{\textit{et al.}}: Bare Demo of IEEEtran.cls for Computer Society Journals}

\IEEEtitleabstractindextext{%
\begin{abstract}		
Unmanned aerial vehicle (UAV)-assisted mobile edge computing (MEC) is emerging as a promising paradigm to provide aerial-terrestrial computing services in close proximity to mobile devices (MDs). However, meeting the demands of computation-intensive and delay-sensitive tasks for MDs poses several challenges, including the demand-supply contradiction between MDs and MEC servers, the demand-supply heterogeneity between MDs and MEC servers, the trajectory control requirements on energy efficiency and timeliness, and the different time-scale dynamics of the network. To address these issues, we first present a hierarchical architecture by incorporating terrestrial-aerial computing capabilities and leveraging UAV flexibility. Furthermore, we formulate a joint computing resource allocation, computation offloading, and trajectory control problem to maximize the system utility. Since the problem is a non-convex and NP-hard mixed integer nonlinear programming (MINLP), we propose a two-timescale joint computing resource allocation, computation offloading, and trajectory control (TJCCT) approach for solving the problem. In the short timescale, we propose a price-incentive model for on-demand computing resource allocation and a matching mechanism-based method for computation offloading. In the long timescale, we propose a convex optimization-based method for UAV trajectory control. Besides, we theoretically prove the stability, optimality, and polynomial complexity of TJCCT. Extended simulation results demonstrate that the proposed TJCCT outperforms the comparative algorithms in terms of the system utility, average processing rate, average completion delay, and average completion ratio.
\end{abstract}
	
\begin{IEEEkeywords}
		UAV-assisted MEC network, computing resource allocation, computation offloading, trajectory control
\end{IEEEkeywords}}

\maketitle
\IEEEdisplaynontitleabstractindextext
\IEEEpeerreviewmaketitle

%
%

\section{Introduction}
\label{sec_introduction}

\par  \IEEEPARstart{T}{he} development of wireless communications and the proliferation of mobile devices (MDs) triers various emerging applications, such as autonomous driving, online gaming, and augmented reality. These applications often require extensive computing resources and low latency to satisfy the quality of experience (QoE). However, fulfilling the computation-intensive and delay-sensitive computation tasks of these applications poses a great challenge to MDs with insufficient computational capability and finite energy capacity. To tackle this challenge, mobile edge computing (MEC) has been identified as a promising technology to meet the stringent requirements of these applications \cite{Qu2022}. By offloading the computation-intensive tasks to proximate MEC servers, the QoE of MDs can be significantly enhanced in a cost-effective and energy-efficient way \cite{li2023tapfinger}. However, due to the dependence on terrestrial infrastructures and environment, conventional terrestrial MEC servers are limited by the high cost of deployment, low adaptability to the network dynamic, and fixed service range. 

\par Recent years have seen a paradigm shift from terrestrial edge computing toward aerial-terrestrial edge computing, i.e., UAV-assisted MEC networks, by integrating UAVs with MEC. With high maneuverability, UAVs could be rapidly and flexibly deployed as aerial MEC servers to assist the terrestrial MEC servers in providing computation offloading services whenever and wherever needed. Moreover, the line-of-sight (LoS) link of UAVs can improve the communication reliability and network capacity of the terrestrial MEC networks.

\par Despite the aforementioned benefits, designing an efficient scheme of computation offloading in UAV-assisted MEC systems is facing several unprecedented challenges. \textit{\textbf{i) Demand-Supply Contradiction for Resource Allocation.}} Compared to the cloud, MEC servers have limited computing capabilities, particularly for aerial MEC servers with constrained carrying capacity. However, the computation tasks of MDs are often computation-hungry and latency-sensitive. This demand-supply contradiction between the limited computing resources of MEC servers and the stringent requirement of MDs poses a challenge for efficient computing resource allocation. \textit{\textbf{ii) Demand-Supply Heterogeneity for Computation Offloading.}} Different computation tasks of MDs have diverse requirements on computing resources, while different MEC servers possess varying computing capabilities. This demand-supply heterogeneity between the computation tasks of MDs and MEC servers could incur resource under-utilization among MEC servers, which brings difficulties in designing efficient computation offloading methods to ensure satisfied QoE for MDs and high resource utilization among MEC servers. \textit{\textbf{iii) Energy-Efficient and Real-Time Trajectory Control.}} The mobility of MDs and random generation of computation tasks lead to spatiotemporal dynamics in the offloading requirements, which necessitates real-time trajectory control. However, the intrinsic limited onboard energy of UAVs restricts the service time, thus posing challenges for energy-efficient and real-time UAV trajectory control. \textit{\textbf{iv) Different Time-Scale Dynamics.}} The dynamic characteristics of the UAV-assisted MEC network, such as the dynamic of the channel, random arrival of tasks, and mobility of MDs, vary across different timescales. Accordingly, integrating these features into a joint optimization framework for computing resource allocation, computation offloading, and trajectory control is significant but leads complexity to the algorithm design.

\par This work presents a two-timescale computing resource allocation, task offloading, and UAV trajectory control approach in UAV-assisted MEC. The main contributions are summarized as follows.

\begin{itemize}
	\item \textit{\textbf{System Architecture.}} We employ a hierarchical architecture for the UAV-assisted MEC network that consists of an MD layer, a terrestrial edge layer, an aerial edge layer, and a control layer. Under the coordination of the software-defined network (SDN) controller, the two-timescale decisions are made to deal with the demand-supply contradiction between MDs and MEC servers, demand-supply heterogeneity between MDs and MEC servers, and the spatio-temporal dynamics of computation tasks.
	
	\item \textit{\textbf{Problem Formulation.}} We formulate a joint computing resource allocation, computation offloading, and trajectory control problem to maximize the system utility that is theoretically modeled by synthesizing the network dynamics between MDs and terrestrial/aerial edge links, MD mobility, MD QoE, and the energy consumption of terrestrial/aerial MEC servers. Moreover, the optimization problem is proved to be a non-convex and NP-hard mixed integer nonlinear programming (MINLP).
	
	\item \textit{\textbf{Algorithm Design.}} To solve the formulated problem, we propose a two-timescale joint computing resource allocation, computation offloading, and trajectory control (TJCCT) algorithm. Specifically, TJCCT consists of a price-incentive method for on-demand computing resource allocation, a matching mechanism-based method for computation offloading, and a convex optimization-based method for UAV trajectory control.
	
	\item \textbf{\textit{Performance Evaluation.}} The performance of TJCCT is verified through both theoretical analysis and simulation. First, the stability, optimality, and computation complexity of TJCCT are proved theoretically. Furthermore, simulation results demonstrate that TJCCT has better performances and scalability than the baseline algorithms.
\end{itemize}

\par The remaining of this work is organized as follows. Section \ref{sec_related work} reviews the related work. Section \ref{sec_model} presents the system model and problem formulation. Section \ref{sec_jointOffloading} elaborates on the proposed TJCCT. The theoretical analysis is given in Section \ref{sec_Analysis}. Section \ref{sec_simulation} shows the simulation results and discussions. Finally, this work is concluded in Section \ref{sec_conclusion}.


%
%
\section{Related work}
\label{sec_related work}

\par In this section, we review the related work on the edge computing architecture, joint computation offloading, computing resource allocation, and the optimization approach.

\subsection{MEC Architecture}

\par The MEC has been extensively studied to extend the computing capabilities of MDs through computation offloading. Numerous research efforts have focused on leveraging terrestrial MEC to offer low-latency offloading services in different scenarios. For example, Wang et al. \cite{Wang2020} presented a non-cooperative computation offloading approach in a single-BS vehicular MEC network. Furthermore, Xia et al. \cite{Xia2021} proposed a distributed approach of computation offloading and computing resource management in the energy harvesting-enabled MEC system. Jiang et al. \cite{jiang2022joint} proposed a joint offloading and resource allocation framework for the energy-constrained MEC system, aiming to guarantee the QoE of the users. Moreover, Ding et al. \cite{ding2022hybrid} explored minimizing the system energy consumption for the single-terrestrial server and multi-user MEC system. Tao et al. \cite{tao2023single} proposed a dynamic pricing-based computation offloading scheme for a single-cell and mult-user MEC system. Additionally, He et al. \cite{he2023computation} designed an
architecture that combines digital twin and terrestrial MEC to jointly optimize computation offloading and resource allocation. Besides, Zhou et al. \cite{zhou2023dag} focused on an MEC-enabled Internet of things (IoT) system with a BS-amounted MEC and multiple users, and studied the computation offloading for dependent tasks. Besides, However, these studies mainly rely on the terrestrial MEC servers which have fixed coverage, especially in the dense areas of the city where the link between users and MEC servers could experience severe blockage and poor signal strength.

\par To offer flexible aerial computing services for users, recent studies have expanded the scope of terrestrial MEC to UAV-assisted MEC. For example, Ding et al. \cite{Ding2023} proposed a UAV-assisted MEC secure communication system where the UAVs assist the ground users in computing the offloading task, and the ground jammer generates jamming signals to prevent the UAV eavesdroppers. Lin et al. \cite{lin2023pddqnlp} explored maximizing the energy efficiency and offloading fairness in a single-UAV-assisted MEC system. Yang et al. \cite{9814972} focused on a UAV-enabled MEC system where the MEC server provides offloading service for energy harvesting devices. In \cite{Wang2022a}, the authors proposed a UAV-assisted architecture for post-disaster rescue by leveraging the computing capability of vehicles.  Moreover, Li et al. \cite{li2023robust} considered a multi-UAV-assisted MEC system and designed a robust approach of computation offloading and trajectory optimization. Besides, Coletta et al. \cite{Coletta2023} proposed a novel application-aware optimization framework where the UAVs are deployed as MEC servers for image analyses. However, most of these studies consider relatively simple scenarios where only one UAV-assisted MEC server is deployed, or users are assumed to be stationary. This may not be suitable for complex situations with heterogeneous MEC servers and MDs.

\par In summary, the existing MEC architecture is inadequate to adapt to the heterogeneous and dynamic scenario of the UAV-assisted MEC system. To this end, we propose a hierarchical architecture to address the limitations of the existing works.

\subsection{Computation Offloading, Resource Allocation, and Trajectory Control}

\par Researchers have extensively explored various aspects of UAV-assisted MEC systems, with a primary focus on computation offloading, resource allocation, and UAV trajectory control \cite{Zeng2019Access, Wu2021Acom}. 

\par To address the limited computing capability of the UAV-assisted MEC system, several studies focused on joint computation offloading and computing resource allocation to optimize the performances such as delay, energy consumption, or power. For example, Tun et al. \cite{Tun2022Colla} focused on a latency minimization problem in the collaborative UAV-assisted MEC system by jointly optimizing the strategies of offloading and resource allocation. Furthermore, Guo et al. \cite{Guo2023Multi} studied the problem of joint task scheduling and computing resource allocation to minimize the task processing delay under the constraints of data dependency and UAV energy. Diamanti et al. \cite{diamanti2023delay} focused on minimizing the sum of users’ maximum experienced delay by jointly optimizing the strategies of computation offloading and computing resource allocation. Ei et al. \cite{Ei2022Energy} presented a joint computation offloading and resource allocation problem for UAV-assisted edge computing to minimize the energy
consumption of the system. Nie et al. \cite{nie2021semi}
addressed the UAV power minimization problem by
jointly optimizing the computation offloading and resource allocation of the UAV-assisted MEC system. However, in most of these studies, UAVs are either fixed or follow a predetermined trajectory, which may not be suitable for scenarios with random-distributed and mobile MDs.

\par To harness the full potential of flexible offloading services, considerable research attention has been dedicated to joint optimization of computation offloading and UAV trajectory control. The objectives of these works include delay minimization \cite{Han2023joint}, energy consumption minimization \cite{liu2022resource, Wang2022c,chen2022Distributed}, secure calculation maximization \cite{Chen2023}, etc. For instance, Han et al. \cite{Han2023joint} formulated an optimization problem to minimize the average task
delay via jointly optimizing user association and UAV deployment. Furthermore, Liu et al. \cite{liu2022resource} presented a system energy consumption minimization problem by jointly optimizing the UAV trajectory for the
multiple input single output UAV-assisted MEC network. Wang et al. \cite{Wang2022c} aimed to minimize the energy consumption of the UAV via joint region partitioning and UAV trajectory scheduling. Chen et al. \cite{chen2022Distributed} jointly optimized the computation offloading of IoT nodes and trajectory planning of multiple UAVs to maximize the total energy efficiency of the system. Besides, the authors in \cite{Chen2023} focused on the maximum-minimum average secrecy capacity for UAV-assisted MEC systems by jointly optimizing the trajectory and computation offloading.

\par The limitations of these previous studies are summarized as follows. First, these works primarily concentrated on optimizing a single performance of the system, such as latency and energy. Furthermore, most of these studies formulated optimization problems from the perspective of the users or MEC servers while neglecting the heterogeneous requirements of the two sides. However, focusing solely on optimizing one aspect of the performance metric from the perspective of either users or MEC servers, without a holistic consideration of the diverse requirements could lead to imbalanced system performance and poor user experience. To address these limitations, we formulate a problem of joint computation offloading, resource allocation, and trajectory control problem to maximize the system utility, which incorporates delay and energy consumption of both MDs and MEC servers, as well as the dynamics of the network.

\subsection{Optimization Approach}

\par To tackle the intricate optimization problem of computation offloading, resource allocation, and trajectory control for UAV-assisted MEC, researchers are devoted to effective algorithm design by employing methodologies such as heuristic algorithms \cite{Laboni2024Hyper}, swarm intelligent algorithms \cite{Goudarzi2023UAV,tian2023service}, game theory\cite{zhou2022uav, Ning2023}, and reinforcement learning (RL)\cite{seid2021collaborative,song2023evo}. For example, Laboni et al. \cite{Laboni2024Hyper} proposed a hyper-heuristic algorithm for resource allocation in the MEC network to jointly optimize the latency, computational, and network load. Furthermore, Goudarzi et al. \cite{Goudarzi2023UAV} utilized a cooperative evolutionary method to solve the joint optimization of computation offloading and computing resource allocation. Tian et al. \cite{tian2023service} developed a genetic algorithm (GA)-based algorithm for task offloading and UAV scheduling in the UAV-assisted MEC system. Moreover, Zhou et al. \cite{zhou2022uav} studied the computation offloading in the UAV-assisted MEC network by employing a Stackelberg game to model the interaction between the MDs and the MEC server. Ning et al. \cite{Ning2023} proposed a stochastic game-based approach for multi-user computation offloading and MEC server deployment. In addition, Seid et al. \cite{seid2021collaborative} proposed a deep reinforcement learning (DRL)-based algorithm for collaborative computation offloading and resource allocation. Song et al. \cite{song2023evo} integrated the evolutionary algorithm with the RL algorithm for computation offloading and trajectory control in UAV-assisted MEC networks.

\par However, heuristic algorithms generally lack guaranteed optimality and can be sensitive to initial conditions. Furthermore, swarm intelligence algorithms and Stochastic games often require a substantial number of iterations to converge to a near-optimal solution, resulting in high computation complexity. Moreover, the Stackelberg games are well-suited for single-server situations but may not be appropriate for scenarios involving multiple MEC servers.  Moreover, although RL is powerful for training agents to make decisions, it 
requires a large number of interactions with the environment and significant computational resources, making it costly in the resource-constrained and heterogeneous MEC system. Therefore, we aim to design a low-complexity algorithm capable of facilitating real-time decision-making, mitigating heterogeneity between users and servers, and adapting to the varying timescale dynamics inherent in the UAV-assited MEC system.


%
%
\section{System Model and Problem Formulation}
\label{sec_model}

\par In this section, a UAV-assisted MEC architecture is first introduced, followed by the models of mobility, communication, computation, and energy consumption.

%
%

\subsection{System Model}
\label{sec_system_model}

\par This section presents the system overview, basic models, and problem formulation. 

\subsubsection{System Overview}
\label{sec:system_overview}

\par We consider a hierarchical UAV-assisted MEC system as shown in Fig. \ref{fig_systemModel}. In the spatial dimension, the hierarchical UAV-assisted MEC system comprises an MD layer, a terrestrial edge layer, an aerial edge layer, and a control layer. 

\begin{figure*}[!hbt] 
	\centering
        \setlength{\abovecaptionskip}{2pt}%
	\setlength{\belowcaptionskip}{2pt}%
	\includegraphics[width =6.6in]{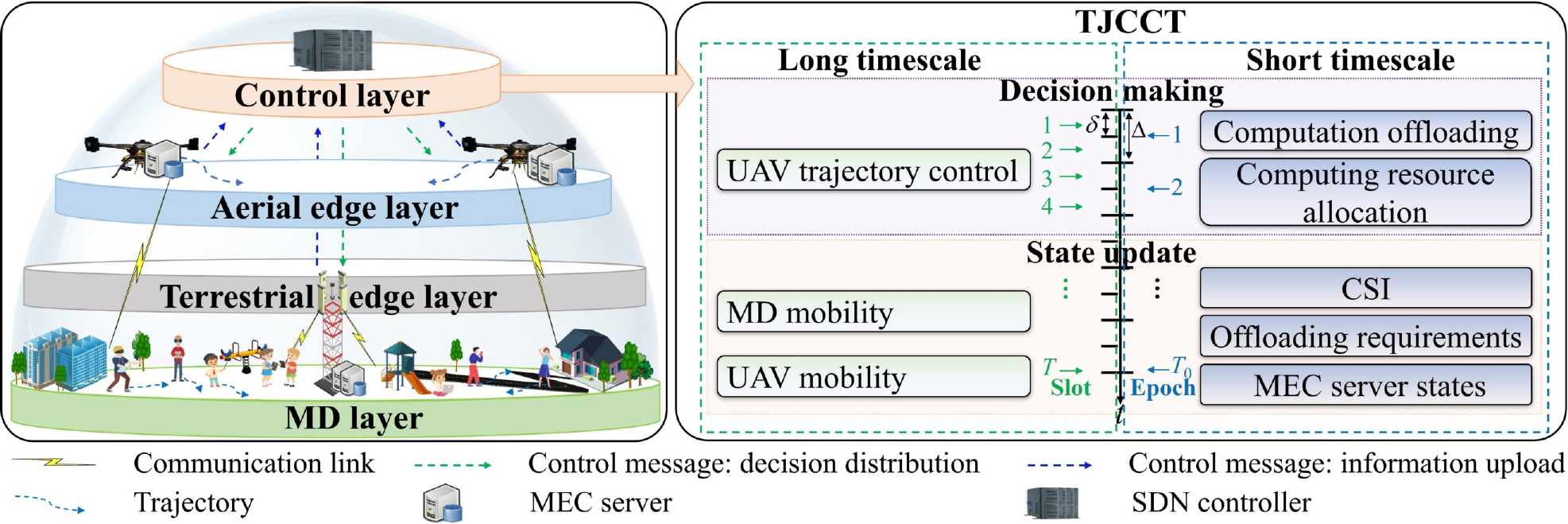}
	\caption{The architecture of computing resource allocation, computation offloading, and UAV trajectory control for UAV-assisted MEC system.}
	\label{fig_systemModel}
        \vspace{-1em}
\end{figure*}

    \par Specifically, \textit{at the MD layer}, a set of MDs $\mathcal{I}=\{1, \ldots, i, \ldots, I\}$ moving in the considered area periodically handle the tasks with diverse requirements. \textit{At the terrestrial edge layer},  a macro base station (MBS) $b$ equipped with terrestrial MEC server \footnote{Terrestrial MEC server and MBS will be used interchangeably.} provides edge computing services for the MDs within its service range. \textit{At the aerial edge layer}, the rotary-wing UAVs $\mathcal{U}=\{1, \ldots, U\}$ equipped with aerial MEC servers \footnote{Aerial MEC server and UAV will be used interchangeably.} are dispatched as aerial base stations to assist the MBS in providing supplementary computing services for MDs. \textit{At the control layer}, the regional SDN controller, on which our algorithm runs, coordinates the decisions regarding computation offloading for MDs, computing resource allocation for MEC servers, and trajectory control for UAVs based on the knowledge acquired from the terrestrial and aerial edge layers. Furthermore, we consider that the radio access links of the system are allocated orthogonal frequency bands \cite{ji2022trajectory}. Besides, both the terrestrial MEC server and aerial MEC servers are collectively referred to MEC server, which is indexed as $j \in \{ b \} \cup U$.

\par In the temporal dimension, we consider that the system operates in a two-timescale manner since the dynamic characteristics for channel state information (CSI), task arrival of MDs, and workload update of edge servers vary in a fine-grained timescale, while the mobility of MDs varies in a long timescale \cite{Li2021}. Specifically, the system time horizon is discretized into $T$ time slots $\mathcal{T}=\{1,\ldots,t,\ldots, T\}$ with equal slot duration $\delta$, which is consistent with the coherence block of the wireless channel \cite{Liao2021}. Furthermore, every $\Delta$ consecutive slots are combined into a time epoch indexed by  $t_0\in \mathcal{T}_0=\{1,\ldots, T_0\}$, where each time epoch is denoted as $\mathcal{T}(t_0) = \{(t_0 -1)\Delta + 1,\ldots,t_0\Delta\}$. Here, $\Delta$ is selected to be sufficiently small to guarantee that the UAV's location is approximately constant within each epoch. Therefore, in the short timescale, the CSI, offloading requirements of MDs, and states of MEC servers are captured and updated, and the decisions of task offloading and computing resource allocation are decided. In the long timescale, the mobility states of MDs are captured and updated, and the UAV trajectory planning is decided.

%
%
\subsubsection{Basic Models}
	
\par The basic models of the system are given as follows.

\par \textit{\textbf{(1) MD Mobility Model.}} The horizontal coordinate of each MD $i\in\mathcal{I}$ is denoted as $\mathbf{q}_{i,t}=[x_{i}^t,y_{i}^t]_{i\in\mathcal{I}}^{\text{T}}$. Moreover, we adopt the Gauss Markov model \cite{Batabyal2015} to capture the temporal-dependent randomness in the movement of MDs. Specifically, the velocity of MD $i$ at time epoch $t_0+1$ (i.e., time slot $(t_0+1)\Delta$) can be given as:
\begin{equation}
	\label{eq_MD_mobility}
		 \mathbf{v}_{i}^{(t_0+1)\Delta}=\alpha \mathbf{v}_{i}^{t_0\Delta}+\left(1-\alpha\right) \bar{\mathbf{v}}+\bar{\sigma}\sqrt{1-\alpha^2} \mathbf{w}_{i},
\end{equation}

\noindent where $ \mathbf{v}_i^{t_0\Delta}$ is the velocity vector at time epoch $t_0$ and $ \mathbf{w}_i$ is the uncorrelated random Gaussian process, i.e., $ \mathbf{w}_i\sim f^{\text{Gua}}\left(0,\bar{\sigma}^2\right)$. Besides, $0\leq \alpha\leq 1$, $\bar{\mathbf{v}}$, and $\bar{\sigma}$ denote the memory degree, asymptotic mean, and asymptotic standard deviation of velocity, respectively. Therefore, the location of each MD can be updated as:
\begin{equation}
	\label{eq_MD_position}
	\begin{aligned}
    \mathbf{q}_{i}^{(t_0+1)\Delta}=\mathbf{q}_{i}^{t_0\Delta} + \mathbf{v}_i^{t_0\Delta} \delta \Delta, \forall i\in \mathcal{I}, \  t_0 \in \mathcal{T}_0.
	\end{aligned}
\end{equation}

\par \textit{\textbf{(2) UAV Mobility Model.}} Similar to most existing studies, e.g.,  \cite{Wang2022a}, we consider that each UAV $j\in \mathcal{U}$ flies at a fixed altitude of $H$ with the instantaneous horizontal coordinate of $\mathbf{q}_{j}^{t}=[x_{j}^t,y_{j}^t]_{j\in\mathcal{U}}^T$. Therefore, the location of each UAV can be updated as follows:
\begin{equation}
	\label{eq_UAV_position}
	\begin{aligned}
\mathbf{q}_{j}^{(t_0+1)\Delta}=\mathbf{q}_{j}^{t_0\Delta}+\mathbf{v}_{u}^{t_0\Delta}\delta\Delta, \forall j\in \mathcal{U},\ t_0 \in \mathcal{T}_0,
	\end{aligned}
\end{equation}

\noindent where $\mathbf{v}_{j}^{t_0\Delta}$ denotes the velocity of UAV $j$ at time epoch $t_0$. Furthermore, the position of each UAV should satisfy several practical constraints as follows: 
\begin{subequations}
	\label{eq_UAV_mob}
	\begin{alignat}{2}
		& 0\leq x_{j}^t\leq x^{\max}, \ \forall j\in\mathcal{U}, \ t \in \mathcal{T}, \label{eq_UAV_mob_x}\\
		& 0\leq y_j^t\leq y^{\max}, \ \forall j\in\mathcal{U}, \ t \in \mathcal{T}, \label{eq_UAV_mob_y}\\
		&\textbf{q}_{j}^{1}=\textbf{q}_j^{I},\ \textbf{q}_{j}^{T_0\Delta}=\textbf{q}_j^{F}, \ \forall j\in\mathcal{U},\label{eq_UAV_mob_ini_fin}\\
		&||\textbf{q}_{j}^{(t_0+1)\Delta}-\textbf{q}_{j}^{t_0\Delta}||\leq v_{\text{U}}^{\max}\delta\Delta,\ \forall j\in\mathcal{U}, \ t_0 \in \mathcal{T}_0, \label{eq_UAV_mob_dis_each}\\
		&||\textbf{q}_{j}^{F}-\textbf{q}_{j}^{t_0\Delta}||\leq v_{\text{U}}^{\max}  (T_0-t_0)\delta \Delta,\ \forall j\in\mathcal{U}, \ t_0 \in \mathcal{T}_0, \label{eq_UAV_mob_dis_each1}\\
	&||\textbf{q}_{j}^{t_0\Delta}-\textbf{q}_{j^{\prime}}^{t_0\Delta}||\geq d_{\text{U}}^{\text{safe}}, \ \forall j, j^{\prime}\in\mathcal{U}, j\neq j^{\prime}, \ t_0 \in \mathcal{T}_0, \label{eq_UAV_safe}
	\end{alignat}
\end{subequations}

\noindent where $v_U^{\max}$ is the maximum velocity of UAV. Furthermore, Constraints ~\eqref{eq_UAV_mob_x} and \eqref{eq_UAV_mob_y} guarantee that each UAV cannot fly beyond the boundary of the considered area. Moreover, the initial and final positions of UAV $j$ are predetermined by Constraints $ \textbf{q}_j^{I} $ and $ \textbf{q}_j^{F}$, respectively, as given in Constraint \eqref{eq_UAV_mob_ini_fin} \cite{Ji2020}. In addition, Constraints \eqref{eq_UAV_mob_dis_each} and \eqref{eq_UAV_mob_dis_each1} indicate that the flight distance of a UAV is constrained by the maximum velocity. Finally, each UAV should keep the minimum safe distance of $d_{\text{U}}^{\text{safe}}$ with the other UAVs to avoid collision, as given in Constraint \eqref{eq_UAV_safe}.

\par\textit{\textbf{(3) MD Model.}} Each MD $i$ is characterized by the tuple $<f_i^{\max}, n_i^{\text{core}}, \tau_i^t, \zeta_{i}^t, \mathcal{K}_{i}^{t}>$, where $f_i^{\max}$ represents the computation capability of MD $i$ (in cycles/s), $n_i^{\text{core}}$ denotes the number of CPU core of MD $i$, $\tau_i^t$ denotes the energy constraint of MD $i$, $\zeta_{i}^t\in\{0,1\}$ is a binary variable that indicates whether a task is generated by MD $i$ during time slot $t$, and $ \mathcal{K}_{i}^{t}$ denotes the task generated by MD $i$ in time slot $t$. Specifically, each MD could generate one task ($\zeta_{i}^t=1$) or not ($\zeta_{i}^t=0$) in time slot $t$. Furthermore, due to the resource limitation, we assume that each MD is equipped with a single CPU core, i.e., $n_i^{\text{core}}=1$ \cite{Dai2019a}. 

\par\textit{\textbf{(4) Server Model.}} Similar to \cite{Dai2019a}, we consider that each MEC server $j\in\{b, \mathcal{U}\}$ are equipped with multi-core CPUs to enable parallel processing of multiple tasks. Consequently, each MEC server is characterized by the tuple $<n_j^{\text{core}}, f_j^{\max}, E_j^{\max}>$,  where $n_j^{\text{core}}$ denotes the number of CPU cores, and each of these CPU cores assumed to have homogeneous computational resources of $f_j^{\max}$ (cycles/s).

\par\textit{\textbf{(5) Computation Task Model.}} MDs have heterogeneous requirements on diverse computation tasks due to distinct task characteristics. Specifically, the task generated by MD $i$ in time slot $t$ is characterized by $ \mathcal{K}_{i}^{t}=\left\langle l_{i}^{t}, \mu_{i}^t, \tau_{i}^t\right\rangle$, where $l_{i}^{t}$ is the task size, $\mu_{i}^t$ is the computation intensity (in cycles/bit), and $\tau_{i}^t$ is the deadline of the task.

%
%
\subsection{Communication Model}
\label{sec_communicationModel}

\par By adopting the widely used orthogonal frequency division multiple access (OFDMA) \cite{ji2022trajectory}, the instantaneous uplink data rate between MD $i\in\mathcal{I}$ and MEC server $j \in \{ b \} \cup U$ can be given as:
\begin{equation}
	\label{eq_dataRate}
	r_{i,j}^t=B_{i,j} \log_2\left(1+\frac{P_{i}^{t} g_{i,j}^t}{N_0}\right),
\end{equation}

\noindent where $B_{i,j}$ is the subchannel bandwidth between MD $i$ and MEC server $j$, $P_{i}^{t}$ is the transmit power of MD $i$ in time slot $t$, $N_0$ is the noise power, and $g_{i,j}^t$ is the instantaneous channel power gain between MD $i$ and MEC server $j$.

\par Due to the complex nature of the communication environment in UAV-assisted MEC networks, such as the movement of MDs and UAVs, and the occasional blockages caused by obstacles, the channel power gain between MD $i$ and MEC server $j$ is calculated by incorporating the commonly used probabilistic LoS channel with the large-scale and small-scale fadings as \cite{peng2022directional,ji2022trajectory}:
\begin{equation}
\label{eq_channelPowerGain}
	g_{i,j}^t={\mathbb{P}_{i,j}^{t}} g_{i,j}^{t,\text{L}} +(1-\mathbb{P}_{i,j}^{t})  g_{i,j}^{t,\text{NL}},
\end{equation} 

\noindent where $\mathbb{P}_{i,j}^{t}$ denotes the probability of LoS transmission, $g_{i,j}^{t,x}$ denotes the channel power gain between MD $i$ and MEC server $j$, and $ x\in\{\mathrm{L},\mathrm{N}\}$ represents LoS or non-line of sight (NLoS) links. Moreover, the details of $\mathbb{P}_{i,j}^t$ and $g_{i,j}^{t,x}$ are presented as follows.

\subsubsection{LoS Probability}

\par For the MD-MBS link, according to the 3GPP standard \cite{3GPP3D2015}, the probability of LoS transmission can be given as follows:
\begin{equation}
\label{eq_LoS_probability_MBS}
    \mathbb{P}_{i,j}^t=\min \left(\frac{d_1}{d_{i,j}^t}, 1\right)\left(1-e^{-\frac{d_{i,j}^t}{d_2}}\right)+e^{-\frac{d_{i,j}^t}{d_2}}, \ j=b,
\end{equation}

\noindent where $d_{i,j}^t$ is the instantaneous distance between MD $i$ and the MBS $j$. Besides, $d_1$ and $d_2$ are the parameters to fit the specific scenarios. Furthermore, for the MD-UAV link, an extensively employed LoS probability is calculated as \cite{AlHourani2014}:
\begin{equation}
\label{eq_LoS_probability}
    \mathbb{P}_{i,j}^t=\frac{1}{1 + p_1  \mathrm{e}^ {\left(-p_2\left(\frac{180}{\pi}\arcsin\left(\frac{H}{d_{i,u}^{t}}\right)-p_1\right)\right)}},\ j\in \mathcal{U},
\end{equation}

\noindent where $d_{i,j}^t$ is the horizontal distance between MD $i$ and UAV $j\in \mathcal{U}$. Moreover, $p_1$ and $p_2$ denote the environment-dependent parameters.

\subsubsection {Channel Power Gain}

\par According to \cite{3GPPTR389012020,zhang2018dense, peng2022directional}, the channel power gain between MD $i$ and MEC server $j$ in time slot $t$ can be uniformly given as $g_{i,j}^{t,x}=|h_{i,j}^{t,x}|^2 (L_{i,j}^{t,x})^{-1}$, where $h_{i,j}^{t,x}$ and $L_{i,j}^{t,x}$ denotes the parameters of \textit{small-scale fading} and \textit{large-scale fading}, respectively, which are presented in detail as follows.

\par First, the small-scale fading for both MD-MBS and MD-UAV communications in time slot $t$ is modeled as a parametric-scalable and good-fitting generalized fading, i.e., Nakagami-$m$ fading \cite{Boumaalif2022}, which is given as:
\begin{equation}
\label{eq_LoS_probability}
    \begin{aligned}
        h_{i,j}^{t,x}&\sim f^{\text{Nak}}\left(h_{i,j}^{t,x},m_{y}^{x}\right) \\
        &= \frac{2{\left(m_y^x\right)}^{m_y^x} \left({h_{i,j}^{t,x}}\right)^{2m_y^x-1} e^{\left(-\frac{m_y^x  \left({h_{i,j}^{t,x}}\right)^2}{\overline{p}}\right)}}{\Gamma\left(m_y^x\right)  \left(\overline{p}\right)^{m_{y}^{x}}}, \ j\in \{b,\mathcal{U}\},
    \end{aligned}
\end{equation}

\noindent where $\overline{p}$ is the average received power, $\Gamma(\cdot)$ is the Gamma function, and $m_{y}^{x}\in\{m_\mathrm{T}^\mathrm{L},m_\mathrm{T}^\mathrm{N},m_\mathrm{A}^\mathrm{L},m_\mathrm{A}^\mathrm{N}\}$ is the Nakagami-$m$ fading parameters for terrestrial/aerial LoS/NLoS channel. 

\par Furthermore, the large-scale fading for MD-MBS communication in time slot 
$t$ can be given as:
\begin{equation}
\label{eq_LoS_probability}
    L_{i,j}^{t,x}=\frac{\left(4\pi  d_0^{\mathrm{T}}  f_c\right)^2}{c^2} \left(\frac{d_{i,j}^t}{d_0^\mathrm{T}}\right)^{\beta_\mathrm{T}^x} \chi^x, \ j=b,
\end{equation}
 
\noindent where $d_0^{\mathrm{T}}$ is the reference distance for the communication between MD and terrestrial MEC server, $\beta_\mathrm{T}^x\in\{\beta_\mathrm{T}^\mathrm{L},\beta_\mathrm{T}^\mathrm{N}\}$ is the path loss exponent for LoS/NLoS channel between MD $i$ and the MBS $j$, and $\chi^x\in\{\chi^\mathrm{L},\chi^\mathrm{N}\}$ denotes the standard deviation of shadowing for LoS/NLoS transmission \cite{3GPPTR389012020}, which follows the zero-mean Gaussian distributed random variable, i.e., $\chi^x \sim f^{\text{Gua}}\left(0,\left(\sigma^{x}\right)^2\right)$. 

\par Besides, the large-scale fading for MD-UAV communication in time slot $t$ can be given as:
\begin{sequation}
\label{eq_large_MD_UAV}
     L_{i,j}^{t}=\frac{\left(4\pi  d_0^{A}  f_c\right)^2}{c^2\kappa} \left(\frac{d_{i,j}^t}{d_0^\mathrm{A}}\right)^{\beta_\mathrm{A}}, \ j\in \mathcal{U},
\end{sequation}

\noindent where $d_0^\mathrm{A}$ is the reference distance for the communication between MD $i$ and the UAV $j$, $\beta_\mathrm{A}$ is the path loss exponent of MD-UAV communication, and $\kappa$ is the additional attenuation factor due to the NLoS link.

%
%

\subsection{Computation Model}
\label{sec_DelayModel}

\par Each MD is capable of performing local computing and computation offloading simultaneously, as the communication circuit and computation unit are separate \cite{Yang2022}. Both local computing and computation offloading generally incur overheads in terms of delay and energy consumption, which are detailed in the following subsections.

\subsubsection{Local Computing Model}

\par The completing delay and energy consumption for local computing are given as follows.
 
\par \textbf{(1) Completion Delay.} When task $\mathcal{K}_{i}^{t}$ is executed locally by MD $i$, the task completion delay is mainly incurred by task computation, which can be given as follows:
 \begin{equation}
 	\label{eq_time_local}
 	D_{i,i}^t=\frac{\mu_{i}^t}{f_i^t},
 \end{equation}
\noindent where $f_i^t$ is the available computing resources of MD $i$ in time slot $t$.
 
\par \textbf{(2) Energy Consumption.} Correspondingly, the energy consumption of MD $i$ to execute task $\mathcal{K}_i^t$ locally is given as:
\begin{equation}
	\label{eq_energyLocalExe} E_{i,i}^t=\gamma_i(f_i^t)^{2}\mu_{i}^t,
\end{equation}

\noindent where $\gamma_i\geq0$ denotes the effective capacitance of MD $i$'s CPU that depends on the CPU chip architecture \cite{pan2021cost}.

\subsubsection{Edge Offloading Model}

\par The completing delay and energy consumption for edge offloading are given as follows.

\par \textbf{(1) Completion Delay.} When task $\mathcal{K}_{i}^{t}$ is offloaded to MEC server $j$ for remote processing, the task completion delay mainly consists of transmission delay and computation delay, i.e.,
\begin{equation}
	\label{eq_edge_delay}
	\begin{aligned}
		D_{i,j}^t=\frac{l_{i}^{t}}{r_{i,j}^t}+\frac{\mu_{i}^t}{f_{j,i}^t},
	\end{aligned}
\end{equation}

\noindent where $l_{i}^{t}/r_{i,j}^t$ represents the transmission delay that the task is uploaded from MD $i$ to MEC server $j$. Furthermore, $\mu_{i}^t/f_{j,i}^t$ represents the computation delay at MEC server $j$, which depends on the computing resources $f_{j,i}^t$ allocated by MEC server $j$ in time slot $t$. 

\par \textbf{(2) Energy Consumption.} The remote execution of task $\mathcal{K}_i^t$ on MEC server $j$ could result in transmission energy consumption for the MD, computation energy consumption for the MEC server, and flight energy consumption for the UAV. For MD $i\in\mathcal{I}$, the energy consumed for task uploading is as follows:
\begin{equation}
	\label{eq_energyMDMEC}
	E_{i,j}^t=\frac{P_{i}^{t} l_{i}^{t}}{r_{i,j}^t},
\end{equation}

\noindent where $P_{i}^{t}$ denotes the transmit power of MD $i$ in time slot $t$. 

\par For MEC server $j$, the energy consumption for task execution can be given as:
\begin{equation}
	\label{eq_energyMecComp} E_{j,i}^{t,\text{comp}}=\gamma_j(f_{j,i}^t)^{2}\mu_{i}^t,
\end{equation}

\noindent where $\gamma_j\geq0$ denotes the effective capacitance of terrestrial MEC server $j$'s CPU. 

\par For aerial MEC server $j$, energy consumption also occurs during flight. Specifically, the unit propulsion energy of the rotary-wing UAV in straight-and-level flight includes the components of blade profile power, induced power, and parasite power\cite{Zeng2019}, which can be given as:
\begin{sequation}
	\begin{aligned}
		\label{eq_UAV_flight}
		E_{j}^{p}&=\underbrace{\eta_1\left(1+\frac{3 \left(v_{j}^{t}\right)^2}{{v_j^{\text{tip}}}^2}\right)}_{\text {Blade profile power}}+\underbrace{\eta_2 \sqrt{\sqrt{\eta_3+\frac{\left(v_{j}^{t}\right)^4}{4}}-\frac{\left(v_{j}^{t}\right)^2}{2}}}_{\text {Induced power}}\underbrace{+\eta_4 \left(v_{j}^{t}\right)^3}_{\text {Parasite power}},
	\end{aligned}
\end{sequation}

\noindent where $v_j^{\text{tip}}$ denotes the tip speed of the rotor blade. Moreover, $\eta_1$, $\eta_2$, $\eta_3$, and $\eta_4$ are the constants that depend on the aerodynamic parameters of the UAV. Consequently, the energy consumption of UAV $j$ to provide computation service for task $\mathcal{K}_{i}^{t}$ can be concluded as:
\begin{subnumcases}{\label{eq_energyMecComp1}E_{j,i}^t=}
	$$ \gamma_j(f_{j,i}^t)^{2}\mu_{i}^t, \ j=b$$, \label{eq_energyMecComp1A}\\
	$$ \gamma_j(f_{j,i}^t)^{2}\mu_{i}^t + E_j^p \delta, \ j\in \mathcal{U}$$.\label{eq_energyMecComp1B}
\end{subnumcases}

\par Note that the delay and energy consumption of the result
feedback are neglected since the result of a task is
generally much smaller than that of the input \cite{Wang2022b}. 

%
%

\subsection{Utility Model}
\label{sec_utilityModel}

\par In this section, we present the QoE of MDs, the revenue of MEC servers, and the utility of the system.
\vspace{-1 em}
\subsubsection{QoE of MDs}
\label{sec_MDUtility}

\par The QoE achieved by MD $i$ in time slot $t$ is calculated as the difference between the satisfaction degree from task completion and the costs associated with task offloading. Specifically, the costs consist of the energy consumption and the fees paid to the MEC server, which depend on the offloading decision. Moreover, the offloading decision is denoted as a binary variable $o_{i,n}^t\in\{0,1\}, n \in \mathcal{N} = \{0, b\} \cup \mathcal{U}$, which indicates that task $\mathcal{K}_i^t$ of MD $i$ is processed locally ($o_{i,0}^t = 1$) or offloaded to MEC server $j \in \{ b \} \cup U$ ($o_{i,j}^t=1$) in time slot $t$. Therefore, the QoE achieved by MD $i$ in time slot $t$ is calculated as follows:
\begin{tequation}
	\begin{aligned}
	\label{eq_MD_utility}
    		&U_{i,n}^t=w_i\underbrace{ \frac{\log\left(1+\tau_i^t-D_{i,n}^t\right)}{\log\left(1+\tau_i^t\right)}}_{\text{Satisfaction degree}}-(1-w_i) \\&\Bigg(\overbrace{\mathbb{I}_{\left(n=0\right)}\underbrace{\frac{E_{i,i}^t}{E_{i}^{\max}}}_{\text{Cost of energy}}}^{\text{Local computing}}+ \overbrace{\mathbb{I}_{\left(n=j\right)} \Big(\underbrace{\frac{E_{i,j}^t}{ E_{j}^{\max}}}_{\text{Cost of energy}}+\underbrace{\frac{f_{j,i}^t p_{j,i}^t }{G_i^{\max}}}_{\text{Cost of payment}}\Big)}^{\text{Edge offloading}}\Bigg), \\
		& \forall i \in \mathcal{I},\ 
  j\in \{ b \} \cup \mathcal{U}, \ n \in \mathcal{N}, 
	\end{aligned}
\end{tequation}

\noindent where the metrics of satisfaction degree and cost, which have different units, are normalized and then incorporated using the weight parameter $w_i$. Specifically, the normalized satisfaction degree, i.e., $\log\big(1+\tau_i^t-D_{i,j}^t\big)/ \log\left(1+\tau_i^t\right)$, reflects the satisfaction level of MD $i$ in completing task $\mathcal{K}_i^t$, which is commonly modeled as a logarithmic function and extensively used for evaluating the benefit of task processing in mobile computing domains \cite{Xia2021}. Additionally, the terms $E_{i,i}^t/E_i^{\max}$ and $E_{i,j}^t/E_i^{\max}$ denote the normalized energy consumption of local computing ($\mathbb{I}_{\left(n=0\right)}=1$) and edge offloading ($\mathbb{I}_{\left(n=j\right)}=1$), respectively, where $E_i^{\max}$ is the energy constraint of MD $i$. Besides, $p_{j,i}^t f_{j,i}^t/G_i^{\max}$ represents the normalized payment to MEC server $j$, where $f_{j,i}^t$ is amount of computing resource allocated by MEC server $j$ to MD $i$, $p_{j,i}^t$ is the unit price of computing resource paid by MD $i$, and $G_i^{\max}$ is the budget of MD $i$.

\subsubsection{Revenue of MEC Servers}
\label{sec_MECUtility}

\par  The utility gained by MEC server $j$ in time slot $t$ from executing task $\mathcal{K}_i^t$ is calculated as the difference between the reward received from MD $i$ and the cost of energy consumption, i.e.,
\begin{sequation}
	\label{eq_MEC_utility}
        U_{j,i}^t = w_j  \underbrace{\frac{f_{j,i}^t        p_{j,i}^t}{ f_{j}^{\max} p_j^{\max}}}_{\text{Reward from MD}}-(1-w_j)  \underbrace{\frac{E_{j,i}^t}{E_j^{\max}}}_{\text{Energy cost}}, \forall i \in \mathcal{I},\ 
        j\in \{ b \} \cup \mathcal{U}.
\end{sequation}

\noindent Similar to Eq. \eqref{eq_MD_utility}, the metrics of reward and cost are normalized and then incorporated using the weight parameter $w_j$. Specifically, $p_{j,i}^t  f_{j,i}^t/(f_{j}^{\max} p_j^{\max})$ denotes the normalized reward of MEC $j$ received by providing computational services to MD $i$, where $p_j^{\max}$ is the maximum price for the computing resource of MEC server $j$. Furthermore, $E_{j,i}^t/E_j^{\max}$ represents the normalized cost of energy consumption of MEC server $j$ with $E_j^{\max}$ being its maximum tolerable energy consumption.

\subsubsection{Utility of System}
\label{sec_social_welfare}

\par The utility of the system is defined to evaluate the overall system performance, which is calculated by summing the QoE of MDs and the revenue of MEC servers in each time slot. Specifically, combining Eqs. \eqref{eq_MD_utility} and \eqref{eq_MEC_utility}, the utility of system can be calculated as:
\begin{sequation}
	\label{eq_socialwelfare}
	\begin{aligned}
			U^t&=\underbrace{\sum_{i\in \mathcal{I}}\sum_{n\in \mathcal{N}}\zeta_{i}^t  o_{i,n}^t U_{i,n}^t}_{\text{The total utility of MDs}}+\underbrace{\sum_{j\in \{b,\mathcal{U}\}}\sum_{i\in \mathcal{I}}\zeta_{i}^t  o_{i,j}^t  U_{j,i}^t}_{\text{The total utility of MEC servers}}
			\\&=\sum_{i\in \mathcal{I}}\sum_{n\in \mathcal{N}}\zeta_{i}^t o_{i,n}^t \left(U_{i,n}^t+U_{n,i}^t\right), \ \forall i \in \mathcal{I},\ j\in \{ b \} \cup \mathcal{U}, n \in \mathcal{N}.
		\end{aligned}
\end{sequation}

%
%
\subsection{Problem Formulation}
\label{sec_problemFormulation}

\par The optimization problem is formulated to maximize the system utility over $T$ slots by jointly determining the computation offloading strategy $\mathbf{O}= \{o_{i,n}^t\}_{i\in \mathcal{I}, n \in \mathcal{N}, t \in \mathcal{T}}$, computing resource allocation and pricing strategy $\mathbf{F}= \{f_{j,i}^t, p_{j,i}^t\}_{i\in \mathcal{I}, j\in \{ b \} \cup \mathcal{U}, t \in \mathcal{T}}$, and UAV trajectory $\mathbf{Q}= \{\mathbf{q}_u^t\}_{j\in \mathcal{U},t \in \mathcal{T}}$. Therefore, the problem can be formulated as:
{
\begin{subequations}
	\label{eq_problem}
	\begin{alignat}{2}
		\mathbf{P}: \quad &\max_{\mathbf{O},\mathbf{F},\mathbf{Q}}  \sum_{t=1}^T U^t, \label{utility}\\
		\text{s.t.} \quad &  o_{i,n}^t\in\{0,1\}, \ \forall i\in \mathcal{I}, \ n\in \mathcal{N}, \ t\in \mathcal{T},\label{pro_c1}\\
		& \sum_{n\in \mathcal{N}}o_{i,n}^t\leq 1, \ \forall  i\in \mathcal{I},  \ t\in \mathcal{T}, \label{pro_c2}\\
		& \zeta_i^t=\{0,1\}, \ \forall i\in \mathcal{I},  \ t\in \mathcal{T}, \label{pro_c3}\\
		& o_{i,n}^tD_{i,n}^t\leq \tau_i^t, \ \forall i\in \mathcal{I}, \ j\in \{ b \} \cup \mathcal{U}, \ n\in \mathcal{N},  \ t\in \mathcal{T},  \label{pro_c4}\\
		&\sum_{i\in\mathcal{I}} o_{i,j}^t f_{j,i}^t \leq  f_j^{\max}, \ \forall j\in \{ b \} \cup \mathcal{U},  \ t\in \mathcal{T}, \label{pro_c6}\\
		&\sum_{i \in \mathcal{I}} o_{i,j}^t \leq n_{j}^{\text{core}}, \ \forall j\in \{ b \} \cup \mathcal{U},  \ t\in \mathcal{T}, \label{pro_c7}\\
		&o_{i,j}^t c_{j,i}^t f_{j,i}^t\leq G_i^{\max}, \ \forall i\in \mathcal{I}, \ j\in \{ b \} \cup \mathcal{U},  \ t\in \mathcal{T}, \label{pro_c10}\\
		& \eqref{eq_MD_mobility} \sim \eqref{eq_UAV_mob}. \label{pro_c11}		
	\end{alignat}
\end{subequations}
}

\noindent Constraints \eqref{pro_c1} and \eqref{pro_c2} indicate that the offloading strategy for each task is binary. Constraint \eqref{pro_c3} means that each MD generates at most one task in each time slot. Constraint \eqref{pro_c4} ensures that the delay of completing the task should not exceed the deadline. Constraints \eqref{pro_c6} and \eqref{pro_c7} constrain the computing resources and the number of CPU cores, respectively, for MEC server. Moreover, Constraint \eqref{pro_c10} guarantees that the price paid by each MD to the MEC server should not exceed its budget. In addition, Constraint \eqref{pro_c11} limits the mobility of MDs and UAVs.

\begin{theorem}
	\label{lemma_NP}
	Problem $\mathbf{P}$ is a non-convex and NP-hard MINLP.
\end{theorem}

\begin{proof}
        Problem $\mathbf{P}$ involves binary variables (i.e., computation offloading strategy $\mathbf{O}$), continuous variables (i.e., computing resource allocation $\mathbf{F}$ and trajectory control $\mathbf{Q}$), and non-linear objective function. Consequently, problem $\mathbf{P}$ is a MINLP, which is also non-convex and NP hard ~\cite{boyd2004convex}. 
\end{proof}

%
%
\section{Algorithm}
\label{sec_jointOffloading}

\par To solve problem $\mathbf{P}$, we propose TJCCT, which is comprised of two-timescale optimization methods. Specifically, in the short timescale, a price-incentive trading model is constructed based on the bargaining mechanism to facilitate the negotiation between the MDs and the MEC servers for the on-demand computing resource allocation and pricing. Furthermore, to deal with the heterogeneity between the computation tasks of MDs and MEC servers, a many-to-one matching is established to stimulate the end-edge collaboration for mutual-satisfactory computation offloading. In the long timescale, based on the optimal strategies of computing resource allocation and computation offloading, UAV trajectory is optimized by using convex optimization.

%
%
\subsection{Short Timescale: Computing Resource Allocation and Computation Offloading}
\label{time_slot}

\par In each time slot, the strategies of computing resource allocation and computation offloading are decided.

%
%
\subsubsection{Computing Resource Allocation}
\label{sec_bargain}

\par In this subsection, given that the task $\mathcal{K}_i^t$ generated by MD $i$ will be processed by MEC server $j$ at time $t$, the optimal computing resource allocation is presented. Specifically, the metric of the unit price for computing resources is introduced to capture the task processing costs of MEC servers that should be covered by MDs. This can be viewed as a market where MDs purchase computing resources from suitable MEC servers for task processing. However, the price of the computation resource affects the utilities of both MDs and MEC servers. Therefore, to optimize the utilities of MDs and MEC servers, an appropriate pricing strategy needs to be incorporated with the resource allocation strategy, taking into account the requirements of MDs and the computing capabilities of the MEC servers.

\par  According to the analysis above, we construct a price-incentive trading model for MDs and MEC servers. Specifically, the negotiation is modeled as a Rubinstein bargaining model \cite{rubinstein1982perfect} where the MD $i$ with an offloading request for its task $\mathcal{K}_i^t$ acts as the buyer and the target MEC server $j$ acts as a seller. The objective of the negotiation is to determine the optimal strategy of computing resource allocation $f_{j,i}^{t^*}$ with satisfied price incentive $p_{j,i}^{t^*}$ during time slot $\delta$, which are detailed below.

\textbf{(1) Optimal Computing Resource Allocation}. Given any price of the computing resource $p_{j,i}^{t}$, the optimal computing resource allocation $f_{j,i}^{t^*}$ can be determined as Theorem \ref{opt_resource}.

\begin{theorem}
\label{opt_resource}
    The optimal computing resources that MD $i$ expects to request from the target MEC server $j$ to offload task $\mathcal{K}_i^t$ is determined as follows:
    \begin{equation}
        f_{j,i}^{t^*}=\frac{2w_i G_i^{\max}}{\vartheta(p_{j,i}^t)-\log\left(1+\tau_i^t\right) p_{j,i}^t\left(1-w_i\right)}.
    \end{equation}
\end{theorem}

\begin{proof}
    \par For MEC server $j$, the second-order derivative of $U_{i,j}^t$ with respect to $f_{j,i}^t$ can be calculated as:
     \begin{equation}
     \label{eq_pr_resource1}
           \frac{\partial^2 U_{i,j}^t}{\partial^2 f_{j,i}^t}=\frac{\mu_{i,t} w_i\left(2f_{j,i}^t\left(l_{i,t}/r_{i,j}^t-\tau_{i,t}+\frac{\mu_{i,t}}{f_{j,i}^t}-1\right)-\mu_{i,t}\right)}{({f_{j,i}^t})^4  \log \left(1+\tau_{i,t}\right) \left(l_{i,t}/r_{i,j}^t-\tau_{i,t}+\frac{\mu_{i,t}}{f_{j,i}^t}-1\right)^2}.
    \end{equation}
   Since it can be inferred from Eq. \eqref{eq_MD_utility} that $\frac{l_{i,t}}{r_{i,j}^t}<1+\tau_{i,t}-\frac{\mu_{i,t}}{f_{j,i}^t}<1+\tau_{i,t}-\frac{\mu_{i,t}}{2f_{j,i}^t}$, we can derive the following inequations, i.e., $\frac{l_{i,t}}{r_{i,j}^t}<1+\tau_{i,t}+\frac{\mu_{i,t}}{2f_{j,i}^t}-\frac{\mu_{i,t}}{f_{j,i}^t}$ $\Rightarrow \frac{l_{i,t}}{r_{i,j}^t}-\tau_{i,t}+\frac{\mu_{i,t}}{f_{j,i}^t}-1<\frac{\mu_{i,t}}{2f_{j,i}^t}\Rightarrow$ $2f_{j,i}^t\left(\frac{l_{i,t}}{r_{i,j}^t}-\tau_{i,t}+\frac{\mu_{i,t}}{f_{j,i}^t}-1\right)-\mu_{i,t}<0$. Therefore, it is obvious that $U_{i,j}^t$ is concave with respect to $f_{j,i}^t$ since $\frac{\partial^2 U_{i,j}^t}{\partial^2 f_{j,i}^t}<0$, which implies that the maximal value of $U_{i,j}^t$ is existed. Consequently, the optimal amount of computational resources required by MD $i$ can be determined by applying the first-order optimality condition, as follows:
     \begin{equation}
     \label{eq_pr_resource2}
         \begin{aligned} f_{j,i}^{t^*}=
    		\begin{cases} \frac{-2w_i G_i^{\max}}            {\vartheta(p_{j,i}^t)+\log\left(1+\tau_i^t\right) p_{j,i}^t \left(1-w_i\right)}, & \\ \frac{2w_i G_i^{\max}}{\vartheta(p_{j,i}^t)-\log\left(1+\tau_i^t\right) p_{j,i}^t\left(1-w_i\right)}, &\\  
    		\end{cases}
        \end{aligned}
     \end{equation}

 \noindent where	
     \begin{sequation}
     \label{eq_pr_resource2}
        \begin{aligned}
         &\vartheta(p_{j,i}^t)=\sqrt{\frac{p_{j,i}^t 	\log(1+\tau_i^t)(1-w_i)}{\mu_i^t}} \\
         &\sqrt{p_{j,i}^t \mu_i^t\log(1+\tau_i^t) (1-w_i)+4G_i^{\max} w_i\left(1+\tau_i^t-l_{i}^{t}/r_{i,j}^t\right)}.
         \end{aligned}
     \end{sequation} 

\noindent From Eq. \eqref{eq_pr_resource2}, it is clear that $\frac{-2w_i G_i^{\max}}{\vartheta(p_{j,i}^t)+\log\left(1+\tau_i^t\right) p_{j,i}^t \left(1-w_i\right)}<0
$ since $\vartheta(p_{j,i}^t)>0$, $w_i>0$,  $p_{j,i}^t>0$, and $\log(1+\tau_i^t)>0$. Therefore, the optimal amount of computation resource desired by MD $i$ can be finally determined as $f_{j,i}^{t^*}=\frac{2w_i G_i^{\max}}{\vartheta-\log\left(1+\tau_{i,t}\right) p_{j,i}^t\left(1-w_i\right)}$.
\end{proof}

\par \textbf{(2) Satisfied Computing Resource Pricing.} Given any allocation of the computing resource $f_{j,i}^{t}$ for task $\mathcal{K}_i^t$, the satisfied price of computing resource $p_{j,i}^{t^*}$ can be determined by the following steps.

\par \textbf{First}, the upper bound and lower bound of the unit price of computing resource can be derived as Lemma \ref{lemma_priceRange}.

\begin{lemma}
\label{lemma_priceRange}
    \par To achieve a successful negotiation between MD $i$ and MEC server $j$ regarding the resource allocation and pricing of task $\mathcal{K}_i^t$, the unit price of the computing resource should be bounded by $\underline{p}_{j,i}^t \leq p_{j,i}^t\leq \overline{p}_{j,i}^t$, where
      \begin{equation}
      \label{eq_pro_price_low}
          \begin{aligned}
            \underline{p}_{j,i}^t =\frac{(1-w_j)E_{j,i}^tp_j^{\max}f_j^{\max}}{w_jE_j^{\max}f_{j,i}^t},
          \end{aligned}
      \end{equation}
      \begin{equation}
      \label{eq_pro_price_low}
          \begin{aligned}
            \overline{p}_{j,i}^t=\left(\frac{w_i\log\left(1+\tau_i^t-D_{i,j}^t\right)}{ (1-w_i)\log\left(1+\tau_i^t\right)}-\frac{P_i^t l_i^t}{r_{i,j}^t \tau_i^t}\right) \frac{G_i^{\max}}{f_{j,i}^t}.
          \end{aligned}
  \end{equation}
\end{lemma}

\begin{proof}
 	The price incentive should ensure that $U_{i,j}^t > 0$ and $U_{j,i}^t > 0$. Otherwise, the negotiation between MD $i$ and MEC server $j$ would be failed. Therefore, $\underline{p}_{j,i}^t$ and $\overline{p}_{j,i}^t$ can be obtained by substituting Eqs. \eqref{eq_MD_utility} and \eqref{eq_MEC_utility} into the above inequations.
\end{proof}

\par \textbf{Second}, we present the optimal negotiation between an MD and an MEC server by employing the Rubinstein bargaining model. According to Lemma \ref{lemma_priceRange}, the surplus of the computing resource price can be obtained as $\pi_{j,i}^t = \overline{p}_{j,i}^t - \underline{p}_{j,i}^t$. Therefore, the negotiation between MD $i$ and MEC server $j$ on the price of the computation resources can be modeled as the bargaining over the surplus $\pi_{j,i}^t$. Specifically, MD $i$ and MEC server $j$ take turns making offers about how to divide the surplus. Apparently, both the seller and buyer are subject to impatience and prefer a quick consensus on the division to a postponed trading delay because the utilities are discounted in the future. Consequently, the discount factors \cite{rubinstein1982perfect} of MD $i$ and MEC server $j$ are used to evaluate the patience with the negotiation delay, which are given as:
  \begin{subequations}
	\label{eq_discount_factor}
	\begin{alignat}{1}
            &\lambda_{i}^t=1-\frac{l_i^t}{r_{i,j}^t \tau_i^t}, \label{eq_discount_factor_A}\\
            &\lambda_{j}^t=1-\frac{\mu_i^t}{f_{j,i}^t\tau_i^t}. \label{eq_discount_factor_B}
	\end{alignat}
    \end{subequations}
 
 \noindent Eqs. \eqref{eq_discount_factor_A} and \eqref{eq_discount_factor_B} indicates that MD $i$ could be impatient with the long delay of task uploading, resulting in lower $\lambda_{i}^t$. Moreover, MEC server $j$ could be impatient with the long delay of task execution, resulting in the lower value of $\lambda_{j}^t$. In addition, trading parties could be more patient with the longer tolerable delay of the task $\tau_i^t$, leading to higher $\lambda_{i}^t$ and $\lambda_{j}^t$. 

\par \textbf{Third}, to derive the satisfied partition of the surplus $\pi_{j,i}^t$, we introduce the concepts of Nash equilibrium (NE) and subgame perfect Nash equilibrium (SPE) \cite{rubinstein1982perfect} as follows.

\begin{definition}
    \label{def_partition}
    NE. Any partition $\xi$ can be an NE outcome of the negotiation if it satisfies the following conditions: MD $i$ (or MEC server $j$) always proposes $\xi=\left(\xi_{i},\xi_{j}\right)$ and only accepts the offers $\xi^{\prime}$ where $\xi_{i}^{\prime}>\xi_{i}$ (or $\xi_{j}^{\prime}>\xi_{j}$).
\end{definition}

\begin{definition}
    \label{def_spe}
    SPE. The partition $\xi^*=\left(\xi_{i}^*,\xi_{j}^*\right)$ is an SPE if $\xi^*$ induces an NE each time a new offer is made and rejected.
\end{definition}

\noindent According to Definitions \ref{def_partition} and \ref{def_spe}, the satisfied partition of the surplus $\pi_{j,i}^t$ can be obtained by Lemma \ref{the_opt_partition}.

\begin{lemma}
	\label{the_opt_partition}
	\par The bargaining model has a unique SPE. In the period $T^b$ when MD $i$ makes a proposal, the satisfied partitions of the surplus $\pi_{j,i}^t$ for the computing resource price are given as: 
    \begin{subequations}
    \label{eq_opt_partition_MD}
        \begin{alignat}{1}
            & \xi_{i,i}^{t^*}=\lambda_{i}^t-\frac{\left(1-\lambda_{i}^t\right)\left(1-\left(\lambda_{i}^t\lambda_{j}^t\right)^{\lceil \frac{T^b}{2}\rceil}\right)}{1-\lambda_{i}^t\lambda_{j}^t}, \label{eq_opt_partition_MDA}\\
            &\xi_{j,i}^{t^*}=\frac{\left(1-\lambda_{i}^t\right)\left(2-\lambda_{i}^t\lambda_{j}^t-\left(\lambda_{i}^t\lambda_{j}^t\right)^{\lceil \frac{T^b}{2}\rceil}\right)}{1-\lambda_{i}^t\lambda_{j}^t}, \label{eq_opt_partition_MDB}
        \end{alignat}
    \end{subequations}
    
    \noindent In the period $T^b$ when MEC server $j$ makes a proposal, the satisfied partitions are given as:
      \begin{subequations}
    	\label{eq_opt_partition_MEC}
    	\begin{alignat}{1}
                &\xi_{i,j}^{t^*}=\frac{\left(1-\lambda_{j}^t\right)\left(1-\left(\lambda_{i}^t \lambda_{j}^t\right)^{\lceil{\frac{T^{b}}{2}\rceil}}\right)}{1-\lambda_{i}^t \lambda_{j}^t}, \label{eq_opt_partition_MECA}\\
            &\xi_{j,j}^{t^*}=\frac{\lambda_{j}^t\left(1-\lambda_{i}^t\right)-\left(1-\lambda_{j}^t\right)\left(\lambda_{i}^t \lambda_{j}^t\right)^{\lceil{\frac{T^{b}}{2}\rceil}}}{1-\lambda_{i}^t \lambda_{j}^t}. \label{eq_opt_partition_MECB}
    	\end{alignat}
        \end{subequations}
\end{lemma}

\begin{proof}
\label{proof_opt_partition}
	Suppose that MD $i$ makes a proposal in period $T^{b}\in \{t, t+\Delta t\}$ where the negotiation ends with a satisfactory outcome of partition. Therefore, it can be inferred that if MD $i$ proposes $\left(1,0\right)$ in period $T^b$, MEC server $j$ proposes $\left(\lambda_i^t,1-\lambda_i^t\right)$ in period $T^b-1$, etc. As a result, if $T^b$ is even, the partition obtained by MD $i$ in the first period can be calculated by:
      \begin{equation}
      \label{eq_opt_partition}
          \begin{aligned}
          {\xi_{i,j}^{t^*}}&=\lambda_{i}^t\left(1-\lambda_{j}^t\left(1-\lambda_i^t(\ldots)\right)\right)=\lambda_i^t-\lambda_i^t \lambda_{j}^t+(\lambda_i^t)^2 \lambda_{j}^t\\
          &-\ldots+\left(\lambda_i^t\right)^{\frac{T}{2}} \left(\lambda_{j}^t\right)^{\frac{T}{2}-1} =\lambda_i^t-\left(1-\lambda_i^t\right)\sum_{\mathbf{t}=0}^{\frac{T}{2}-1}\left(\lambda_i^t \lambda_{j}^t\right)^{\mathbf{t}}\\
          &=\lambda_i^t-\frac{\left(1-\lambda_i^t\right)\left(1-\left(\lambda_i^t\lambda_{j}^t\right)^{\frac{T^b}{2}}\right)}{1-\lambda_i^t\lambda_{j}^t}.
          \end{aligned}
      \end{equation}

\noindent For the cases that $T^b$ is odd and MEC server $j$ makes an offer, ${\xi_{i,j}^{t^*}}$ and ${\xi_{j,i}^{t^*}}$ can be derived in a similar way.
\end{proof}
\vspace{-0.5em}
\par \textbf{Fourth}, according to Lemmas \ref{lemma_priceRange} and \ref{the_opt_partition}, the optimal price of the computing resource can be derived by Theorem \ref{theo_pricing}.

\begin{theorem}
\label{theo_pricing}	
     \par The satisfied outcome for the price of the computation resource $p_{j,i}^{t^*}$ can be obtained as:
    \begin{subequations}
	\label{eq_opt_price}
	\begin{alignat}{1}
            & p_{j,i}^{t^*}=\overline{p}_{j,i}^t-\pi_{j,i}^t {\delta_{i,i}^{t^*}}, \label{eq_opt_price_A}\\
            &  p_{j,i}^{t^*}=\overline{p}_{j,i}^t-\pi_{j,i}^t  \delta_{i,j}^{t^*}, \label{eq_opt_price_B}
	\end{alignat}
    \end{subequations}

\noindent where Eq. \eqref{eq_opt_price_A} indicates the optimal price of the computing resource obtained in the period when MD $i$ makes an offer, and Eq. \eqref{eq_opt_price_B} indicates the optimal price of the computing resource obtained in the period when MEC server $j$ makes an offer.
\end{theorem}

\begin{proof}
\label{eq_optPrice}
    \par According to Theorem \ref{the_opt_partition}, the optimal partitions of the price difference obtained by MD $i$ and MEC server $j$ in the period when MD $i$ makes an offer can be calculated as:
    \begin{subequations}
	\label{eq_pro_partition}
	\begin{alignat}{1}
            &\varpi_{i,i}^{t^*}= \pi_{j,i}^t\xi_{i,i}^{t^*}, \label{eq_pro_partitionA}\\
            &\varpi_{j,i}^{t^*}=\pi_{j,i}^t\xi_{j,i}^{t^*}. \label{eq_pro_partitionB}
	\end{alignat}
    \end{subequations}
    
    \noindent Therefore, in the period when MD $i$ makes a proposal, MD $i$ decides its optimal bidding price as $p_{i,\text{bid}}^{t}=\overline{p}_{j,i}^t-\varpi_{i,i}^{t^*}=\overline{p}_{j,i}^t-\pi_{j,i}^t \xi_{i,i}^{t^*}$. Similarly, MEC server $j$ decides its optimal asking price as $p_{j,\text{ask}}^{t}=\underline{p}_{j,i}^t+\varpi_{j,i}^{t^*}=\underline{p}_{j,i}^t+\pi_{j,i}^t \xi_{j,i}^{t^*}$. It can be easily proved that $p_{i,\text{bid}}^{t}=p_{j,\text{ask}}^{t}$ since $\pi_{j,i}^{t}=\overline{p}_{j,i}^t-\underline{p}_{j,i}^t$ and $\xi_{j,i}^{t^*}=1-\xi_{i,i}^{t^*}$. Therefore, a deal on the price of the allocated resource can be obtained as $p_{j,i}^{t^*}=p_{i,\text{bid}}^{t}=p_{j,\text{ask}}^{t}$. Combining  $\xi_{i,i}^{t^*}$ into the formula of $p_{i,\text{bid}}^{t}$,  $p_{j,i}^{t^*}$ can be obtained as Eq. (\ref{eq_opt_price_A}). The proof of Eq. (\ref{eq_opt_price_B}) can be obtained similarly.
\end{proof}

\par \textbf{(3) Optimal Computing Resource Allocation with Price Incentive.} The optimal strategies of computing resource allocation with price incentive can be concluded from Theorems \ref{opt_resource} and \ref{theo_pricing}, as given in Corollary \ref{cor_deal}.

\begin{corollary}
\label{cor_deal}
\par A trading consensus can be reached on the amount and unit price of the allocated computing resources: 
\begin{sequation}
	\label{eq_optAllo}
	f_{j,i}^{t^*}=\frac{2w_i G_i^{\max}}{\vartheta(p_{j,i}^{t^*})-\log\left(1+\tau_i^t\right) p_{j,i}^{t^*}\left(1-w_i\right)},
\end{sequation}
\begin{subnumcases}{\label{eq_optAlloPrice}p_{j,i}^{t^*}=}
	$$\overline{p}_{j,i}^t-\Delta p_{j,i}^t \xi_{i,i}^{t^*}$$,  \label{eq_optAlloPriceA}\\
	$$\overline{p}_{j,i}^t-\Delta p_{j,i}^t  \xi_{i,j}^{t^*}$$.\label{eq_optAlloPriceB}
\end{subnumcases}
\end{corollary}

\par The trading contract between MD $i$ and MEC server $j$ is presented in Definition \ref{def_pricingRule}. According to Corollary \ref{cor_deal} and Definition \ref{def_pricingRule}, an alternative algorithm that iteratively optimizes the strategies of computing resource allocation and pricing is described in Algorithm \ref{algo_allo_price}. Specifically, the optimal strategy of computing resource allocation is initially set as the available resources of MEC server $j$ (line 2). Furthermore, in each iteration, MD $i$ and MEC server $j$ negotiate the satisfied price of the computing resource based on the trading contract (line 6). Then, they update the optimal strategy of computing resource allocation (line 7). The steps above are iterated until a consensus is reached. 

\begin{definition}
	\label{def_pricingRule}
	Trading contract. The trading amount and price are determined based on the following terms.
	\begin{itemize}
		\item If $U_{i,j}^t>0$, $U_{j,i}^t>0$, a consensus is reached on the trading amount and price based on Eqs. \eqref{eq_optAllo} and \eqref{eq_optAlloPrice}.
		
		\item If $U_{i,j}^t>0$, $U_{j,i}^t<0$,  MD $i$ makes an offer of the unit price of computing resources based on Eq. \eqref{eq_optAlloPriceA}.
		
		\item If $U_{i,j}^t<0$, $U_{j,i}^t>0$, MEC server $j$ makes an offer of the unit price of computing resources based on Eq. \eqref{eq_optAlloPriceB}.
		
		\item If $U_{i,j}^t<0$, $U_{j,i}^t<0$,  either MD $i$ or MEC server $j$ can make an offer.
	\end{itemize}
\end{definition}

\begin{algorithm}[]	
	\label{algo_allo_price}	
	\SetAlgoLined
	\KwIn{Task $\mathcal{K}_i^t$ of MD $i$ and MEC server $j$}
	\KwOut{The optimal resource allocation with satisfied price $(f_{j,i}^{t^*}, p_{j,i}^{t^*})$ in time slot $t$}
	\textbf{ Initialization:} 
	$U_{i,j}^t= 0$; $U_{j,i}^t= 0$; $\iota=0$; $\iota^{\max}= 100$\;
	Set the optimal resource allocation as $f_{j,i}^{t^*} = f_j^{\text{avl}}$\;
	\While{$\iota \leq \iota^{\max}$}
	{
		Update $p_{j,i}^{t^*}$ based on Eq. \eqref{eq_optAlloPrice}\;
		Calculate $U_{i,j}^t$, $U_{j,i}^t$ based on Eqs. \eqref{eq_MD_utility} and \eqref{eq_MEC_utility}\;
		Perform the trading contract based on Definition \ref{def_pricingRule}\;
		Update $f_{j,i}^{t^*}$ based on Eq. \eqref{eq_optAllo}\;
		$\iota = \iota +1$\;
	}
	\Return{$(f_{j,i}^{t^*}, p_{j,i}^{t^*})$}\;
	\caption{Computing Resource Allocation.}
\end{algorithm}	

%
%
\subsubsection{Computation Offloading}
\label{sec_matching}

\par Matching mechanism offers an efficient tool to construct the mutual-beneficial relationship between two sets of entities with heterogeneous preferences. This motivates us to construct the matching between the computation tasks of MDs and MEC servers to alleviate the demand-supply heterogeneity. By doing so, the MDs and MEC servers can achieve mutual-beneficial computation offloading results of satisfied QoE and high computing resource utilization. Denote the set of computation tasks that have not begun execution in time slot $t$ as $\mathcal{K}_{\text{req}}^t=\{\mathcal{K}_i^{\mathbf{t}}|i\in \mathcal{I}, \mathbf{t}\in \mathcal{T}\}$, where $\mathbf{t}$ is the generation time of the computation task. Then the offloading strategy for these computation tasks in each time slot is decided using a many-to-one matching mechanism, which is defined by Definition \ref{def_match}.

\begin{definition}
	\label{def_match}
	The current matching is defined as a triplet of $(\mathcal{M}^t,\mathcal{L}^t, \Pi^t)$: 
	\begin{itemize}[]
		\item $\mathcal{M}^t=\left(\mathcal{K}_{\text{req}}^{t}, \{ b \} \cup \mathcal{U}\right)$ consists of the tasks of MDs and the MEC servers.
		\item $\mathcal{L}^t=\left(\mathcal{L}_{\mathcal{K}_i^{\mathbf{t}}}^t,\mathcal{L}_{j}^t\right)$ consists of the preference lists of the tasks and MEC servers. Each task $\mathcal{K}_i^{\mathbf{t}}\in\mathcal{K}_{\text{req}}^{t}$ has a descending ordered preferences over the MEC servers, i.e., $\mathcal{L}_{\mathcal{K}_i^{\mathbf{t}}}^t=\{j|j\in \{ b \} \cup \mathcal{U}, j\succ_{\mathcal{K}_i^{\mathbf{t}}}{j^\prime}\}$, where $\succ_{\mathcal{K}_i^{\mathbf{t}}}$ is the preference of task $\mathcal{K}_i$ towards the  servers. Moreover, each MEC server $j$ has a descending ordered preference list over the tasks, i.e., $\mathcal{L}_j^t=\{\mathcal{K}_i^{\mathbf{t}}\in \mathcal{K}_{\text{req}}^{t}, \mathcal{K}_i^{\mathbf{t}} \succ_{j} {\mathcal{K}_i^{\mathbf{t}}}^{\prime}\}$. 
	
		\item $\Pi^t\subseteq \mathcal{K}_{\text{req}}^{t} \times \{ b \} \cup \mathcal{U}$ denotes the matching between the tasks and MEC servers. Each task $\mathcal{K}_i^{\mathbf{t}}\in \mathcal{K}_{\text{req}}^{t}$ can be matched with at most one MEC server, i.e., $\Pi_{\mathcal{K}_i^{\mathbf{t}}}^t\in \{ b \} \cup \mathcal{U}$, while each MEC server $j$ can be matched with multiple tasks, i.e., $\Pi_j^t\subseteq \mathcal{K}_{\text{req}}^{t}$.
	\end{itemize}
\end{definition}

\begin{algorithm}[]	
	\label{algo_matching}	
	\SetAlgoLined
	\KwIn{Tasks $\mathcal{K}_{\text{req}}^t=\{\mathcal{K}_i^{\mathbf{t}}|i\in \mathcal{I}, t\in \mathcal{T}\}$, and MEC servers $\{ b \} \cup \mathcal{U}$}
	\KwOut{The optimal matching list $\Pi^{t^*}$, offloading $\mathbf{O}^{t^*}$, and computing resource allocation $\mathbf{F}^{t^*}$}
	\textbf{ Initialization:} 
	$\mathcal{K}_{\text{rej}}^t=\mathcal{K}_{\text{req}}^{t}$, $\Pi^{t^*}=\emptyset$\;
	\For{$\mathcal{K}_i^{\mathbf{t}}\in \mathcal{K}_{\text{req}}^{t}$}
	{
		\For {$j\in \{b,\mathcal{U}\}$}
		{
			Call Algorithm \ref{algo_allo_price} to obtain $\big(f_{j,i}^{\mathbf{t}^*},p_{j,i}^{\mathbf{t}^*}\big)$\;
			Calculate $V_{\mathcal{K}_i^{\mathbf{t}},j}^t=U_{i,j}^{\mathbf{t}}$,$\ V_{j,\mathcal{K}_i^{\mathbf{t}}}=U_{j,i}^{\mathbf{t}}$\;
			$V_{\mathcal{K}_i^{\mathbf{t}},j}^t>V_{\mathcal{K}_i^{\mathbf{t}},j^{\prime}}^t \Leftrightarrow j\succ_{\mathcal{K}_i^{\mathbf{t}}} j^{\prime}, \ \mathcal{L}_{\mathcal{K}_i^{\mathbf{t}}}=\{j,j^{\prime}\}$\;
			$V_{j,\mathcal{K}}>V_{j,\mathcal{K}^{\prime}} \Leftrightarrow \mathcal{K}\succ_j \mathcal{K}^{\prime}, \ \mathcal{L}_j^t=\{\mathcal{K},\mathcal{K}^{\prime}\}$ \;
		}
	}
	\While{\rm{There exists} $\mathcal{K}_i^{\mathbf{t}}\in \mathcal{K}_{\text{rej}}^t$: $\mathcal{L}_{\mathcal{K}_i^{\mathbf{t}}}^t \neq \emptyset \ \&\& \ \mathcal{K}_i^{\mathbf{t}} \notin \mathcal{L}_j^t $}
	{
		\For{$\mathcal{K}_i^{\mathbf{t}}\in \mathcal{K}_{\text{rej}}^t$}
		{$\Pi_{\mathcal{K}_i^{\mathbf{t}}}^t=\Pi_{\mathcal{K}_i^{\mathbf{t}}}^t\cup j^{\prime}$, $\ j^{\prime} =  \mathcal{L}_{\mathcal{K}_i^{\mathbf{t}}}^t[1]$   \; 
			\If{$V_{\mathcal{K}_i^{\mathbf{t}},j^{\prime}}^t>0$}
			{
				$\Pi_{j^{\prime}}^t=\Pi_{j^{\prime}}^t\cup \mathcal{K}_i^{\mathbf{t}}$
			}
		}
		\For{$j\in \{b, \mathcal{U}\}$ \rm{that receives new requests}}
		{
			$|\Pi_j^t|\leq N_j \leq N_{j}^{\text{idl}}\ $,
			$\sum_{\mathcal{K}_i^{\mathbf{t}} \in \Phi(j)} f_{j,i}^{\mathbf{t}^*} \leq f_j^{t,\text{avl}}$\;
			$\Pi_j^t=\Pi_j^t\ \backslash \ \mathcal{D}_j^t$,$\ \mathcal{K}_{\text{rej}}^t=\mathcal{K}_{\text{rej}}^t\cup \mathcal{D}_j^t$ \;	
			{
				\For {$\mathcal{K}_i^{\mathbf{t}} \in \mathcal{D}_j^t$}
				{$\mathcal{L}_{\mathcal{K}_i^{\mathbf{t}}}^t=\mathcal{L}_{\mathcal{K}_i^{\mathbf{t}}}^t\ \backslash \ \{j\}$,$\ \Pi_{\mathcal{K}_i^{\mathbf{t}}}^t=\Pi_{\mathcal{K}_i^{\mathbf{t}}}^t\ \backslash \ \{j\}$; 
				}
			}	
		}	 	
	}	 	 	
	\Return {$\Pi^{t^*}=\Pi^t$, $\mathbf{O}^{t^*}=\{o_{i,j}^{t}| j=\Pi_{\mathcal{K}_i^{\mathbf{t}}}^t, \mathcal{K}_i^{\mathbf{t}}\in \mathcal{K}_{\text{req}}^{t}\}$, $\mathbf{F}^{t^*}=\{\big(f_{j,i}^{\mathbf{t}^*}, p_{j,i}^{\mathbf{t}^*}\big)|j=\Pi_{\mathcal{K}_i^{\mathbf{t}}}^t, \mathcal{K}_i^{\mathbf{t}}\in \mathcal{K}_{\text{req}}^{t}\}$}\;
	\caption{Computation Offloading.}
\end{algorithm}

\par The main steps of the matching process are presented in Algorithm \ref{algo_matching}, and the details are further described as follows.

\textit{\textbf{Preference List Construction.}} For each task $\mathcal{K}_i^{\mathbf{t}}\in \mathcal{K}_{\text{req}}^{t}$ and MEC server $j$, the preference lists are constructed based on the following steps: \textbf{\textit{i)}} predict the optimal computing resource allocation and pricing by calling Algorithm \ref{algo_allo_price} (line 4), \textbf{\textit{ii)}} calculate the preference value for each task on MEC servers and the preference value for each MEC server on tasks (line 5), \textbf{\textit{iii)}} construct the preference list for each task and MEC server by ranking the preference values in descending order (lines 6 and 7).

\par \textit{\textbf{Matching Construction.}} The matching process is implemented according to the following steps: \textbf{\textit{i)}} for each computation task $\mathcal{K}_i^{\mathbf{t}}\in \mathcal{K}_{\text{rej}}^t$, select the most preferred MEC server $j^{\prime}$ and add it to the matching list temporarily (line 10), \textbf{\textit{ii)}} if the computation task prefers MEC server $j^{\prime}$, add the computation task to the matching list of $j^{\prime}$ temporarily (lines 11 and 12), \textbf{\textit{iii)}} for each MEC server that receives new requests, update the matching list by remaining the top\textendash $N_j$ most preferred computation tasks and removing the less preferred computation tasks (lines 13 to 14) to guarantee that the current number of tasks and the allocated computing resources should not exceed the number of idle CPU cores $N_j^{t,\text{idl}}$ and the available computing resources $f_j^{t,\text{avl}}$ of the MEC server, respectively, \textbf{\textit{iv)}} add the deleted computation tasks into the rejected set, \textbf{\textit{v)}} update the preference list and matching list for the deleted computation tasks (lines 16 and 17). Note that the steps above are repeated until all computation tasks have been matched with an MEC server, or the unmatched computation tasks have been rejected by all MEC servers.

%
%
\subsection {Long Timescale: UAV Trajectory Control}
\label{sec_UAV Trajectory Control}

\par In each time epoch, the UAV trajectory is optimized by applying the convex approximation method. Specifically, the movement of UAVs in the next time epoch is optimized based on the optimized strategies of computing resource allocation $\mathbf{F}^{t^*}$ and computation offloading $\mathbf{O}^{t^*}$. Therefore, fixing $\mathbf{F}^{t^*}$ and $\mathbf{O}^{t^*}$, while eliminating the unrelated terms in the objective function and constraints, the problem of UAV trajectory optimization can be given as:
\begin{subequations}
	\label{eq_problem1}
	\begin{alignat}{2}
		\mathbf{P_t}: \ &\max_{\mathbf{Q}^{t^\prime}} U^t=\max_{\mathbf{Q}^{t^\prime}}\sum_{i\in \mathcal{I}}\sum_{j\in \mathcal{U}}\zeta_{i}^to_{i,j}^t \left(U_{i,j}^t+U_{j,i}^t\right) \\
		&\eqref{eq_UAV_mob_x} \sim \eqref{eq_UAV_safe},\notag
	\end{alignat}
\end{subequations}

\noindent where $\mathbf{Q}^{t^\prime}$ denotes the positions of UAVs in the next time epoch $t^\prime=\left(\lceil t/\Delta \rceil+1\right)\Delta$.

\begin{lemma}
	\label{lem_appro_p1}
	Problem $\mathbf{P_t}$ can be approximately converted into:
	{\small
	\begin{subequations}
		\label{eq_problem11}
		\begin{alignat}{2}
			\overline{\mathbf{P}}_\mathbf{t}: \ \max_{\mathbf{Q}^{t^\prime}}  \ &\sum_{i\in \mathcal{I}}\sum_{j\in \mathcal{U}}\zeta_{i}^t  o_{i,j}^t  \Bigg(\vartheta_0\log\left(1+\left(\tau_i^t-\frac{l_{i}^{t}}{\overline{r}_{i,j}^{t^\prime}}-\vartheta_1\right)\right)\Bigg.\notag\\
			&\Bigg.-\vartheta_2\frac{l_i^t}{\overline{r}_{i,j}^t }-\vartheta_3E_j^p\delta\Bigg)\\
			&\eqref{eq_UAV_mob_x} \sim \eqref{eq_UAV_mob_dis_each1}.\notag
		\end{alignat}
	\end{subequations}
    }
	\noindent where $\vartheta_0 = w_i/(1+\tau_i^t)$, $\vartheta_1=\mu_i^t/f_{j,i}^t$, $\vartheta_2=(1-w_i)P_i^t/E_i^{\max}$, $\vartheta_3 = (1-w_j)/E_j^{\max}$, and $\overline{r}_{i,j}^{t^\prime} = B_{i,j}^t\log_2 \left(1+\frac{P_{i}^t\bar{g}_{i,j}^t}{N_0 \left(\left\|{\mathbf{q}}_j^{t^\prime}-\mathbf{q}_i^t\right\|^2 + H^2\right)^{\beta_{A}/2}}\right)$.
\end{lemma}

\begin{proof}	
	\par Combining Eqs. \eqref{eq_edge_delay}, \eqref{eq_energyMecComp}, \eqref{eq_MD_utility}, and \eqref{eq_MEC_utility} with Eq. \eqref{eq_socialwelfare}, the objective function can be given as:
	\begin{sequation}
		\label{eq_problemConvert1}
		\begin{aligned}
			&\max_{\mathbf{Q}^{t^\prime}}  \ \sum_{i\in \mathcal{I}}\sum_{j\in \mathcal{U}}\zeta_{i}^t  o_{i,j}^t \Bigg(\frac{w_i\log \bigg(1+\Big(\tau_i^t-\frac{l_{i}^{t}}{r_{i,j}^{t^\prime}}-\frac{\mu_i^t}{f_{j,i}^t}\Big) \bigg)}{\log(1+\tau_i^t)}\Bigg.\\ 
			&\Bigg.- (1-w_i)\bigg(\frac{P_i^t l_i^t}{r_{i,j}^{t^\prime} \tau_i^t}+\frac{p_{j,i}^t f_{j,i}^t}{G_i^{\max}}\bigg)  +\frac{ w_j f_{j,i}^t p_{j,i}^t}{ f_{j}^{\max} p_j^{\max}} \Bigg.\\
			& \Bigg.-(1-w_j)\bigg( \frac{\gamma_j(f_{j,i}^t)^{2} \mu_{i}^t + E_{j}^{p} \delta}{E_j^{\max}}\bigg) \Bigg).
		\end{aligned}
	\end{sequation}

	\par Furthermore, by employing a homogeneous approximation for the LoS probability between MD $i\in \mathcal{I}$ and UAV $j\in\mathcal{U}$, i.e.,  $\mathbb{P}_{i,j}^t\sim \bar{\mathbb{P}}_{i,j}^t$ \cite{Zeng2019}, the channel power gain $g_{i,j}^t$ can be approximated as $g_{i,j}^t = \bar{g}_{i,j}^t  (d_{i,j}^t)^{-\beta_{A}}$, where $\bar{g}_{i,j}^t=\left(\bar{\mathbb{P}}_{i,j}^t |h_{i,j}^{t,\text{L}}|^2+\left(1-\bar{\mathbb{P}}_{i,j}^t\right) |h_{i,j}^{t,\text{NL}}|^2\kappa^{NL}\right)c^2 d_0^{\beta_{A}}/\left(4\pi  d_0^{A}  f_c\right)^2$. The value of $\bar{\mathbb{P}}_{i,j}^t$ can be set based on the most likely elevation angle or the average value \cite{Zeng2019}. Accordingly, the transmission rate can be approximated as $\overline{r}_{i,j}^{t^\prime} = B_{i,j}^t\log_2 \left(1+\frac{P_{i}^t\bar{g}_{i,j}^t}{N_0 \left(\left\|{\mathbf{q}}_j^{t^\prime}-\mathbf{q}_i^t\right\|^2 + H^2\right)^{\beta_{A}/2}}\right)$.
	
	\par Besides, by implementing the variable substitutions of $\vartheta_0 = w_i/(1+\tau_i^t)$, $\vartheta_1=\mu_i^t/f_{j,i}^t$, $\vartheta_2=(1-w_i)P_i^t/E_i^{\max}$,  $\vartheta_3 = (1-w_j)/E_j^{\max}$, and $\vartheta_4=\gamma_j(f_{j,i}^t)^{2}\mu_{i}^t$, substituting $\overline{r}_{i,j}^{t^\prime}$ into \eqref{eq_problemConvert1}, and removing the items which are irrelevant to $\mathbf{Q}^{t^\prime}$, the objective function can be obtained as Eq. \eqref{eq_problem11}.
\end{proof}

\par  However, problem $\overline{\mathbf{P}}_\mathbf{t}$ is a non-convex optimization problem due to the non-concavity of the objective function and the non-convexity of Constraint \eqref{eq_UAV_safe}. Therefore, it will be transformed into a convex problem by the following steps.

\par \textbf{First}, since the objective function of $\overline{\mathbf{P}}_\mathbf{t}$ is non-convex with respect to $\overline{r}_{i,j}^{t^\prime}$, the auxiliary variables $\tilde{r}_{i,j}^{t^\prime}$ is first introduced such that $\tilde{r}_{i,j}^{t^\prime} \leq \overline{r}_{i,j}^{t^\prime}$, where the RHS is lower bounded by a concave function as given in Lemma \ref{lemma_r_convex}.

\begin{lemma}
	\label{lemma_r_convex}
	Given the local point $\hat{\mathbf{q}}_j^{s}$ at the $s$-th iteration, $\overline{r}_{i,j}^t$ is lower bounded by:
	\begin{sequation}
		\label{eq_r_firstTaylor}
		\begin{aligned}
			\overline{r}_{i,j}^{t^\prime} &\geq  B_{i,j}\log_2\left(1+\frac{P_{i}^t\bar{g}_{i,j}^t}{N_0\left(H^2+\left\|\hat{\mathbf{q}}_j^{s}-\mathbf{q}_i^t\right\|^2\right)^{\beta/2}}\right)\\
			&-\frac{ B_{i,j}\beta\left(\left\|{\mathbf{q}}_j^{t^\prime}-\mathbf{q}_i^t\right\|^2-\left\|\hat{\mathbf{q}}_j^{s}-\mathbf{q}_i^t\right\|^2\right)}{2\ln2 \left(H^2+\left\|\hat{\mathbf{q}}_j^{s}-\mathbf{q}_i^t\right\|^2\right)}= \tilde{\tilde {r}}_{i,j}^{s}.
		\end{aligned}
	\end{sequation}
\end{lemma}

\begin{proof}
    \par We define a function $f(x) = a_1 \log_2(1+\frac{a_2}{\left(H^2+x\right)^{a_3/ 2}}), \forall j \in \mathcal{U}, t \in \mathcal{T}$, where $a_1$, $a_2$, $a_3$, and $x$ are positive. It can be concluded that $f(x)$ is convex with respect to $x$ by calculating the second-derivative of $f(x)$ as:
    \begin{equation}
    \label{eq_pro_traj_fx}
      \begin{aligned}
        \frac{\partial^2 f}{\partial x^2} =
        \frac{a_1a_2a_3(a_3+2) \left(H^2+x\right)^{a_3/2}}{4\ln2 \left(H^2+x\right)^{a_3+2}\left(\frac{a_2}{\left(H^2+x\right)^{a_3/2}}+1\right)^2}.
      \end{aligned}
    \end{equation}

    \noindent It can be deduced from Eq. \eqref{eq_pro_traj_fx} that $f(x)$ can be globally lower bounded by the first-order Taylor expansion with $x$ at any point $\hat{x}$ as:
    \begin{equation}
    \label{eq_pro_traj_fx_taylor}
      \begin{aligned}
        f(x)&\geq a_1 \log_2\left(1+a_2\left(H^2+\hat{x}\right)^{-a_3/2}\right)\\
        &-\frac{a_1a_2 a_3(x-\hat{x})}{2\ln2\left(\left(H^2+\hat{x}\right)^{(-a_3/2+1)}\left(a_2\left(H^2+\hat{x}\right)^{a_3/2}+1\right)\right)}.
      \end{aligned}
    \end{equation}

    \noindent Therefore, Eq.  \eqref{eq_r_firstTaylor} can be derived by setting $a_1=B_{i,j}^t$, $b=\beta$, $a_2 =\frac{P_{i}^t\bar{g}_{i,u}^t}{N_0}$, $x =\left\|{\mathbf{q}}_j^{t^\prime}-\mathbf{q}_i^{t}\right\|^2$, and $\hat{x} =\left\|{\hat{\mathbf{q}}}_j^{s}-\mathbf{q}_i^t\right\|^2$ and applying certain deductions.
	
\end{proof}
\vspace{-0.8em}
\par \textbf{Second}, for the non-convexity of the UAV propulsion energy $E_{j}^{p}$, we introduce an auxiliary variable $\phi$ such that
\begin{sequation}
	\label{eq_aux_vel}
	\phi^2 \geq \sqrt{\eta_3+\frac{\left(v_j^t\right)^4}{4}}-\frac{\left(v_j^t\right)^2}{2} \Longrightarrow \frac{\eta_3}{\left(\phi\right)^2} \leq \phi^2+\left(v_j^t\right)^2,
\end{sequation}

\noindent where $v_j^t=\frac{\left\|\mathbf{q}_j^{t^\prime}-\mathbf{q}_j^t\right\|}{\delta\Delta}$. For the convex RHS of \eqref{eq_aux_vel}, a global concave lower bound can be obtained at the local point $\phi^s$ by using the first-order Taylor expansion, i.e.,
\begin{sequation}
	\label{eq_aux_vel1}
	\begin{aligned}
		& \phi^2+\left(v_j^t\right)^2\geq\left(\phi^s\right)^2+2 \phi^s\left(\phi-\phi^s\right) \\
		& +\frac{\left\|\hat{\mathbf{q}}_j^{s}-\mathbf{q}_i^t\right\|}{\Delta^2}+\frac{2}{\Delta^2}\left(\hat{\mathbf{q}}_j^{s}-\mathbf{q}_i^t\right)^T\left(\hat{\mathbf{q}}_u^{t^\prime}-\mathbf{q}_i^t\right)=\tilde{\phi}^s\
		&
	\end{aligned}
\end{sequation}

\par  \textbf{Third}, to deal with the non-convex Constraint \eqref{eq_UAV_safe}, $\left|\left|\mathbf{q}_j^t-\textbf{q}_{j^{\prime}}^t\right|\right|^2$ can be lower bounded by applying the first-order Taylor expansion at any given point $\hat{\mathbf{q}}_j^t$ and $\hat{\textbf{q}}_{j^{\prime}}^t$, i.e., 
\begin{sequation}
	\label{eq_bound_safe_dis}
	\begin{aligned}
		&\left\|\mathbf{q}_j^t-\mathbf{q}_{j^{\prime}}^t\right\|^2 \geq  \left\|\hat{\mathbf{q}}_j^t- \hat{\textbf{q}}_{j^{\prime}}^t\right\|^2 + 2\left(\hat{\mathbf{q}}_j^t-\hat{\textbf{q}}_{j^{\prime}}^t\right)^T\\
		&\left(\mathbf{q}_j^t-\mathbf{q}_{j^{\prime}}^t-\left(\hat{\mathbf{q}}_j^t- \hat{\textbf{q}}_{j^{\prime}}^t\right)\right)\\ &=-\left\|\hat{\mathbf{q}}_j^t- \hat{\textbf{q}}_{j^{\prime}}^t\right\|^2+2\left(\hat{\mathbf{q}}_j^t-\hat{\textbf{q}}_{j^{\prime}}^t\right)^T\left(\mathbf{q}_j^t-\mathbf{q}_{j^{\prime}}^t\right)= \widetilde{d}_{j,j^{\prime}}^t.
	\end{aligned}
\end{sequation}

\par Based on Lemmas \ref{lem_appro_p1} and \ref{lemma_r_convex}, by introducing the auxiliary variables $\overline{r}_{i,j}^t$,  $\tilde{\tilde{r}}_{i,j}^s$, $\phi$, $\tilde{\phi}^s$, and $\widetilde{d}_{j,j^{\prime}}^t$, $\overline{\mathbf{P}}_\mathbf{t}$ can be transformed as:
{\small
\begin{subequations}
	\label{eq_problemSp1_fin}
	\begin{alignat}{2}
		\overline{\mathbf{P}}_\mathbf{t1}: &\max_{\mathbf{Q}{t^\prime}}  \ \sum_{i\in \mathcal{I}}\sum_{j\in \mathcal{U}}\zeta_{i}^t  o_{i,j}^t  \left(\vartheta_0\log\left(1+\tau_i^t-\frac{l_{i}^{t}}{\overline{r}_{i,j}^{t^\prime}}-\vartheta_1\right)\right.\notag\\
		&\left.-\vartheta_2\frac{l_i^t}{\overline{r}_{i,j}^t }-\vartheta_3\left(\eta_1\left(1+\frac{3 \left(v_{j}^{t}\right)^2}{{v_j^{\text{tip}}}^2}\right)+\eta_2 \phi+\eta_4 \left(v_{j}^{t}\right)^3\right)\delta\right)\\
		\text{s.t.} \ &\overline{r}_{i,j}^t \leq \tilde{\tilde {r}}_{i,j}^s,\ \forall i \in \mathcal{I}, \  j \in \mathcal{U}, \ t \in \mathcal{T}, \label{eq_Sp1_fin_c1} \\	
		& \frac{\eta_3}{\phi^2} \leq \tilde{\phi}^s\ \forall i \in \mathcal{I}, \  j \in \mathcal{U}, \ t \in \mathcal{T},\\
		 &\widetilde{d}_{j,j^{\prime}}^t\geq  d_{\text{U}}^{\text{safe}}, \ \forall j, j^{\prime}\in\mathcal{U}, j\neq j^{\prime}, \ t\in \mathcal{T}, \\
		 &\eqref{eq_UAV_mob_x} \sim \eqref{eq_UAV_mob_dis_each}, \notag
	\end{alignat}
\end{subequations}
}

\noindent where problem $\overline{\mathbf{P}}_\mathbf{t1}$ is convex since the objective function is concave and the feasible region is convex, which can be easily solved by optimization tools such as CVX.

\par The solution of UAV trajectory control is summarized in Algorithm \ref{alg_Traj}. First, in the $s$-th iteration, the lower bounds $\tilde{\tilde {r}}_{i,j}^{s}$ and $\tilde{\phi}^s$ are calculated (line 3). Then, the optimal trajectory $\mathbf{Q}^{s^*}$ of Problem $\overline{\mathbf{P}}_\mathbf{t1}$ is obtained as the local point for the next iteration $\mathbf{Q}^{s+1}$ (lines 4 and 5). The iteration ends when the difference in the objective value between successive iterations falls below a given threshold $\epsilon$.

\begin{algorithm}[]	
	\label{alg_Traj}	
	\SetAlgoLined
	\KwIn{UAV location $\mathbf{Q}^{t}$, optimal offloading strategy $\mathbf{F}^{t^*}$ and optimal resorce allocation strategy $\mathbf{O}^{t^*}$ in time slot $t$}
	\KwOut{UAV trajecoty in the next time epoch $\mathbf{Q}^{t^{\prime*}}$}
	\textbf{ Initialization:} 
	$\epsilon$, $s=0$, ${\hat{\mathbf{q}}}_j^{s}=\mathbf{q}_j^t$, $U^{s}=0$\;
	\Repeat{$|U^{s}-U^{s-1}|\leq \epsilon$}
	{
		Calculate $\tilde{\tilde {r}}_{i,j}^{s}$ and $\tilde{\phi}^s$ based on Eqs. \eqref{eq_r_firstTaylor} and \eqref{eq_problemSp1_fin}\;
		Solve Problem $\overline{\mathbf{P}}_\mathbf{t1}$ to obtain the optimal trajectory $\mathbf{Q}^{s^*}$ and objective value $U^{s^*}$\;
		Update $\mathbf{Q}^{s+1}=\mathbf{Q}^{s^*}$\;
	 	Update s = s + 1\;
	}	 	 	
	\Return {$\mathbf{Q}^{t^{\prime*}}$}\;
	\caption{UAV Trajectory Control.}
	\end{algorithm}

\subsection{Main Steps of TJCCT}
\label{sec_TJCCT}

\par The main steps of TJCCT are given in Algorithm \ref{algo_TJCCT}. Specifically, in each time slot, the MDs that have unprocessed computation tasks decide to process the tasks locally or offload them to the MEC servers (line 3). Then, obtain the optimal strategies of computing resource allocation and computation offloading in the current time slot by calling Algorithm \ref{algo_matching} (line 4). Furthermore, perform the computing resource allocation and computation offloading based on the obtained strategies in each time slot (line 5). Moreover, update the task processing state and available computing resources of MEC servers in each time slot (line 6). In addition, calculate the optimal trajectories of UAVs and update the mobility states of MDs and UAVs in each time epoch (lines 9 and 10). Finally, calculate and update the system utility system utility (lines 11 and 12).

\begin{algorithm}[!hbt]	
	\label{algo_TJCCT}	
	\SetAlgoLined
	\KwIn{$\mathcal{I}$, $\{b, \mathcal{U}\}$, $\mathcal{T}$}
	\KwOut{$U$}
	\textbf{Initialization:} $t=0$, $U=0$\;
	\While{$t\leq T$}
	{		
		Each MD processes the task locally if $U_{i,0}^t>0$ \;
		Call Algorithm \ref{algo_matching} to obtain $\Pi^{t^*}$,  $\mathbf{O}^{t^*}$, and $\mathbf{F}^{t^*}$\;
		\For{ $\Pi_j^t \in \Pi^{t^*}$}{
			Perform computing resource allocation and charging\;
			Update the task processing state and available computing resources of MEC servers\;
		}	
            \If{$t \ \% \ \Delta == 0$}{
			Call Algorithm \ref{alg_Traj} to obtain $\mathbf{Q}^{t^{\prime*}}$\;
			Update the mobility of MDs and UAVs\;
		}
            Calculate the current system utility $U^t$\;
		Update the total system utility $U=U+U^t$\;
		$t = t+\delta$\;				
	}
	\Return $U$\;
	\caption{TJCCT}
\end{algorithm}

\section{Performance Analysis}
\label{sec_Analysis}

\par In this section, the stability, optimality, and computational complexity of the proposed TJCCT are analyzed.

\subsection{Stability} 
\label{sec_Stability}

\par The stability of TJCCT depends on the decision of computation offloading, which, in turn, relies on the result of matching $\Pi^{t^*}$. The stability of  $\Pi^{t^*}$ is proven by Theorem \ref{theo_stable}.

\begin{theorem}
	\label{theo_stable}
	The result of matching $\Pi^{t^*}$ is stable.
\end{theorem}

\begin{proof}
    \par Assuming the stability of matching $\Pi^{t^*}$ is not guaranteed, we can infer an unmatched pair ${\mathcal{K}_i^{\mathbf{t}}}^{\prime}\in \mathcal{K}_{\text{req}}^{t}$ and $ j^{\prime}\in \{ b \} \cup \mathcal{U}$ that prefer each other over their current matched counterparts., i.e.,   
    ${\mathcal{K}_i^{\mathbf{t}}}^{\prime}\notin \Pi_{j^{\prime}}^{t^*}$ and $j^{\prime}\neq \Pi_{{{\mathcal{K}_i^{\mathbf{t}}}^{\prime}}}^{t^*}$. Therefore, the following conditions hold: 
  \begin{subequations}
	\label{eq_pro_stability}
	\begin{alignat}{1}
            & j^{\prime} \succ_{{\mathcal{K}_i^{\mathbf{t}}}^{\prime}}\Pi_{{\mathcal{K}_i^{\mathbf{t}}}^{\prime}}^{t^*}, \label{eq_pro_opt_A}\\
            & \text{There exists } \Pi_{{\mathcal{K}_i^{\mathbf{t}}}^{\prime\prime}}\in \Pi^{t^*} \text{such that } \notag\\
            &{\mathcal{K}_i^{\mathbf{t}}}^{\prime} \succ_{j^{\prime}}{\mathcal{K}_i^{\mathbf{t}}}^{\prime\prime}\label{eq_pro_stability_B}, \text{ if } \Pi_{j^{\prime}}^{t^*}\neq \emptyset,\\
            &{\mathcal{K}_i^{\mathbf{t}}}^{\prime}\succ_{j^{\prime}} \emptyset, \text{ if }\Pi_{j^{\prime}}^{t^*} == \emptyset, \label{eq_pro_stability_C}\\
            &{\mathcal{K}_i^{\mathbf{t}}}^{\prime}\notin \Pi_{j^{\prime}}^{t^*},
            \label{eq_pro_stability_D}\\
            &j^{\prime}\neq \Pi_{{\mathcal{K}_i^{\mathbf{t}}}^{\prime}}^{t^*}, \label{eq_pro_stability_E}
	\end{alignat}
    \end{subequations}

\noindent where $\Pi_{{\mathcal{K}_i^{\mathbf{t}}}^{\prime}}^{t^*}$ represents the current MEC server currently assigned to task ${\mathcal{K}_i^{\mathbf{t}}}^{\prime}$, while $\Pi_{j^{\prime}}^{t^*}$ represents the set of tasks presently paired with MEC server $j^{\prime}$. 

\par The theorem can be demonstrated by showing that the necessary conditions cannot simultaneously hold. Specifically, if condition \eqref{eq_pro_opt_A} holds true, it implies that task ${\mathcal{K}_i^{\mathbf{t}}}^{\prime}$ prefers MEC server $j^{\prime}$ over its current matching MEC server $\Pi_{{\mathcal{K}_i^{\mathbf{t}}}^{\prime}}^{t^*}$. However, the following inferences can be obtained from condition \eqref{eq_pro_stability_D}:

\begin{itemize}
    \item  if $\Pi_{j^{\prime}}^t\neq \emptyset$, task ${\mathcal{K}_i^{\mathbf{t}}}^{\prime}$ is less preferred by MEC server $j^{\prime}$ than any task ${\mathcal{K}_i^{\mathbf{t}}}^{\prime\prime}\in \Pi_{j^{\prime}}^{t^*}$, i.e., ${\mathcal{K}_i^{\mathbf{t}}}^{\prime\prime} \succ_{j^{\prime}} {\mathcal{K}_i^{\mathbf{t}}}^{\prime}$, which contradicts condition \eqref{eq_pro_stability_B};
    \item if $\Pi_{j^{\prime}}^{t^*} == \emptyset$, MEC server $j^{\prime}$ prefers to be idle over be matched with task ${\mathcal{K}_i^{\mathbf{t}}}^{\prime}$, i.e., $\emptyset \succ_{j^{\prime}} {\mathcal{K}_i^{\mathbf{t}}}^{\prime}$, which contradicts condition \eqref{eq_pro_stability_C}.
\end{itemize}
 
 \par As a result, although task ${\mathcal{K}_i^{\mathbf{t}}}^{\prime}$ prefers MEC server $j^{\prime}$, MEC server $j^{\prime}$ does not prefer task ${{\mathcal{K}_i^{\mathbf{t}}}^{\prime}}$. Therefore, the matching will not be formed between ${{\mathcal{K}_i^{\mathbf{t}}}^{\prime}}$ and $j^{\prime}$. Similar conclusions can be drawn by starting from the other conditions. Therefore, the result of matching $\Pi^{t^*}$ is stable.
\end{proof}

\par According to Theorem \ref{theo_stable}, it can be concluded that the proposed TJCCT is stable.

%
%
\subsection{Optimality} 
\label{sec_Optimality}

\par For computing resource allocation, the optimality can be easily obtained based on the theoretical result in Corollary \ref{cor_deal}. For computation offloading, the weak-Pareto optimality is proved in Theorem \ref{the_weak_pare}. Moreover, for UAV trajectory control, the optimality is proved in Theorem \ref{the_opt_convex}. 

\begin{theorem}
	\label {the_weak_pare}
	The result of computation offloading $\mathbf{O}^{t^*}$ is weak Pareto optimal.
\end{theorem}

\begin{proof}
	\par Assuming that the decision of computation offloading obtained by matching $\Pi^{t^*}$ is not weak Pareto optimal, it implies the existence of a better matching ${\Pi^t}^{\prime}$ such that ${\Pi_x^t}^{\prime}\succ_{x}\Pi_x^{t^*}$ for $\forall x \in \mathcal{K}_{\text{req}}^t \cup \{ b \} \cup \mathcal{U}$ \cite{pardalos2008pareto}. We prove the theorem by demonstrating that the assumption leads to contradictions from the perspectives of MDs and MEC servers, respectively. 
 
    \par For MDs, if the above assumption holds, each task $\mathcal{K}_i^{\mathbf{t}}\in\mathcal{K}_{\text{req}}^{t}$ can be matched with a better MEC server $j^{\prime}=\Pi_{\mathcal{K}_i^{\mathbf{t}}}^{t^\prime}$ compared to the current offloading strategy $j=\Pi_{\mathcal{K}_i^{\mathbf{t}}}^{t^*}$. Therefore, for each task $\mathcal{K}_i^{\mathbf{t}}\in\mathcal{K}_{\text{req}}^{t}$, the following conditions hold: 
    \begin{subequations}
	\label{eq_pro_stability}
	\begin{alignat}{1}
            & j^{\prime}={\Pi_{\mathcal{K}_i^{\mathbf{t}}}^{t^{\prime}}}\succ_{\mathcal{K}_i^{\mathbf{t}}}\Pi_{\mathcal{K}_i^{\mathbf{t}}}^{t^*}=j, \ j \neq j^{\prime}, \label{eq_pro_opt_A}\\
            & \Pi_{\mathcal{K}_i^{\mathbf{t}}}^{t^*}=j, \label{eq_pro_opt_B}\\
            & \Pi_j^{t^*}=\mathcal{K}_i^{\mathbf{t}}, \label{eq_pro_opt_C}\\
            &\Pi_{\mathcal{K}_i^{\mathbf{t}}}^{t^\prime}=j^{\prime},
            \label{eq_pro_opt_D}\\
            &\Pi_{j^{\prime}}^{t^\prime}=\mathcal{K}_i^{\mathbf{t}}, \label{eq_pro_opt_E}
	\end{alignat}
    \end{subequations}
 
\noindent where the conditions \eqref{eq_pro_opt_B} and \eqref{eq_pro_opt_C} are derived from the current matching $\Pi^t$, while the conditions \eqref{eq_pro_opt_D} and \eqref{eq_pro_opt_E} are are derive from matching $\Pi^{t^\prime}$. 

\par The contractions can be obtained from conditions \eqref{eq_pro_opt_A} to \eqref{eq_pro_opt_E}. Specifically, condition \eqref{eq_pro_opt_A} implies that task $\mathcal{K}_i^{\mathbf{t}}$ prefers MEC server $j^{\prime}$ over MEC server $j$ under matching $\Pi^{t^\prime}$. However, conditions \eqref{eq_pro_opt_B} and \eqref{eq_pro_opt_C} reveal that, under the current matching $\Pi^t$, task $\mathcal{K}_i^{\mathbf{t}}$ is actually matched with MEC server $j$, not its preferred MEC server $j^{\prime}$. Therefore, the following inferences can be made: 
\begin{itemize}
    \item if $j^{\prime}=\emptyset$, task $\mathcal{K}_i^{\mathbf{t}}$ prefers to be unmatched rather than be matched, i.e., $\emptyset\succ_{\mathcal{K}_i^{\mathbf{t}}} j,\forall j\in\  \{ b \} \cup \mathcal{U}$, which contradicts the condition  \eqref{eq_pro_opt_B};
    \item if MEC $j^{\prime}>0$, task $\mathcal{K}_i^{\mathbf{t}}$ must have been deleted from the matching list of the MEC server $j^{\prime}$. 
\end{itemize}

\noindent As a result, MEC server $j^{\prime}$ does not prefer task $\mathcal{K}_i^{\mathbf{t}}$, which contradicts the condition \eqref{eq_pro_opt_E}. For the MEC servers, similar contradictions can be drawn. Consequently, the contradictions presented above disprove the validity of the assumption and demonstrate the weak-Pareto optimality of matching $\Pi^{*}$.
\end{proof}

\begin{theorem}
\label {the_opt_convex}
	Problem $\overline{\mathbf{P}}_\mathbf{t}$ does not change the optimality of problem $\mathbf{P_t}$.
	\vspace{-1pt}
\end{theorem}

\begin{proof}
	The inequalities of \eqref{eq_r_firstTaylor} and \eqref{eq_aux_vel} must hold at the optimum. Otherwise, the objective function of $\overline{\mathbf{P}}_\mathbf{t}$ can be further increased without violating the constraint \eqref{eq_r_firstTaylor} or \eqref{eq_aux_vel} by selecting a smaller $\overline{r}_{i,j}^{t^\prime}$ or a smaller $\phi$.  
\end{proof}

\par According to Theorems \ref{the_weak_pare} and \ref{the_opt_convex}, the optimality of TJCCT is demonstrated.

%
%
\subsection {Computation Complexity}
\label{sec_Complexity}

\par The computational complexity of TJCCT is given as Theorem \ref{theo_complexity}.

\begin{theorem}
	\label{theo_complexity}
	The proposed algorithm has a polynomial worst-case complexity in each time slot, i.e., $\mathcal{O}\big(\iota^{\max}(|\mathcal{U}|+1)\big(2|\mathcal{K}_{\text{req}}^{t}|+ \min\{|\mathcal{U}|+1, |\mathcal{K}_{\text{req}}^{t}|\}\big)+\big.$ $\big.({|\mathcal{I}||\mathcal{U}|})^{3.5}\log_2(\frac{1}{\epsilon})\big)$, where $|\mathcal{K}_{\text{req}}^{t}|$ and $|\mathcal{U}|$ are the numbers of undecided tasks and MEC servers in time slot $t$, respectively.
\end{theorem}

\begin{proof}
    For computing resource allocation, it can be inferred from Algorithm \ref{algo_allo_price} that the worst-case complexity is $\mathcal{O}(\iota^{\max})$. For computation offloading, it can be derived from Algorithm \ref{algo_matching} that the computational complexity for constructing the preference list is $\mathcal{O}\left(|\mathcal{K}_{\text{req}}^{t}|\left(|\mathcal{U}|+1\right) \right)$, where  $|\mathcal{U}|+1$ represents the number of MEC servers and $|\mathcal{K}_{\text{req}}^{t}|$ represents the number of undecided tasks in time slot $t$. Furthermore, for matching construction, in the worst case, any task could be rejected $|\mathcal{U}|+1$ times at most, i.e., the task is rejected by all MEC servers. Each time the tasks are rejected, there are at most $\min\{|\mathcal{U}|+1, |\mathcal{K}_{\text{req}}^{t}|\}$ MEC servers that need to update the preference list in the next iteration. Therefore, the worst-case computational complexity of matching construction can be obtained by $\mathcal{O}\left(\left(|\mathcal{U}|+1\right) \left(|\mathcal{K}_{\text{req}}^{t}|+ \min\{|\mathcal{U}|+1, |\mathcal{K}_{\text{req}}^{t}|\}\right)\right)$, and that of Algorithm \ref{algo_matching} is $\mathcal{O}\Big(\left(|\mathcal{U}|+1\right)\big(2|\mathcal{K}_{\text{req}}^{t}|+ \min\{|\mathcal{U}|+1, |\mathcal{K}_{\text{req}}^{t}|\}\big)\Big)$. For UAV trajectory control, the computational complexity of Algorithm \ref{alg_Traj} is $\mathcal{O}({|\mathcal{I}||\mathcal{U}|})^{3.5}\log_2(\frac{1}{\epsilon}))$ according to~\cite{Wang2022b}. As a result, the worst-case complexity of the proposed algorithm is $\mathcal{O}\big(\iota^{\max}(|\mathcal{U}|+1)\big(2|\mathcal{K}_{\text{req}}^{t}|+ \min\{|\mathcal{U}|+1, |\mathcal{K}_{\text{req}}^{t}|\}\big)+\big.$ $\big.({|\mathcal{I}||\mathcal{U}|})^{3.5}\log_2(\frac{1}{\epsilon})\big)$.
\end{proof}

\section{Simulation Results and Analysis}
\label{sec_simulation}

\par In this section, simulation results are presented to validate the effectiveness of the proposed approach.

\subsection{Simulation Setup}
\label{simulation_set_up}

\subsubsection{Scenarios} 

\par We consider a UAV-assisted MEC system where one MBS and four UAVs are deployed to jointly provide offloading service for 30 MDs in a $1000\times 1000 \ \text{m}^2$ rectangular area. The time horizon is set as $60$ s, and it is divided into $T = 600$ time slots with equal length of $\delta=100$ ms, and 10 time slots are grouped into a time epoch, i.e., $\Delta=1$ s. 

\subsubsection{Parameters} 

\par For the MBS, the location and height are set as $[500,500]$ and 10 m \cite{3GPPTR1389012018}, respectively. For UAVs, the fixed altitude is set as $H$ = 100 m, and the initial positions and destinations are set as $(\textbf{q}_{1}^{I},\textbf{q}_{1}^{F})=([50,900],[500,0])$, $(\textbf{q}_{2}^{I},\textbf{q}_{2}^{F})=([900,900],[500,500])$,
$(\textbf{q}_{3}^{I},\textbf{q}_{3}^{F})=([100,100],[500,500])$, $(\textbf{q}_{4}^{I},\textbf{q}_{4}^{F})=([800,1000],[500,500])$, respectively. The default values of the other parameters are listed in Table \ref{tab_simuParameter}.

\begin{table}[!hbp]
	\setlength{\abovecaptionskip}{3pt}%
	\setlength{\belowcaptionskip}{0pt}%
	\caption{Simulation parameters}
	\label{tab_simuParameter}
	\renewcommand*{\arraystretch}{1.2}
	\begin{center}
		\begin{tabular}{p{.07\textwidth}|p{.23\textwidth}|p{.12\textwidth}}
			\hline
			\hline
			\textbf{Symbol}&\textbf{Description}&\textbf{Default value}\\
			\hline
				    $\gamma$&CPU parameters& $10^{-27}$\\
			\hline
				    $\beta_{\mathrm{T}}^{\mathrm{L}}/\beta_{\mathrm{T}}^{\mathrm{N}}$&Path loss exponent for terrestrial communication &$2.42/4.28$ \cite{yang2018dense}\\
			\hline
				    $\beta_{\mathrm{A}}$&Path loss exponent for aerial communication &$2$ \cite{Zeng2019} \\
			\hline
				    $\kappa$&Additional attenuation factor for NLoS aerial communication &$0.2$ \cite{Zeng2019}\\
			\hline
				    $ G_i^{\max} $&The budget of MD $i$ for the costs payed to the servers& 20 \\
			\hline
				    $d_0^{\mathrm{T}}$/$d_0^{\mathrm{A}}$&Reference distances for MD-MBS/MD-UAV communication&1m/1 m \cite{3GPPTR1389012018}\\
			\hline
   				$d_1$/$d_2$&Environment parameters&18 m/36 m \cite{3GPPTR1389012018}\\
			\hline
                    $m_{\mathrm{T}}^{\mathrm{L}}/m_{\mathrm{T}}^{\mathrm{N}}/$\newline
                    $m_{\mathrm{A}}^{\mathrm{L}}/m_{\mathrm{A}}^{\mathrm{N}}$&Nakagami fading parameter &$4/2/3/1$ \cite{zhang2018dense,peng2022directional}\\
			\hline
                     $n_j^{\text{core}}$&CPU core number of MEC server $j$&[2, 10]\\
                \hline
			         $N_0$&Noise power& -98 dBm\\
			\hline
                     $\sigma^{\mathrm{L}}/\sigma^{\mathrm{N}}$&Standard deviation of shadowing for LoS/NLoS communication& $4/6$ \cite{3GPPTR389012020}\\	
			\hline
				    $\tau_i^t$&The deadline of task &[0.1, 5] s \cite{kazmi2021novel}\\
			\hline
				    $w_i/w_j$&Weight coefficient of MD $i$ / MEC server $j$ &[0, 1]/[0, 1]\\
			  \hline 	
				  $v_{\text{U}}^{\min}/v_{\text{U}}^{\max}$& The constraints of UAV velocity&0/30 m/s \cite{Zeng2019} \\
			  \hline
				 $p_1$, $p_2$& Parameters for LoS probability of MD-UAV link & 10, 0.6 \cite{AlHourani2014,Zeng2019}\\
			  \hline
			      $P_i$ &Transmit power of each MD & [10,25] dBm\\
			  \hline
				  $d^{\text{safe}}$ &Safety distance between UAVs &  10 m \cite{Ji2020}\\ 
			  \hline
				   $B_{i,j}$&Bandwidth between MD $i$ and MEC server $j$ & 20 MHz ($j=b$), \newline 10 MHz ($j\in \mathcal{U}$)\cite{Wang2022b}  \\
			  \hline 	 
    			 	$\mu_{i}^t$&Computation intensity of tasks&  [500, 1500] cycles/bit \\ 
			  \hline
			  	$l_i^t$&Task size & [1, 5] Mb \cite{Yu2020} \\
			  \hline
				  $E_i^{\max}$& Energy constraint of MD $i$&1 (W.h/GHz) \cite{ning2020intelligent}\\
			  \hline
				  $E_j^{\max}$& Energy constraint of MEC server $j$&1 (Wh/GHz) ($j=b$) 
                    \newline 360 kJ  ($j\in \mathcal{U}$)  \\
			  \hline
				  $f_{i}^{\max}$&Computing resources of MD $i$&[0.5, 1] GHz \cite{lyu2018energy}\\
			  \hline
				  $f_{j}^{\max}$&Computing resources of MEC server $j$  &[20,40] GHz ($j=b$), [10,20] GHz ($j\in \mathcal{U}$) \\ 
			   \hline
			         $\alpha$ &Memory degree of MD's velocity&0.9 \\
			   \hline
			         $\bar{v}_i$& The average velocity of MD $i$&$[0,1]$ m/s \cite{Yang2022}\\
			   \hline
			  	 $\varsigma_i^v$ & Standard derivation of velocity&2 \cite{Yang2022}\\
			   \hline
		\end{tabular}
	\end{center}
\end{table} 

\subsubsection{Benchmarks} 

\par This work evaluates the proposed TJCCT in comparison with the following schemes. 
\begin{itemize}
	\item \textit{Local-only strategy (LS)}: all MDs process the tasks locally without considering the MEC server resource allocation and UAV trajectory control.
	\item \textit{Equal computing resource allocation strategy (ECRAS)}: the available computing resources of each MEC server are equally allocated to the requested MDs, and the
	the strategies of computation offloading and trajectory control are determined based on TJCCT.
	\item \textit{Price adjustment strategy (PAS)} \cite{ tao2023single}: the price of the computing resource is adjusted upwards or downwards by a fixed factor based on the available computing resource of the MEC server.  Besides, the strategies of computation offloading, computing resource allocation, and trajectory control are determined based on TJCCT.
	\item \textit{Game-theoretic computation offloading strategy (GCOS)}\cite{Wang2020}: the computation offloading strategies of MDs are iteratively determined in a competitive manner while the strategies of computing resource allocation and trajectory control are determined based on the proposed TJCCT.
	\item \textit{Segment-constrained trajectory control strategy} (STCS) \cite{Wang2022a}: the strategy of UAV trajectory control is decided based on a segment-constrained method while the strategies of computation offloading and computing resource allocation are decided based on the proposed TJCCT.
\end{itemize}

\subsubsection{Performance indicators} 

\par To evaluate the overall performance of the proposed method, we adopt the following indicators. \textit{1) System utility} $ \sum_{t\in\mathcal{T}} U^t$, which indicates the aggregated utility of the MDs and MEC servers. \textit{2) Average processing rate} $ \frac{{\sum_{\mathcal{K}_i^t\in \mathcal{K}_{\text{succ}}}}l_i^t\mu_i^t}{T}$, which represents the CPU cycles that are completed per unit time, where $\mathcal{K}_{\text{succ}}$ is the set of tasks that are successfully computed, $T_{\mathcal{K}_i^t}^{\text{cmp}}$ and $ T_{\mathcal{K}_i^t}^{\text{gen}}$ are the completion time and generation time of the computation, respectively. \textit{3) Average completion delay} $\sum_{\mathcal{K}_i^t\in \mathcal{K}_{\text{succ}}}\frac{T_{\mathcal{K}_i^t}^{\text{cmp}} - T_{\mathcal{K}_i^t}^{\text{gen}}}{|\mathcal{K}_{\text{succ}}|}$, which indicates the average delay for successfully completing a task, where $|\mathcal{K}_{succ}|$ is the number of tasks that are completed. \textit{4) Average completion ratio} $\frac{|\mathcal{K}_{\text{succ}}|}{\sum_{t\in\mathcal{T}}\sum_{i \in \mathcal{I}}\zeta_i^t}$, which indicates the average ratio of tasks that are completed.

%
%

\subsection{Evaluation Results}
\label{Numerical Results}

\par In this section, we first evaluate the system performance of the proposed TJCCT over time with default parameters. Subsequently, we compare the impacts of different parameters on the performance of the proposed TJCCT and the benchmark algorithms.

\subsubsection{Performance Evaluation}

\par Figs. \ref{fig_time}(a), \ref{fig_time}(b), \ref{fig_time}(c), and \ref{fig_time}(d) compare the performances of system utility, average processing rate, average completion delay, and average cmpletion ration with time slots. It can be observed that the proposed TJCCT outperforms the other algorithms with gradual-significant performance advantages as the time elapses, and this can be attributed to several reasons. First, the price-incentive resource trading stimulates the MDs and MEC servers to negotiate the on-demand computing resources allocation. Additionally, the matching scheme employed by TJCCT enables mutually satisfactory computation offloading between MDs and terrestrial-aerial MEC servers based on the available computing resources of MEC servers and the QoE requirements of MDs. Moreover, the UAVs dynamically adjust their trajectories to provide satisfactory offloading services for MDs by using the trajectory control method of TJCCT. In conclusion, this set of simulation results demonstrates the superiority of TJCCT among the six algorithms, especially in bringing long-term benefits for both MDs and MEC servers.

\begin{figure*}[!hbt] 
	\centering
		\setlength{\abovecaptionskip}{2pt}%
		\setlength{\belowcaptionskip}{2pt}%
	\subfigure[System utility]
	{
		\begin{minipage}[t]{0.23\linewidth}
			\raggedleft
			\includegraphics[width=3in]{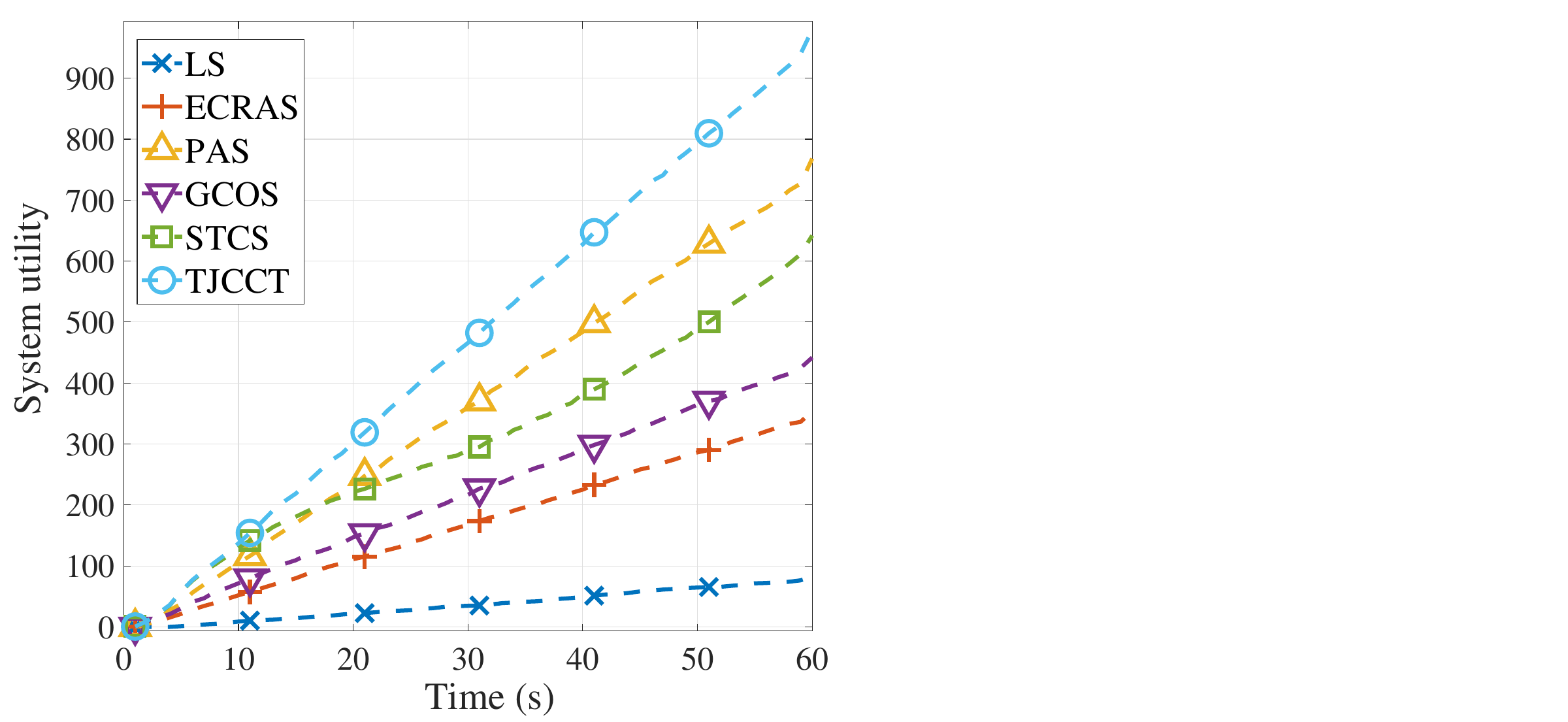}
		\end{minipage}
	}
	\subfigure[Average processing rate]
	{
		\begin{minipage}[t]{0.23\linewidth}
			\centering
			\includegraphics[width=3in]{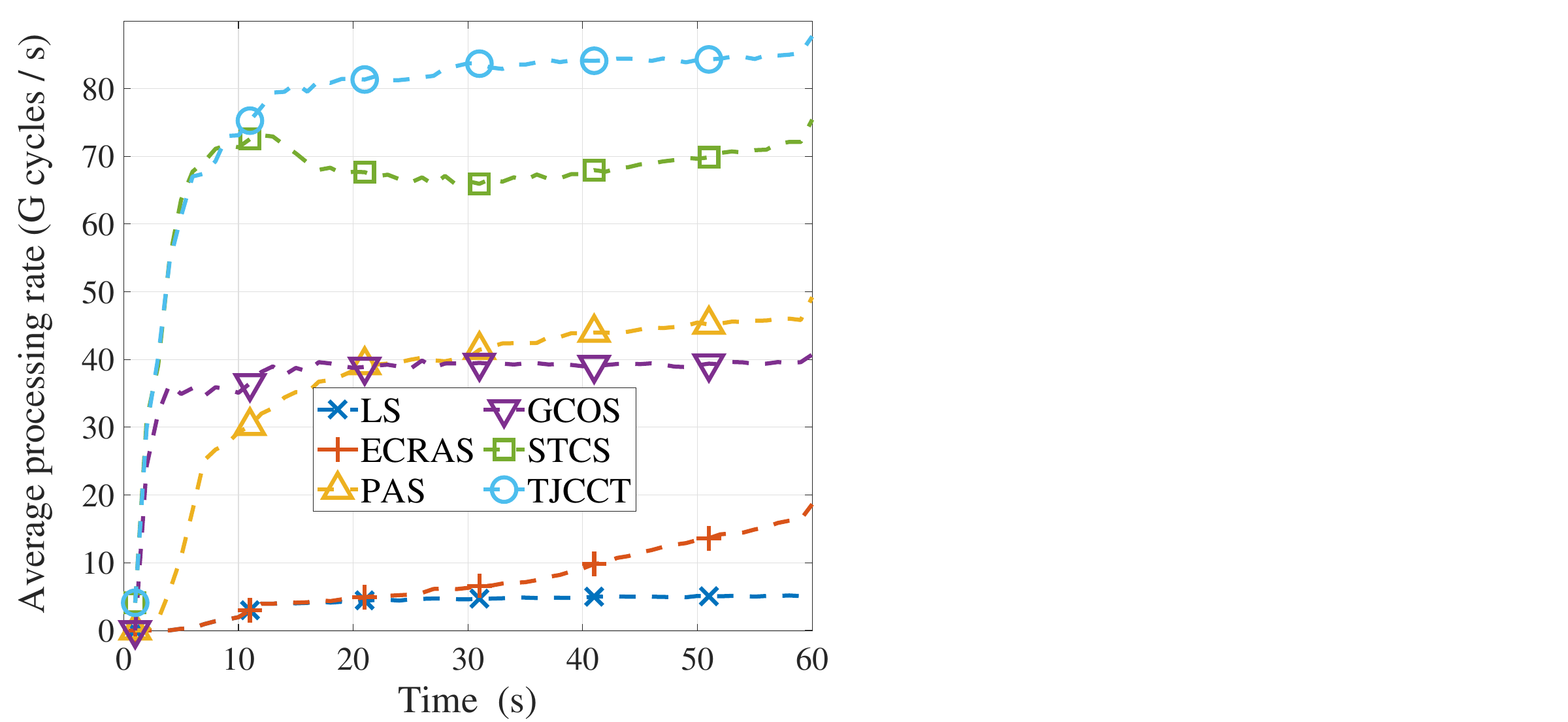}	
		\end{minipage}
	}
	\subfigure[Average completion delay]
	{
		\begin{minipage}[t]{0.23\linewidth}
			\centering
			\includegraphics[width=3in]{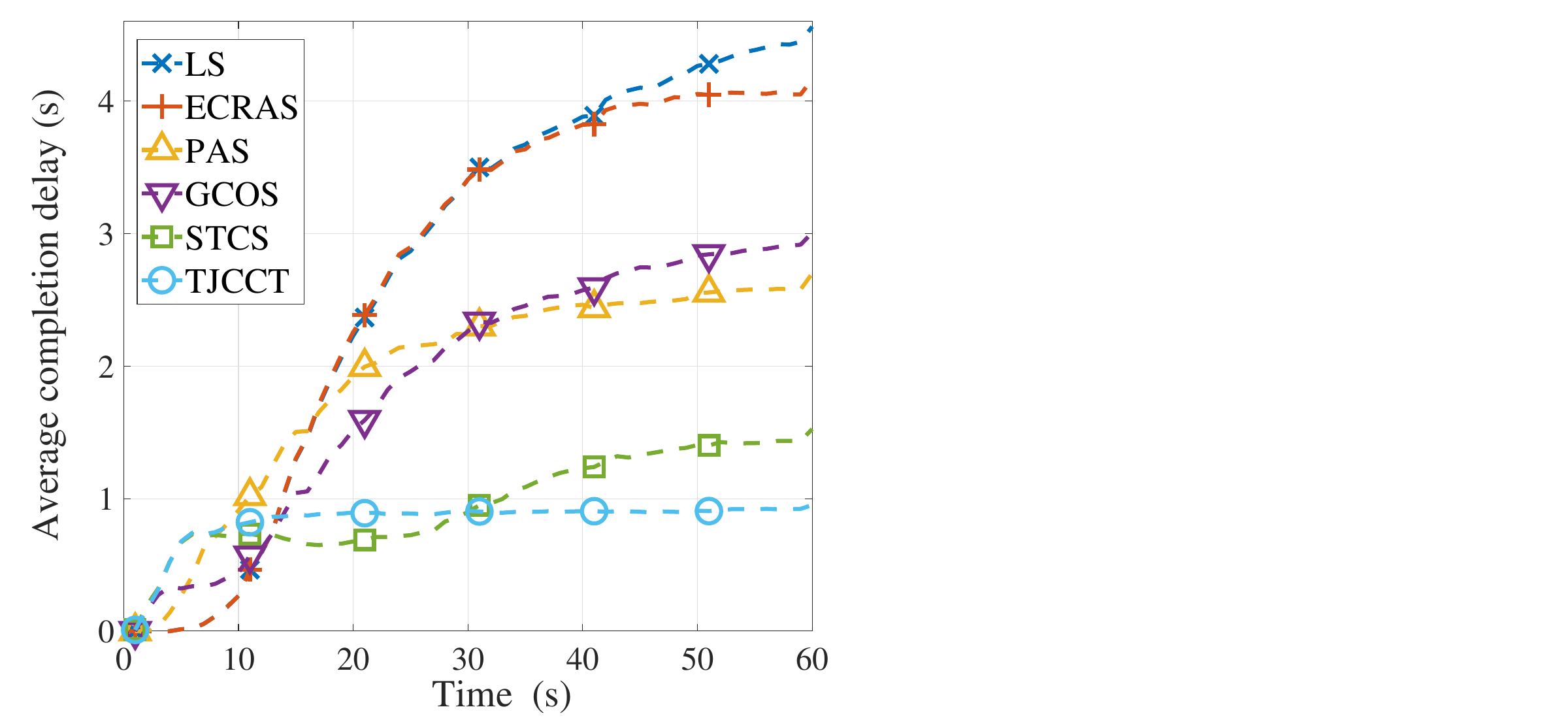}
		\end{minipage}
	}
	\subfigure[Average completion ratio]
	{
		\begin{minipage}[t]{0.23\linewidth}
			\centering
			\includegraphics[width=3in]{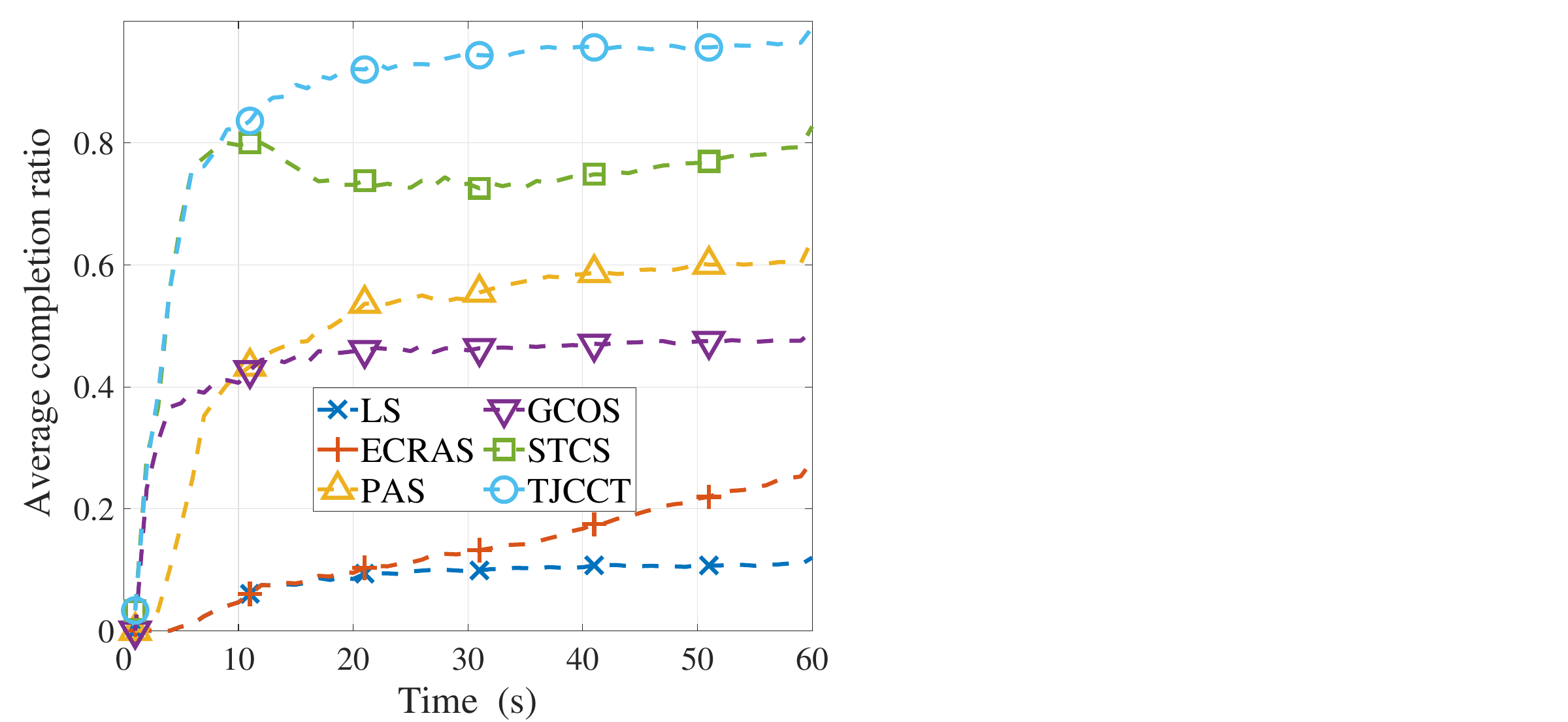}
		\end{minipage}
	}
	\centering
	\caption{System performance with time.}
	\label{fig_time}
	\vspace{-0.5em}
\end{figure*}

\subsubsection{Impact of Parameters}

\par \textbf{Impact of Average Computation Size.} Figs. \ref{fig_task}(a), \ref{fig_task}(b), \ref{fig_task}(c), and \ref{fig_task}(d) show the impact of average computation size on system utility, average processing rate, average completion delay, and average completion ratio, respectively. It can be observed from Fig. \ref{fig_task} that the proposed TJCCT achieves overall superior performances among the six algorithms as the workload increases. Specifically, in comparison to the other algorithms, TJCCT exhibits a relatively gradual decline in system utility, a significant upward trend in processing rate, a slow increase in the cost of completion delay, and a minimal decrease in the success ratio as the workload becomes heavier. Moreover, compared with ECRAS, PAS, GCOS, and STCS, the proposed TJCCT achieves approximately 138\%, 48\%, 345\%, and 87\% performance gains in terms of the average completion rate when the average computation size reaches 10 Mb. Additionally, JTCCT significantly outperforms the other algorithms in both completion delay and completion ratio, falling within the ranges of 32\% to 55\% and 46\% to 377\%, respectively. 

\par As the computational workload increases, the significant performance deterioration of LS, ECRAS, PAS, and GCOS can be attributed to the following reasons. First, LS relies on the computing capabilities of MDs, rendering it unable to handle intensive tasks. Furthermore, ECRAS, with its average resource allocation, overlooks the diverse demands of computation tasks. Moreover, PAS lacks adaptability to varying workloads due to the fixed price incentive factor. Besides, the competitive offloading mechanism of GCOS can lead to competition among MDs. Finally, the segment-based trajectory control of STCS primarily focuses on UAV energy consumption without considering the offloading demands of MDs. Consequently, this set of simulation results demonstrates that the proposed TJCCT is able to adapt to the heavy-loaded scenarios with significantly increased processing rate, relatively low costs of completion delay, and slightly decreased completion rate.

\begin{figure*}[!hbt] 
	\centering
	\subfigure[System utility]
	{
		\begin{minipage}[t]{0.23\linewidth}
			\raggedleft
			\includegraphics[width=3in]{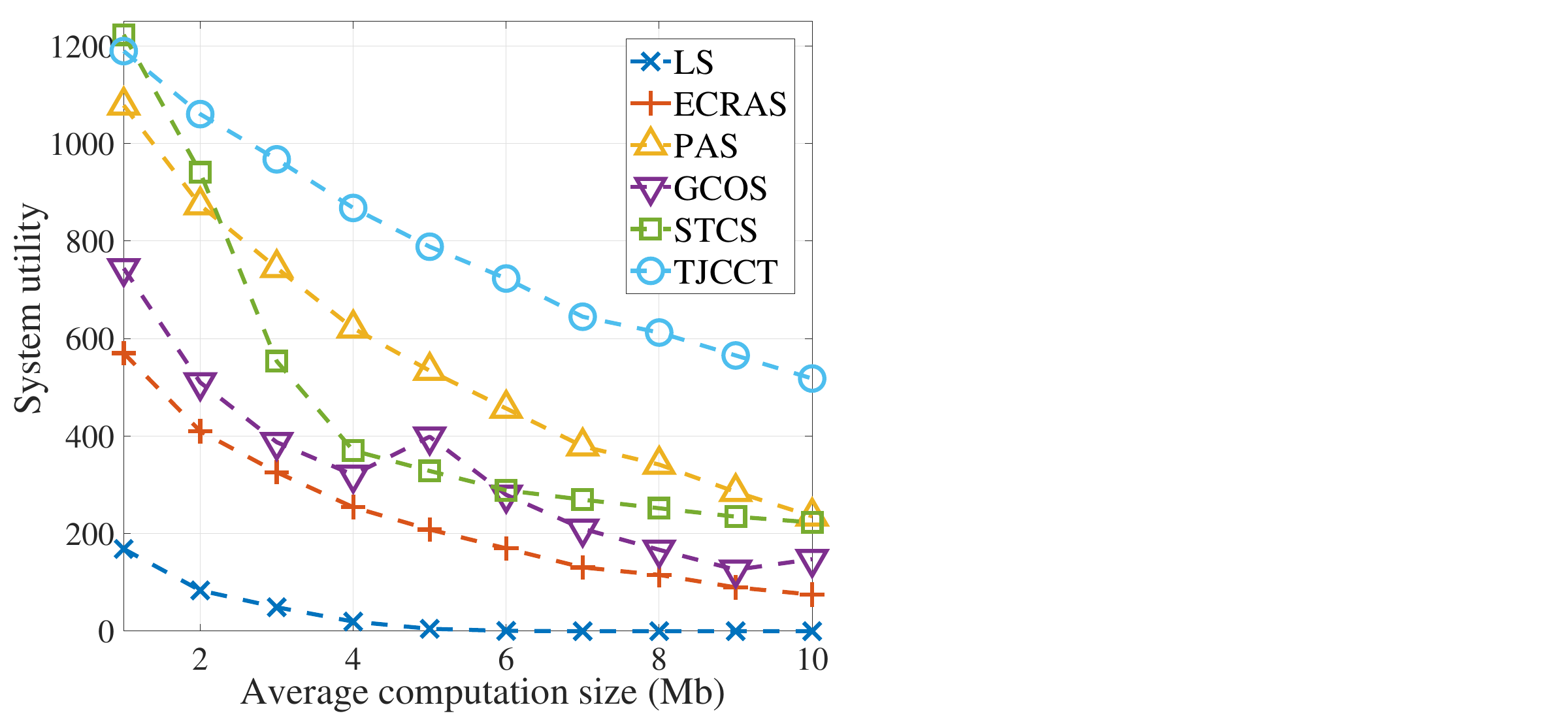}
		\end{minipage}
	}
	\subfigure[Average processing rate]
	{
		\begin{minipage}[t]{0.23\linewidth}
			\centering
			\includegraphics[width=3in]{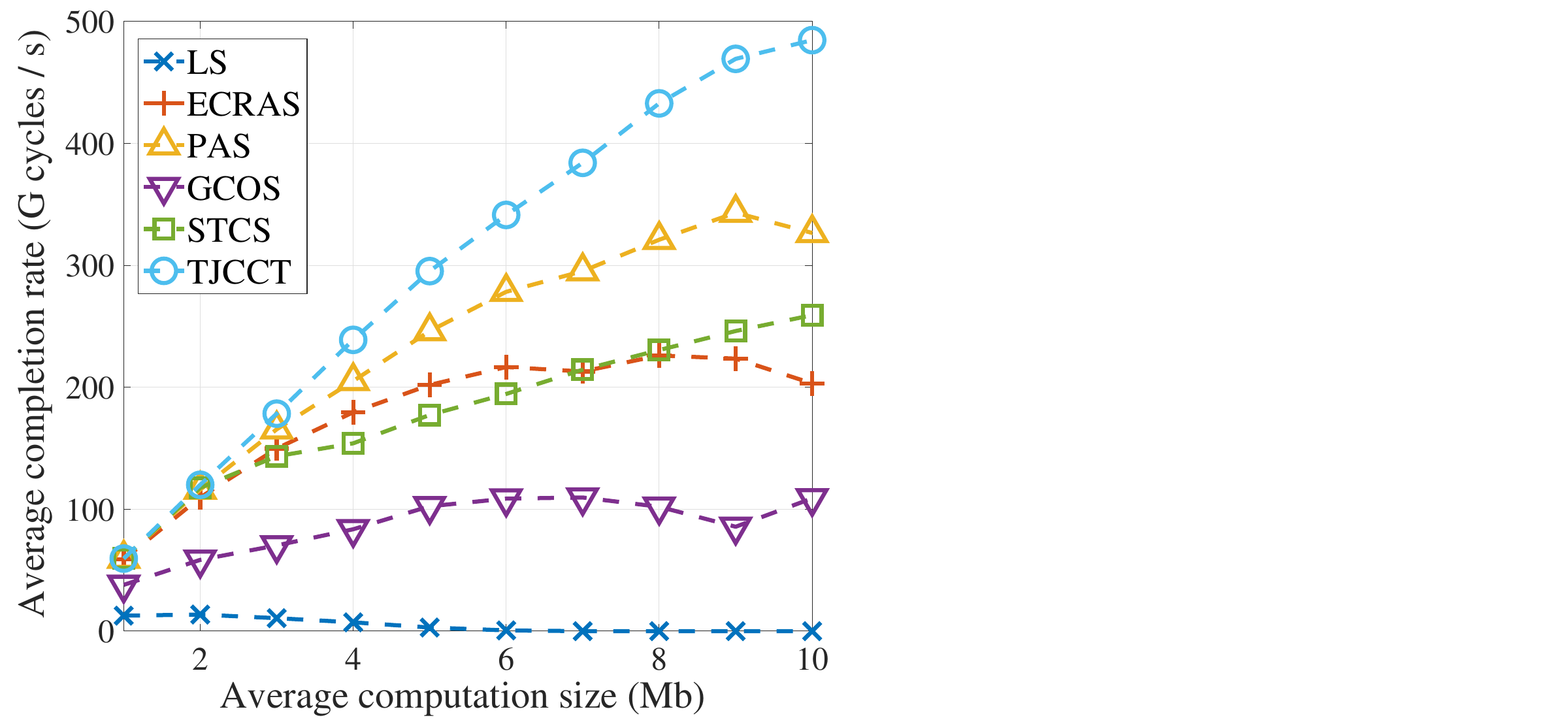}	
		\end{minipage}
	}
	\subfigure[Average completion delay]
	{
		\begin{minipage}[t]{0.23\linewidth}
			\centering
			\includegraphics[width=3in]{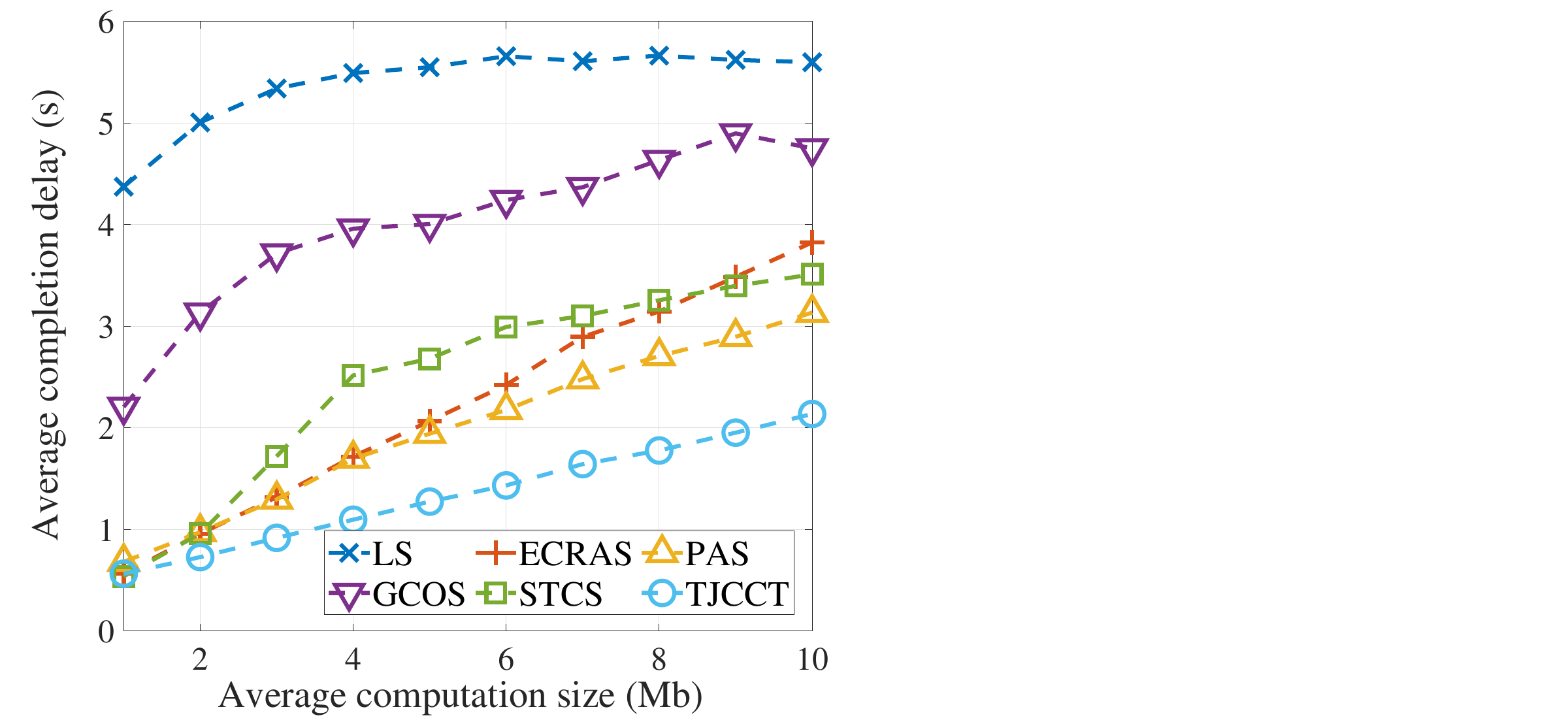}
		\end{minipage}
	}
	\subfigure[Average completion ratio]
	{
		\begin{minipage}[t]{0.23\linewidth}
			\centering
			\includegraphics[width=3in]{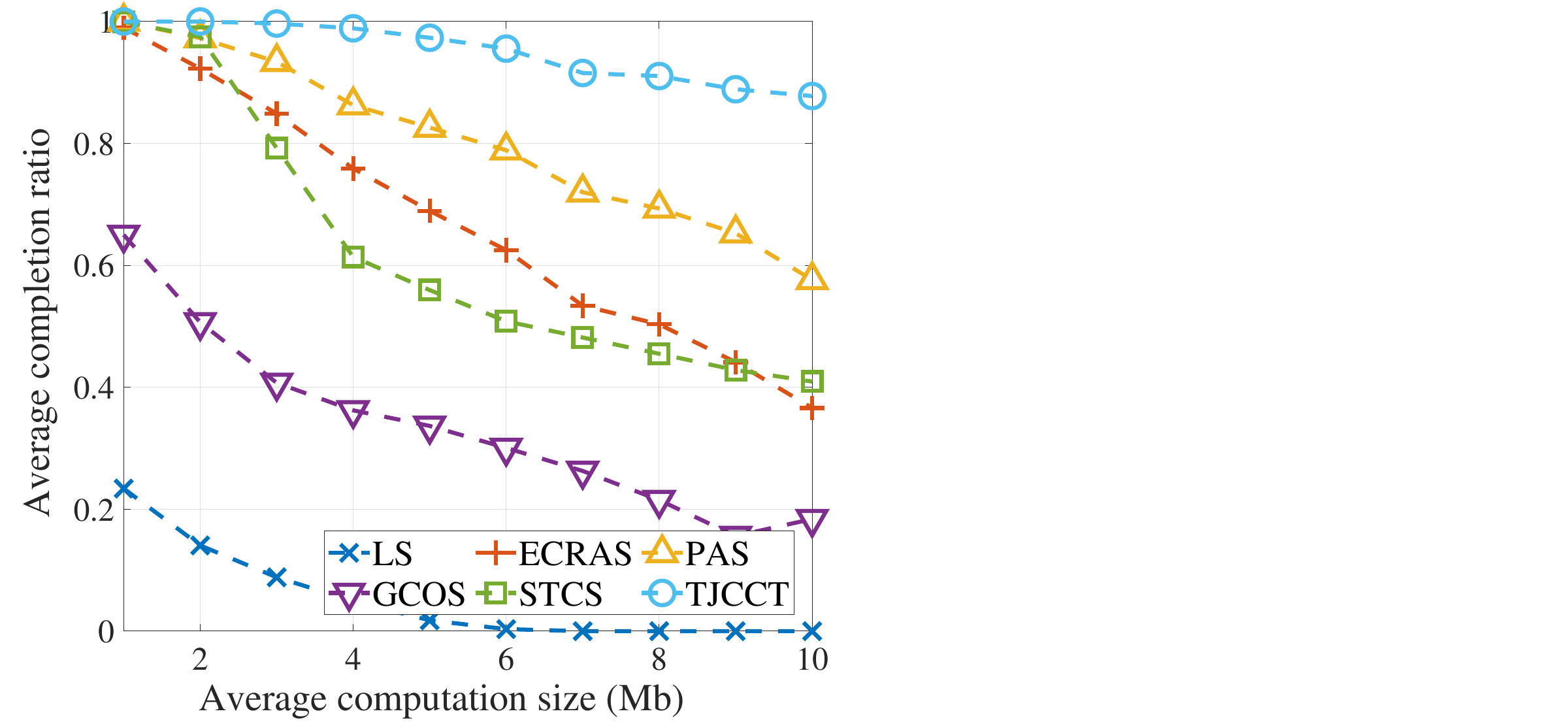}
		\end{minipage}
	}
	\centering
	\caption{System performance with average computation size.}
	\label{fig_task}
	\vspace{-0.5em}
\end{figure*}

\par \textbf{Impact of Average Computing Resources of MEC Server.} Figs. \ref{fig_cpu}(a), \ref{fig_cpu}(b),  \ref{fig_cpu}(c), and \ref{fig_cpu}(d) show the impact of average computing resources of MEC servers on system utility, average processing rate, average completion delay, and average completion ratio, respectively. Overall, it can be observed from Fig. \ref{fig_cpu} that the proposed TJCCT reveals superior performances among the six algorithms. Specifically, for ECRAS, Fig. \ref{fig_cpu}(a) shows a significant and consistent downward trend in system utility as the edge computing capabilities increase. Figs. \ref{fig_cpu}(b) and \ref{fig_cpu}(d) show that the completion ratio and completion rate of ECRAS initially remain almost constant but drop dramatically as the computing resources of MEC servers continuously increase (approximately exceeding 7 GHz). Contrary to the trends in Figs. \ref{fig_cpu}(a) and \ref{fig_cpu}(b), Fig. \ref{fig_cpu}(c) indicates that the average completion delay of ECRAS exhibits a slight downward trend initially but experiences a sudden upward trend when computing resources reaches 7 GHz. The sudden variations in Figs. \ref{fig_cpu}(a), \ref{fig_cpu}(a), and \ref{fig_cpu}(c) is explained as follows. Although MDs can achieve higher QoE with increased computing resources, the equal resource allocation of ECRAS could lead to excessive energy consumption for MEC servers and higher offloading costs for MDs, as the computing resources rise. As a result, more computation tasks are processed locally, leading to a marked deterioration in processing rate, completion delay, and completion ratio.

\par Furthermore, LS consistently delivers poor performance because it executes all computation tasks locally on the MDs without utilizing the assistance of MEC servers. In addition, PAS exhibits slight performance variation with increasing computing resources because computing resource allocation primarily depends on the fixed price incentive factors. Besides, both GCOS and STCS display relatively inferior performance gains in completion ratio, processing rate, and completion delay due to their competitive offloading and energy-dependent trajectory strategies. Comparatively, TJCCT achieves on-demand and efficient resource utilization by adjusting computation offloading, resource allocation, and trajectory control based on the available computing resources of MEC servers and the offloading requirements of MDs. Note that the proposed TJCCT exhibits a significant downward trend when the computing resource is approximately less than 2 GHz. This is because most of the tasks are processed locally when the remote computing resource is insufficient, resulting in reduced energy consumption overhead for both computing and flying. In conclusion, this set of simulation results indicates that the proposed TJCCT can achieve sustainable computing resource utilization and prevent resource over-utilization.

\begin{figure*}[!hbt] 
	\centering
	\subfigure[System utility]
	{
		\begin{minipage}[t]{0.23\linewidth}
			\raggedleft
			\includegraphics[width=3in]{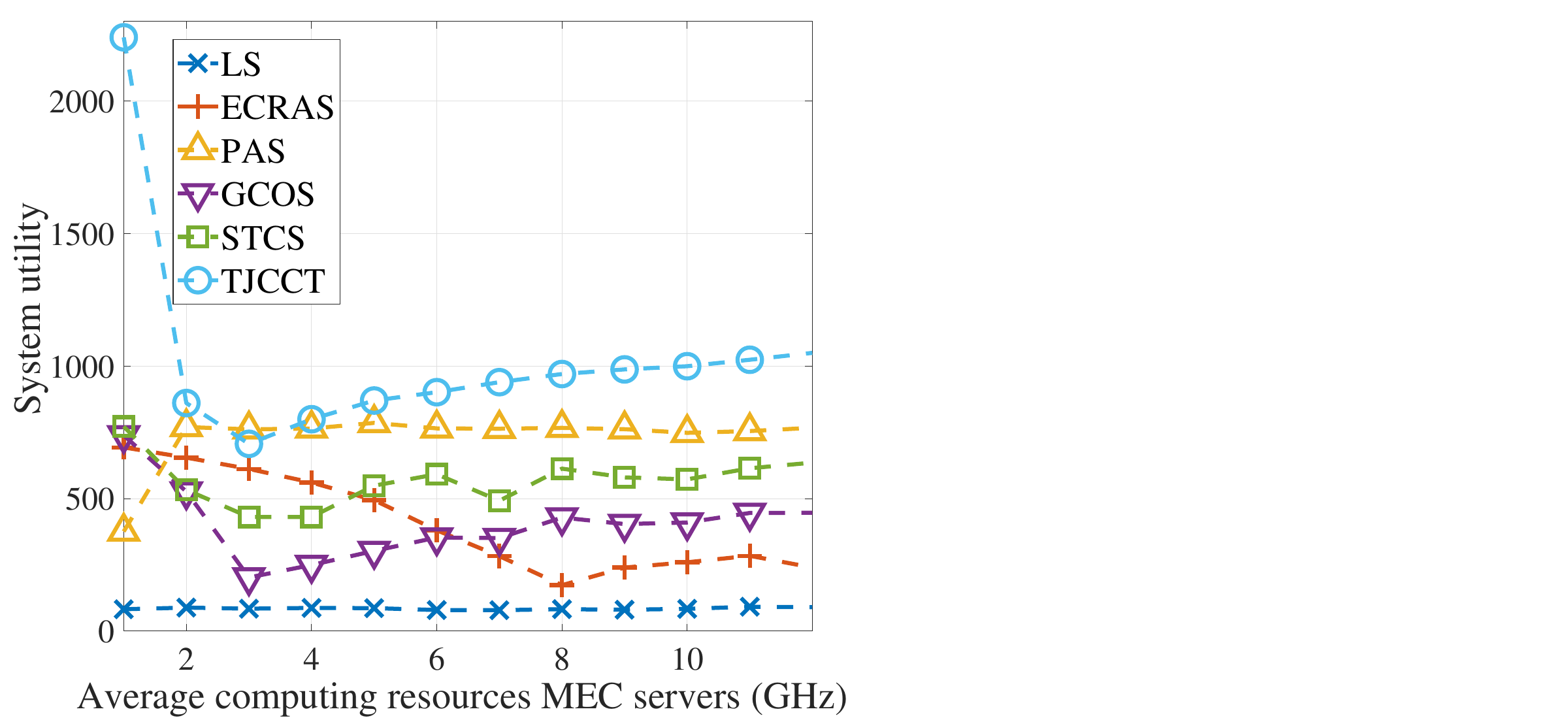}
		\end{minipage}
	}
	\subfigure[Average processing rate]
	{
		\begin{minipage}[t]{0.23\linewidth}
			\centering
			\includegraphics[width=3in]{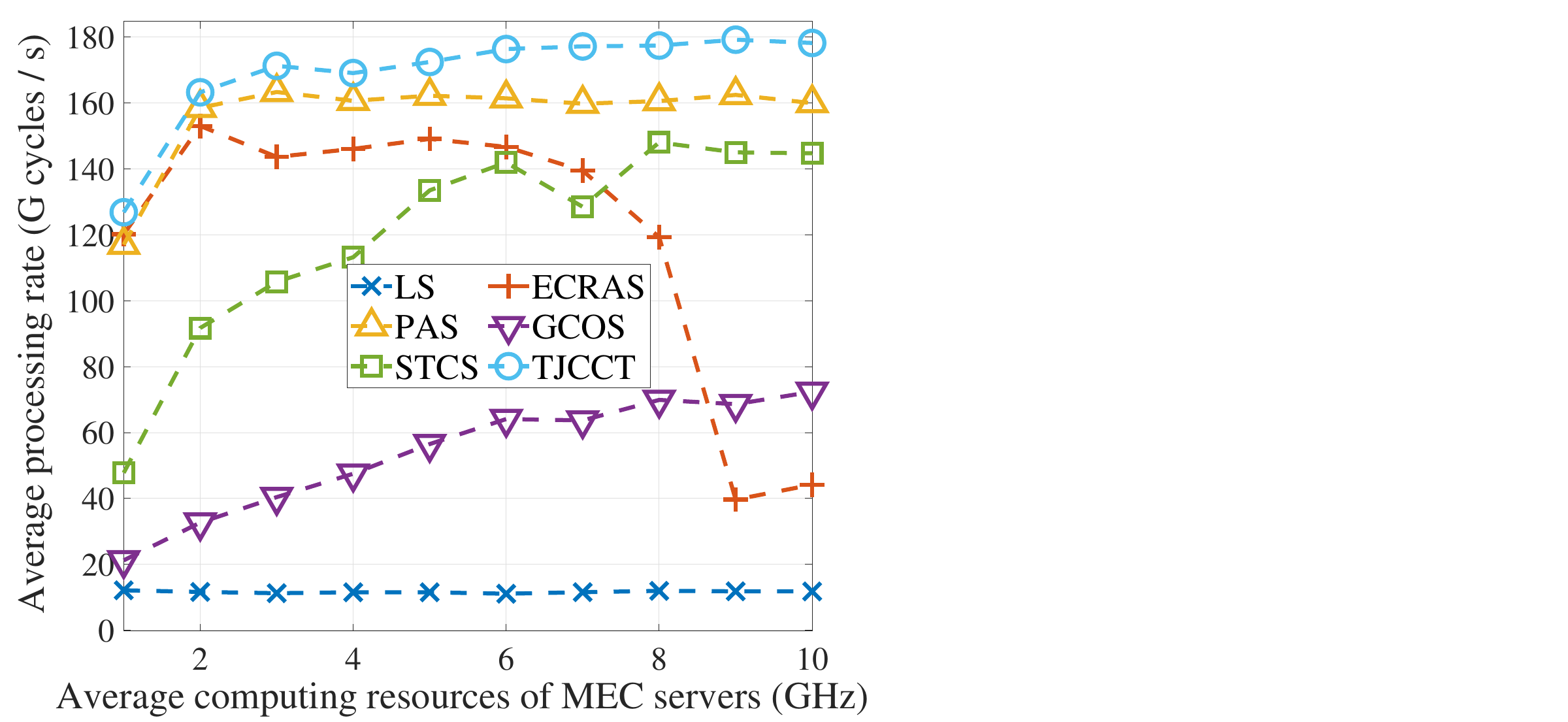}	
		\end{minipage}
	}
	\subfigure[Average completion delay]
	{
		\begin{minipage}[t]{0.23\linewidth}
			\centering
			\includegraphics[width=3in]{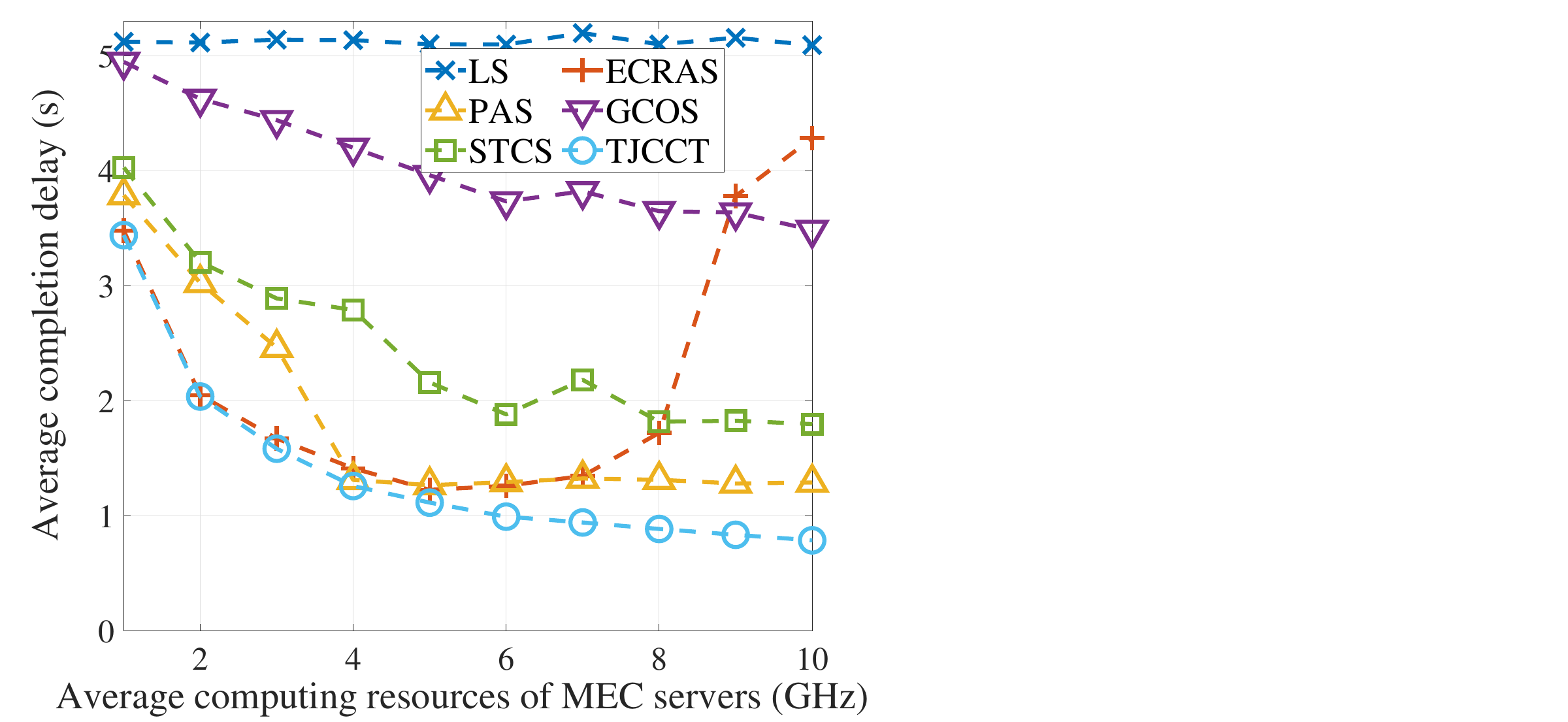}
		\end{minipage}
	}
	\subfigure[Average completion ratio]
	{
		\begin{minipage}[t]{0.23\linewidth}
			\centering
			\includegraphics[width=3in]{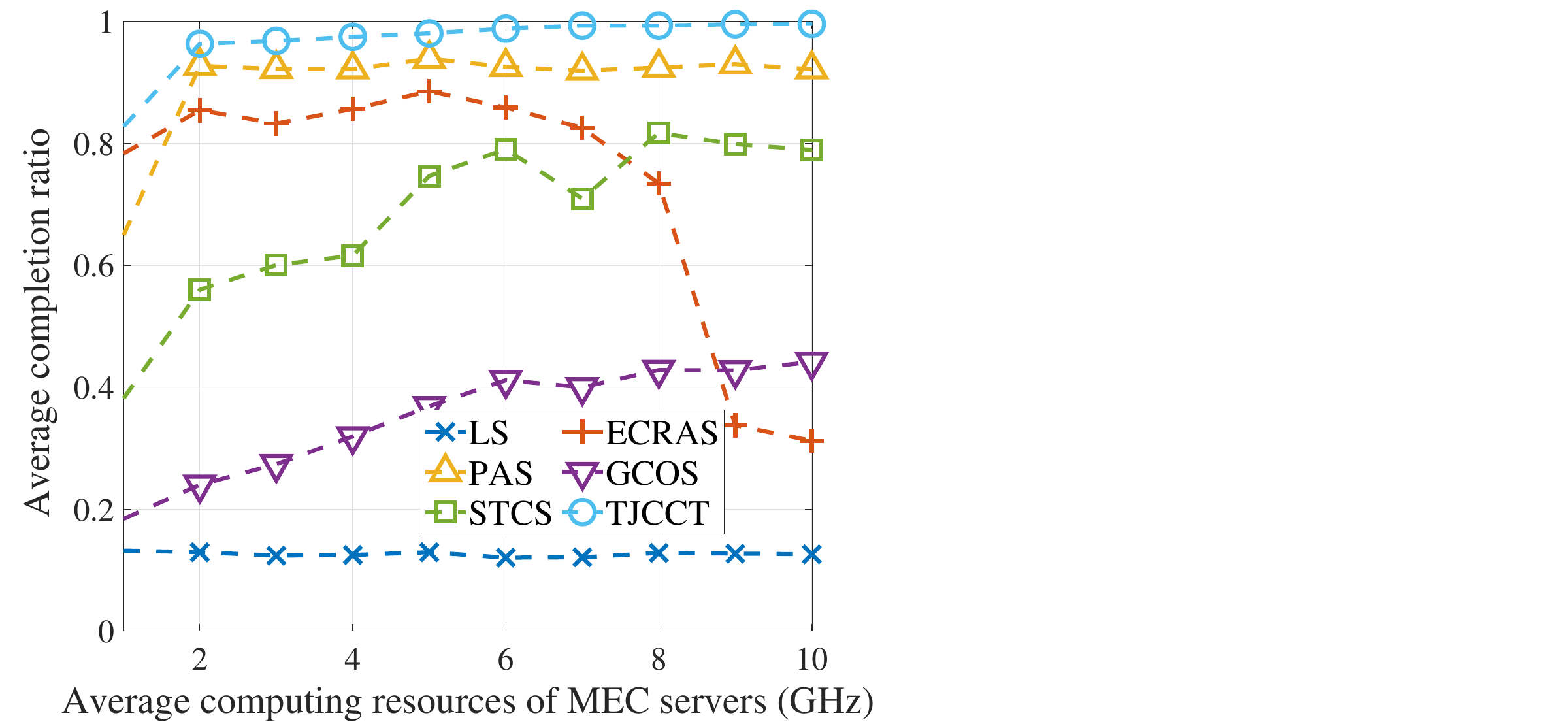}
		\end{minipage}
	}
	\centering
	\caption{System performance with average computing resources of MEC servers.}
	\label{fig_cpu}
	\vspace{-0.5em}
\end{figure*}

\par \textbf{Impact of MD Numbers.} Figs. \ref{fig_veh}(a), \ref{fig_veh}(b), \ref{fig_veh}(c), and \ref{fig_veh}(d) depict the impact of MD numbers on system utility, average processing rate, average completion delay, and average completion ratio, respectively. Overall, TJCCT consistently demonstrates superior performance in terms of system utility, average processing rate, average completion delay, and average completion ratio as the number of MDs increases. Specifically, Figs. \ref{fig_veh}(a) and \ref{fig_veh}(b) reveal that with an increasing number of MDs, the system utility and processing rate of TJCCT steadily rise, while those of the other algorithms exhibit minor initial upward trends, followed by gradual slowdowns, and approach stability or show decreasing tendencies. Furthermore, as shown in Figs. \ref{fig_veh}(c) and \ref{fig_veh}(d), it is evident that with an increasing number of MDs, the average completion delay and completion ratio of LS remain the worst levels while those of ECRAS, PAS, GCOS, and STCS exhibit obvious deterioration trends. In comparison, JTCCT show consistently superior performances with slight performance degradation in the average completion ratio and average completion delay, which vary within the range of 0.70 s to 1.25 s and 1s to 0.93 s, respectively. More specifically, compared with LS, ECRAS, PAS, GCOS, and STCS, the proposed JTCCT can respectively reduce the completion delay by 79\%, 41\%, 23\%, 63\%, 77\%, and can respectively improve the completion ratio by 661\%, 30\%, 8\%, 312\%, 68\% in the relative dense scenario ($|\mathcal{I}|\geq 90$). 

\par The main reasons for the phenomena in Fig. \ref{fig_veh} can be explained as follows. On the one hand, increasing amounts of computation are accomplished with the rising number of MDs, leading to the initial performance gains of most comparative algorithms. On the other hand, without efficient strategies for computation offloading, resource allocation, or trajectory control, the increasingly growing MDs could increase the competition among MDs and resource shortage at MEC servers, which further causes performance saturation or degradation. In conclusion, the proposed TJCCT has better scalability with an increasing number of MDs.

\begin{figure*}[!hbt] 
	\centering
	\setlength{\abovecaptionskip}{2pt}%
	\setlength{\belowcaptionskip}{2pt}%
	\subfigure[System utility]
	{
		\begin{minipage}[t]{0.23\linewidth}
			\raggedleft
			\includegraphics[width=3in]{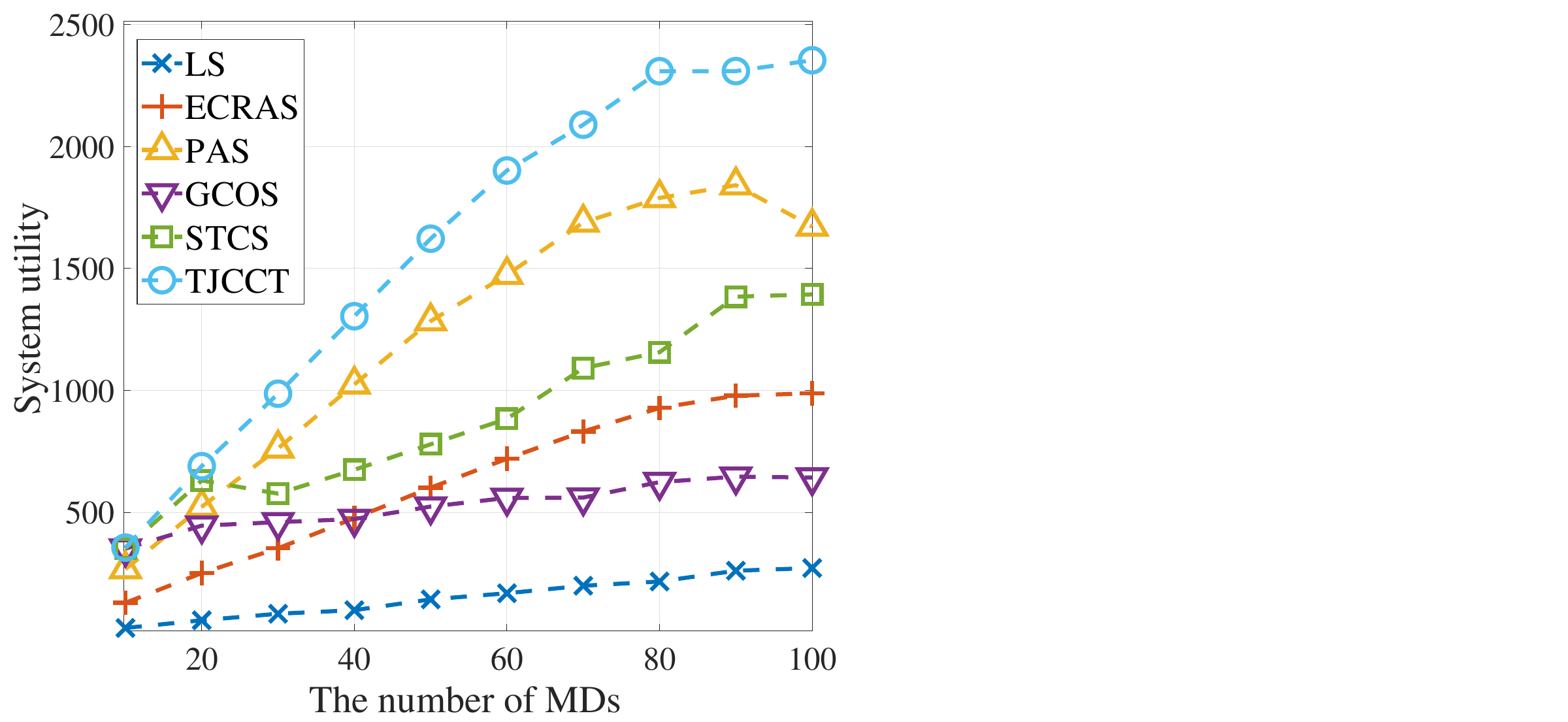}
		\end{minipage}
	}
	\subfigure[Average processing rate]
	{
		\begin{minipage}[t]{0.23\linewidth}
			\centering
			\includegraphics[width=3in]{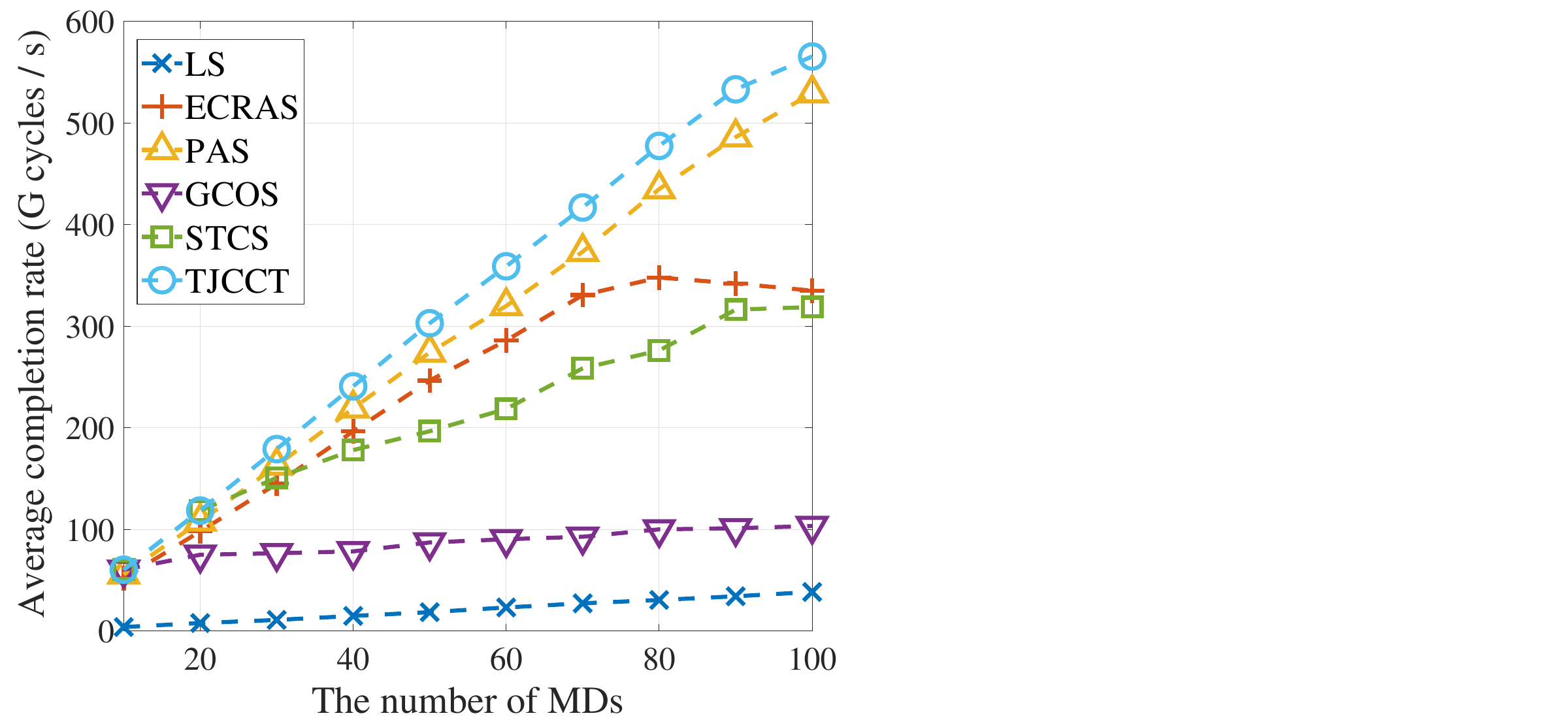}	
		\end{minipage}
	}
	\subfigure[Average completion delay]
	{
		\begin{minipage}[t]{0.23\linewidth}
			\centering
			\includegraphics[width=3in]{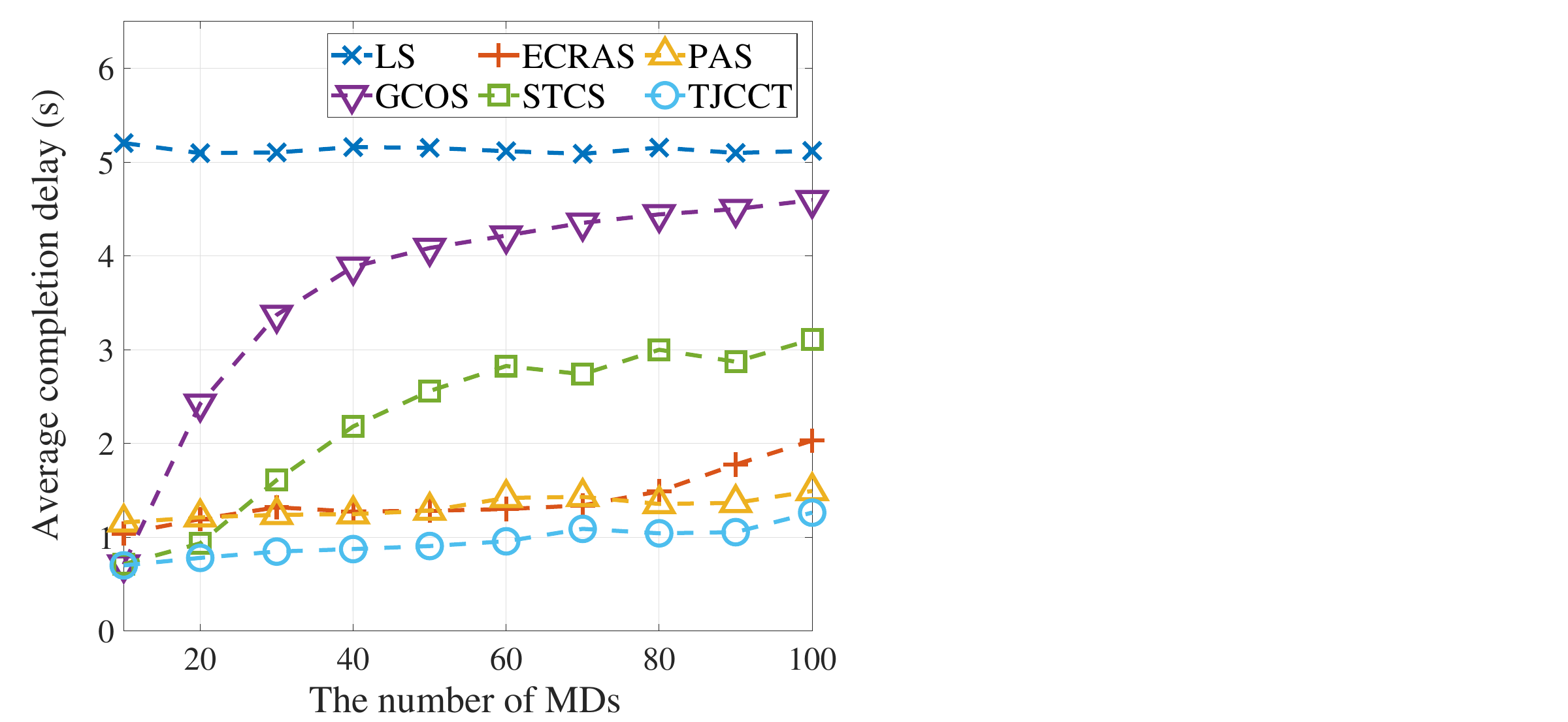}
		\end{minipage}
	}
	\subfigure[Average completion ratio]
	{
		\begin{minipage}[t]{0.23\linewidth}
			\centering
			\includegraphics[width=3in]{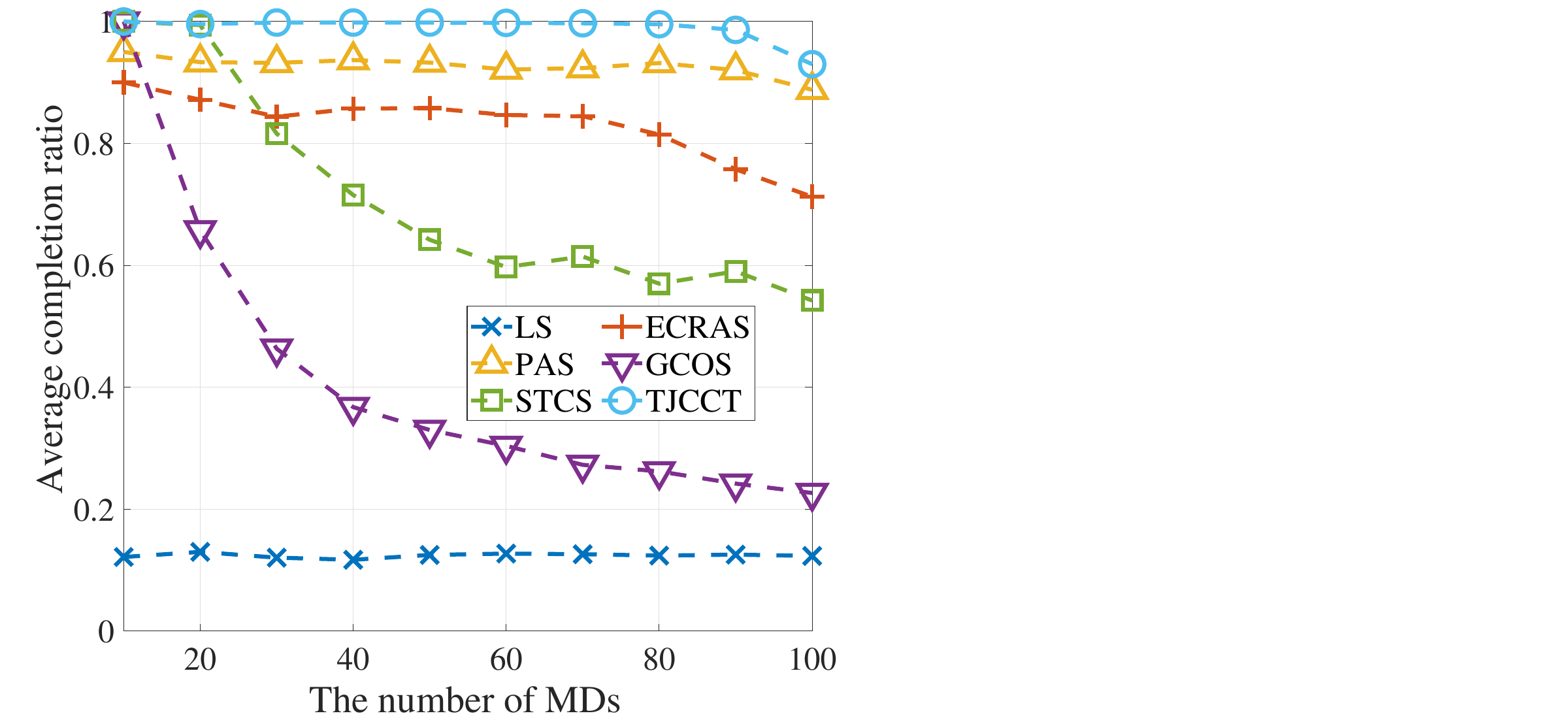}
		\end{minipage}
	}
	\centering
	\caption{System performance with the number of MDs.}
	\label{fig_veh}
	\vspace{-0.5em}
\end{figure*}

\begin{figure}[!hbt] 
	\centering
         \setlength{\abovecaptionskip}{2pt}%
	\setlength{\belowcaptionskip}{2pt}%
	\includegraphics[width=3.3in]{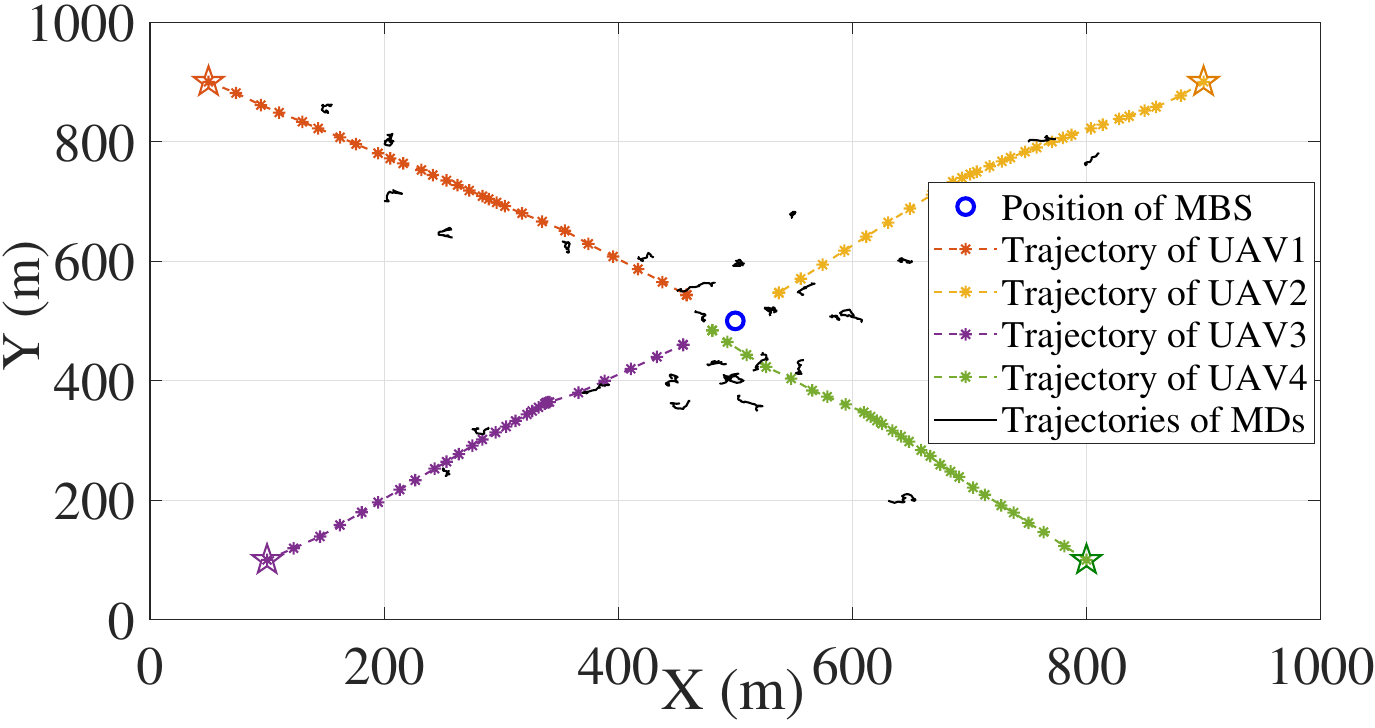}
	\caption{Trajectories of the MDs and UAVs.}
	\label{fig_tra}
        \vspace{-0.5em}
\end{figure}

\subsubsection{UAV Trajectory}  

\par Fig. \ref{fig_tra} shows the trajectories of the MDs and UAVs. It can be observed the trajectories of UAVs are consistent with the intuition. Specifically, UAVs tend to follow the trajectories of the MDs and tend to identify regions with dense data requirements. This is because the UAVs try to satisfy the QoE of the MDs for task offloading under the constraints of energy and velocity. Furthermore, it should be noted that the UAVs approach the final destinations but do not arrive at the destinations. The reason is that UAVs maintain a safe distance between each other to avoid collisions. In conclusion, the simulation result in Fig. \ref{fig_tra} demonstrates that the UAVs can provide satisfied services according to the dynamic requirements of MDs while guaranteeing the safety of UAVs by adopting the trajectory control of the proposed TJCCT.

%
%

\section{Conclusion}
\label{sec_conclusion}

\par In this work, we study computing resource allocation, computation offloading, and UAV trajectory control for UAV-assisted MEC system. First, we employ a hierarchical framework to coordinate the collaboration among MDs, terrestrial edge, aerial edge, and the controller. Then, we formulate an optimization problem to maximize the system utility. To solve the MINLP problem, we propose the TJCCT which consists of two-timescale optimization methods. In the short timescale, we propose a price-incentive model for on-demand computing resource allocation and a matching mechanism-based method for computation offloading. In the long timescale, we propose a convex optimization-based method for UAV trajectory control. Besides, the stability, optimality, and complexity of TJCCT are theoretically proved. Simulation results demonstrate that TJCCT could achieve superior performances in terms of the system utility, average processing rate, average completion delay, and average completion ratio. Furthermore, TJCCT exhibits superior adaptability in heavy-loaded scenarios and demonstrates good scalability with an increasing number of MDs.

\ifCLASSOPTIONcaptionsoff
\newpage
\fi

\bibliographystyle{IEEEtran}
\bibliography{references_1.bib}

\begin{IEEEbiography}[{\includegraphics[width=1in,height=1.23in,clip,keepaspectratio]{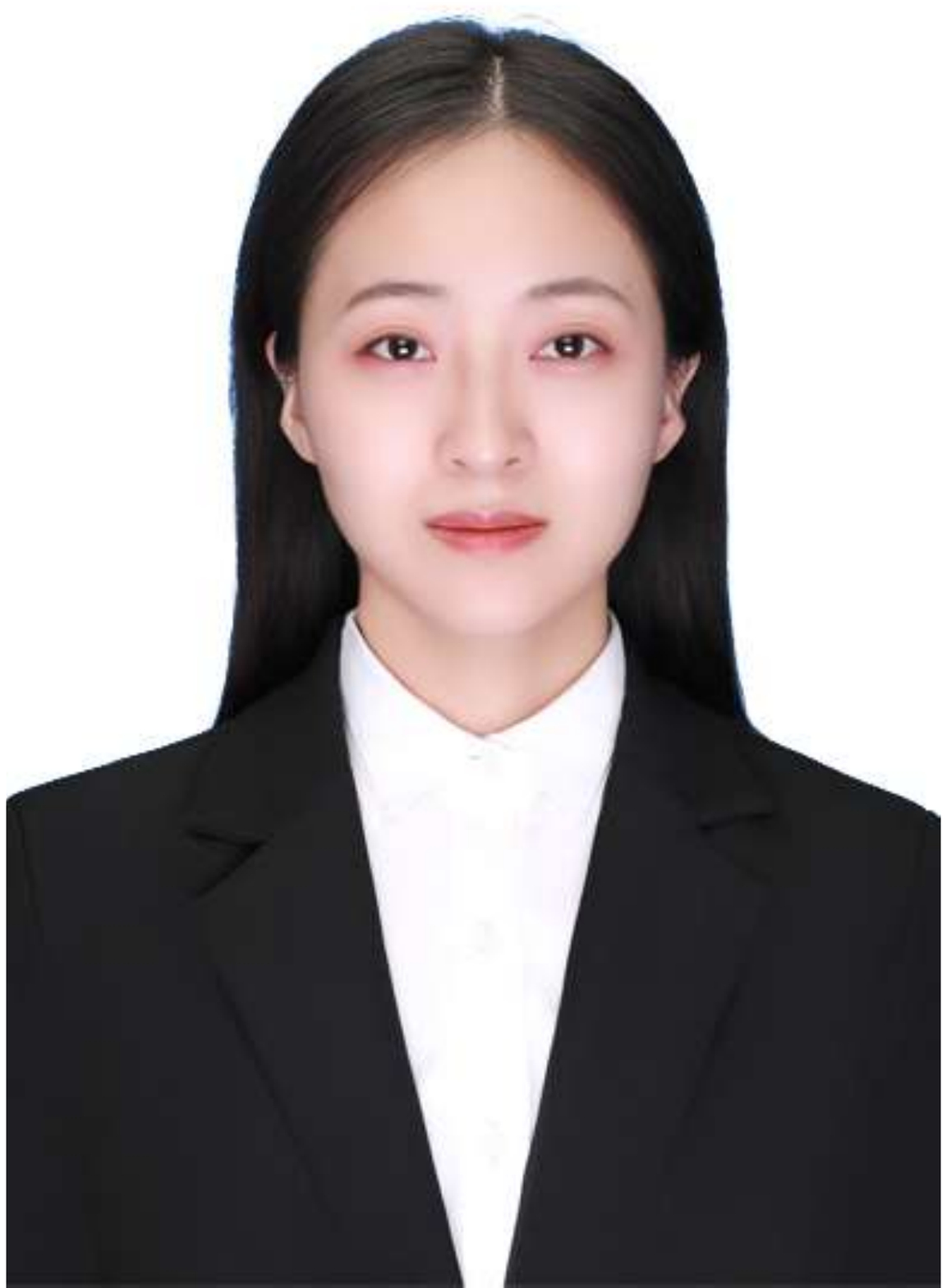}}]{Zemin Sun} received a BS degree in Software Engineering, an MS degree and a Ph.D degree in Computer Science and Technology from Jilin University, Changchun, China, in 2015, 2018, and 2022, respectively. Her research interests include vehicular networks, edge computing, and game theory. 
\end{IEEEbiography}

\begin{IEEEbiography}[{\includegraphics[width=1in,height=1.23in,clip,keepaspectratio]{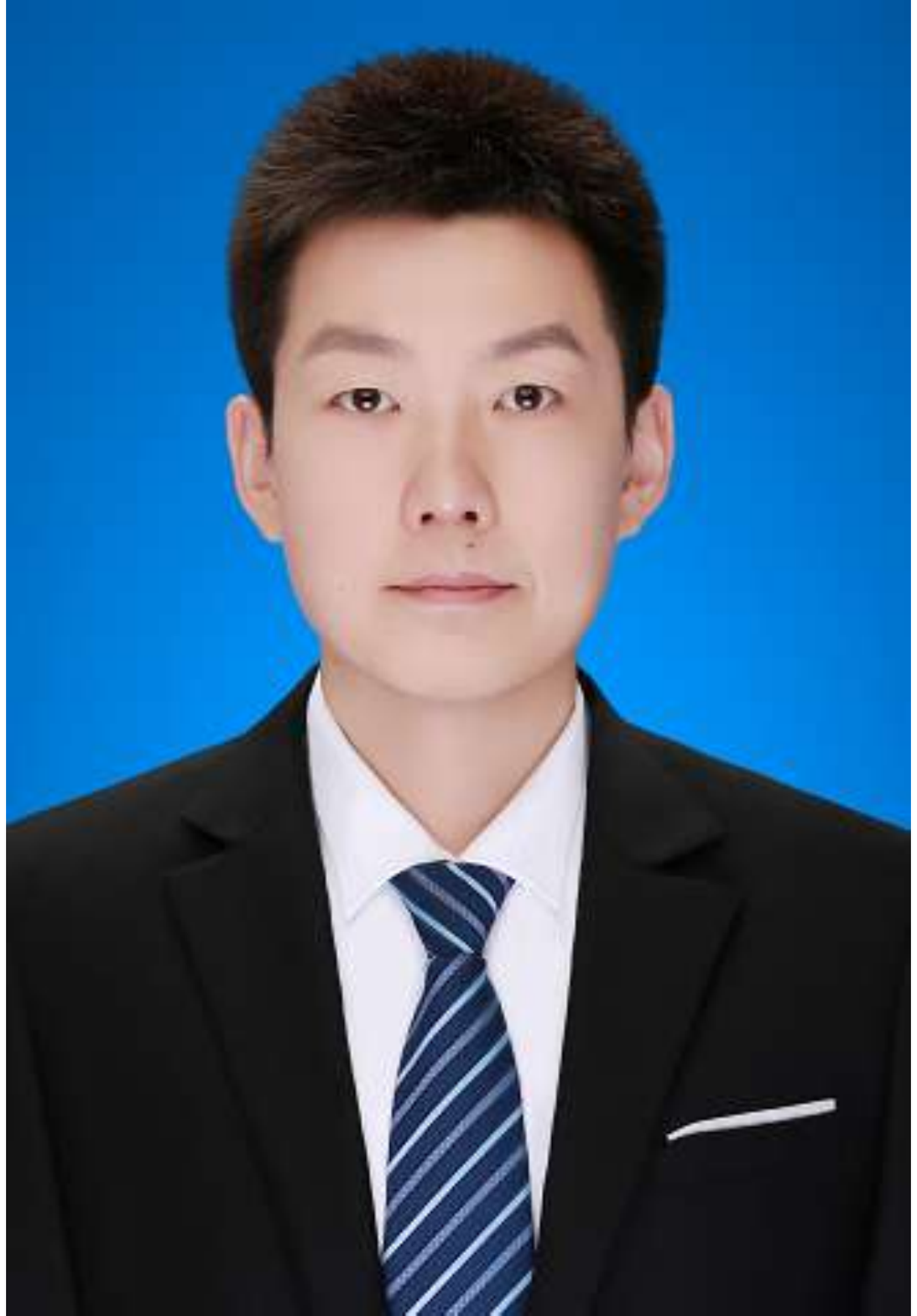}}]{Geng Sun} (S'17-M'19) received the B.S. degree in communication engineering from Dalian Polytechnic University, and the Ph.D. degree in computer science and technology from Jilin University, in 2011 and 2018, respectively. He was a Visiting Researcher with the School of Electrical and Computer Engineering, Georgia Institute of Technology, USA. He is an Associate Professor in College of Computer Science and Technology at Jilin University, and his research interests include wireless networks, UAV communications, collaborative beamforming and optimizations.
\end{IEEEbiography}

\begin{IEEEbiography}[{\includegraphics[width=1in,height=1.23in,clip,keepaspectratio]{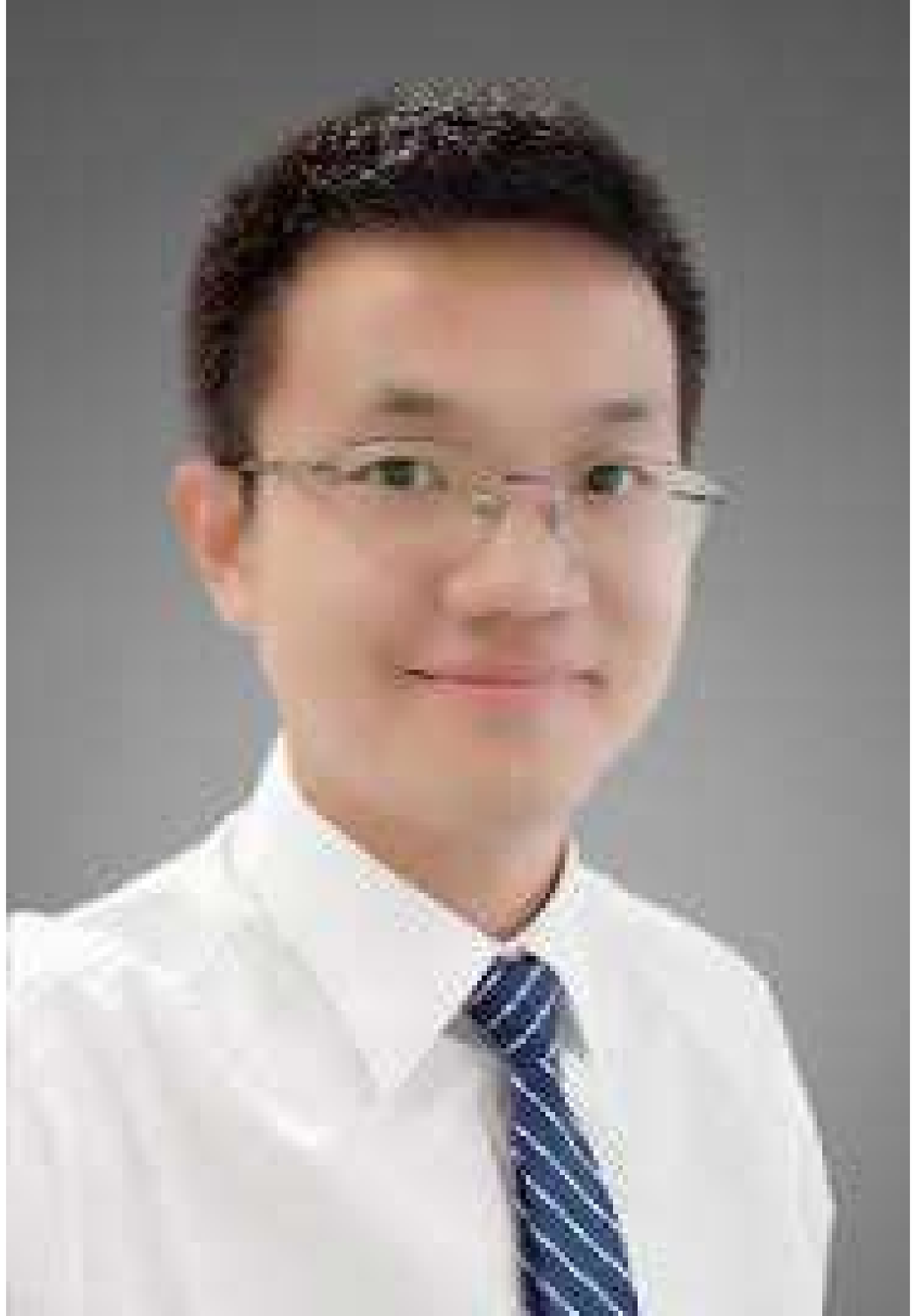}}]{Qingqing Wu} (S’13-M’16-SM’21) received the B.Eng. and the Ph.D. degrees in Electronic Engineering from South China University of Technology and Shanghai Jiao Tong University (SJTU) in 2012 and 2016, respectively. From 2016 to 2020, he was a Research Fellow in the Department of Electrical and Computer Engineering at National University of Singapore. He is currently an Associate Professor with Shanghai Jiao Tong University. His current research interest includes intelligent reflecting surface (IRS), unmanned aerial vehicle (UAV) communications, and MIMO transceiver design. 
\end{IEEEbiography}

\begin{IEEEbiography}[{\includegraphics[width=1in,height=1.23in,clip,keepaspectratio]{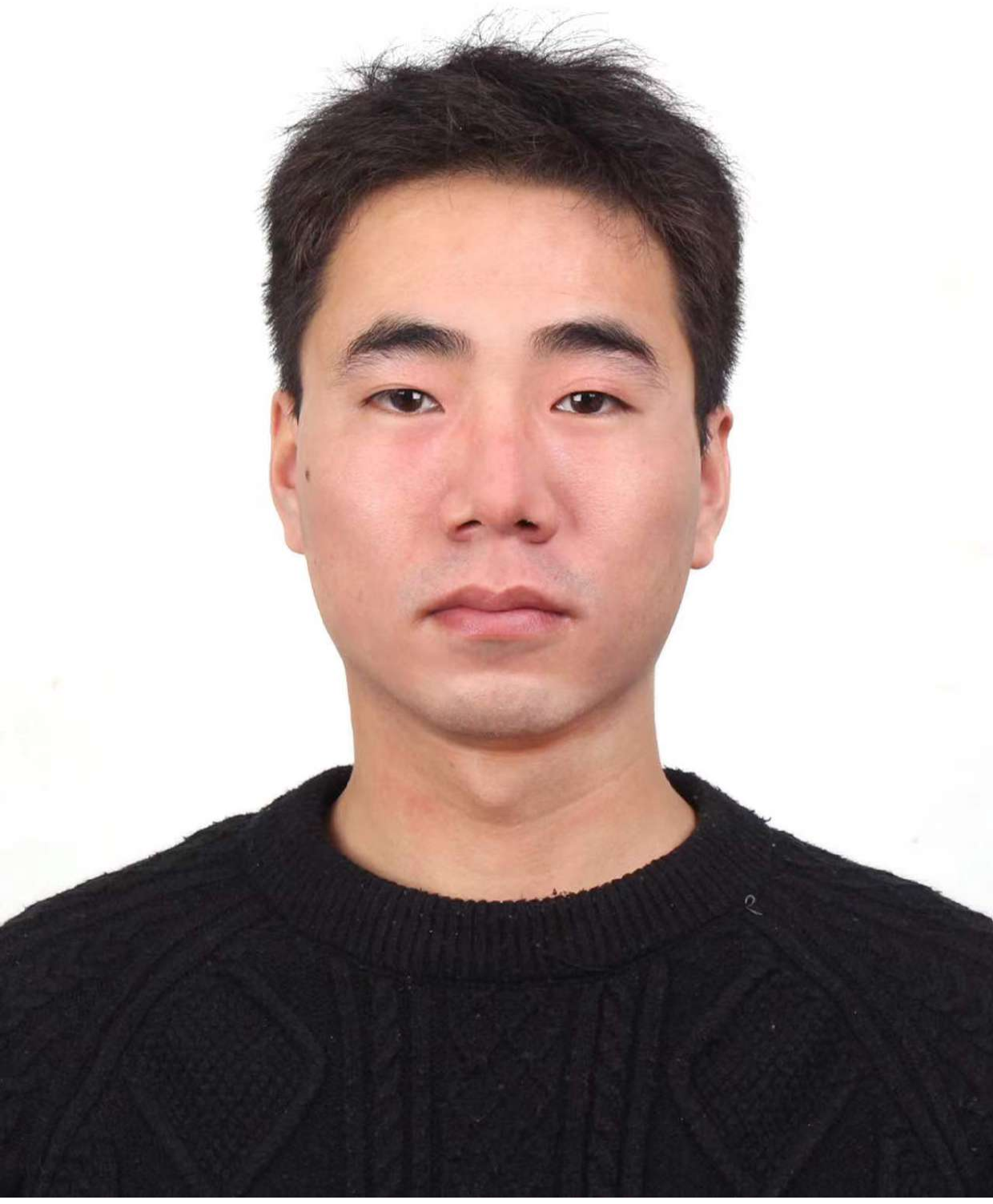}}]{Long He} received a BS degree in Computer Science and Technology from Chengdu University of Technology, Sichuan, China, in 2019. He is currently working toward the PhD degree in Computer Science and Technology at Jilin University, Changchun, China. His research interests include vehicular networks and edge computing.
\end{IEEEbiography}

\begin{IEEEbiography}[{\includegraphics[width=1in,height=1.23in,clip,keepaspectratio]{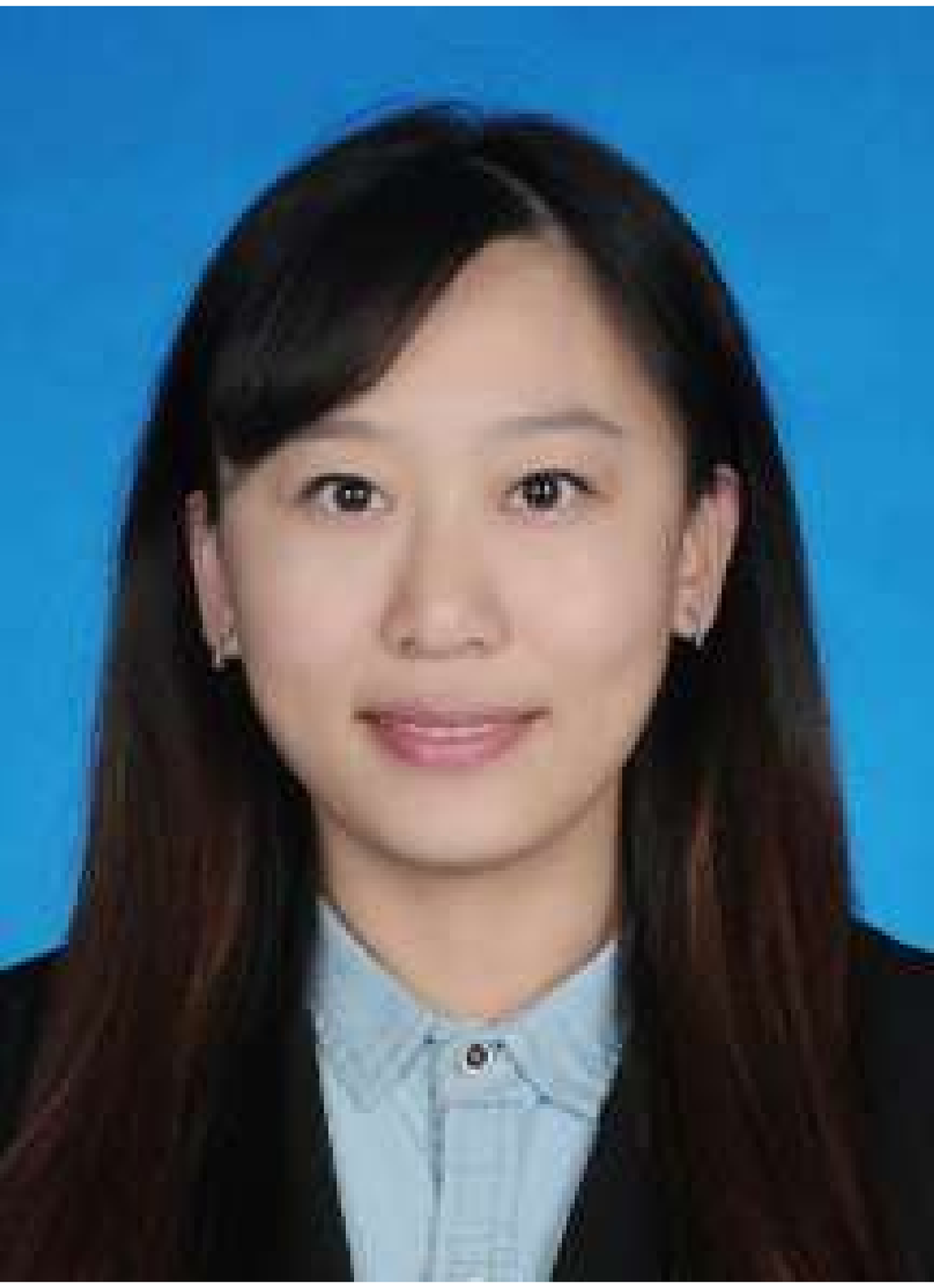}}]{Shuang Liang} received the B.S. degree in Communication Engineering from Dalian Polytechnic University, China in 2011, the M.S. degree in Software Engineering from Jilin University, China in 2017, and the Ph.D. degree in Computer Science from Jilin University, China in 2022. She is a post-doctoral in the School of Information Science and Technology, Northeast Normal University, and her research interests focus on wireless communication and UAV networks.
\end{IEEEbiography}

\begin{IEEEbiography}[{\includegraphics[width=1in,height=1.23in,clip,keepaspectratio]{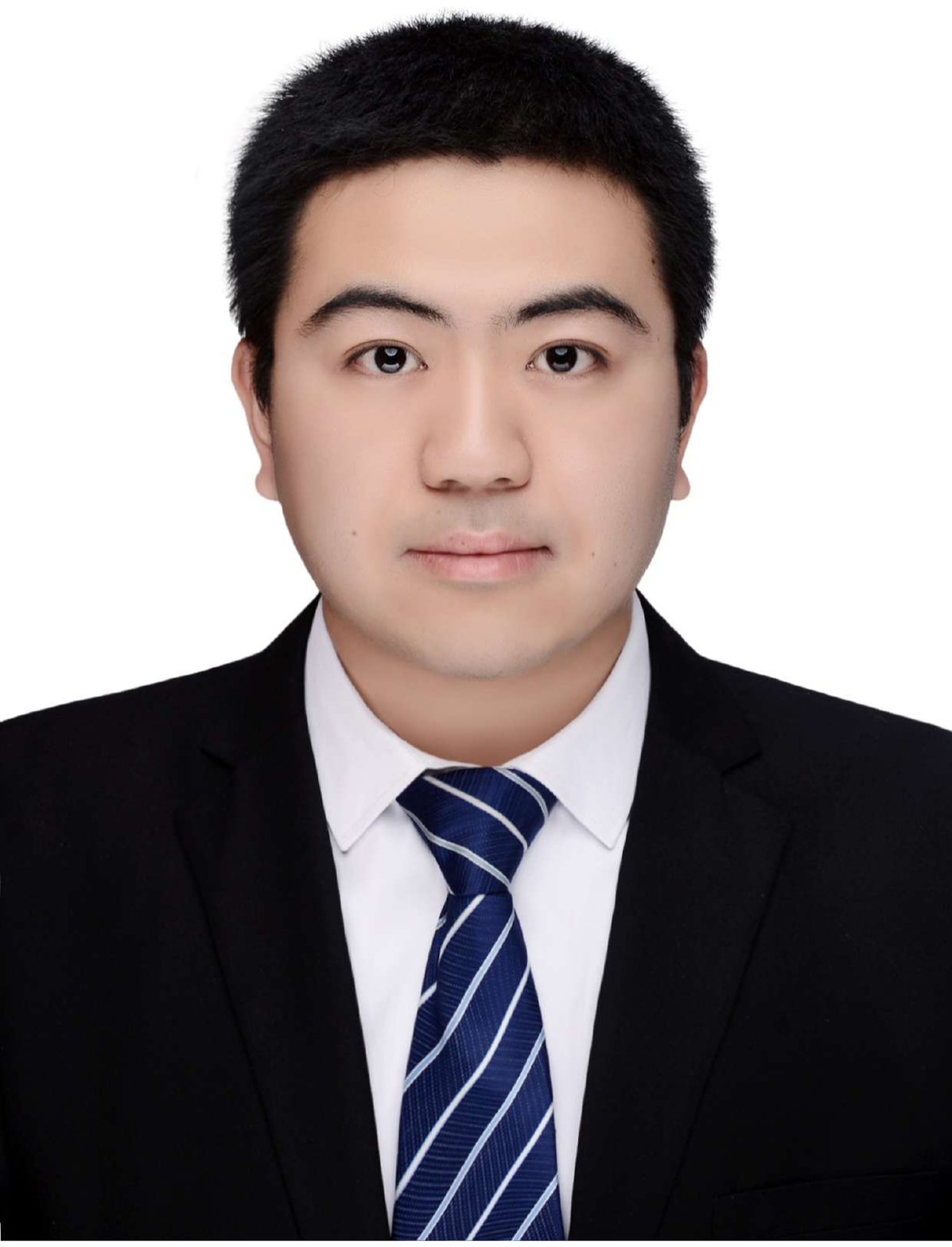}}]{Hongyang Pan} received the B.S. degree in process equipment and control engineering from Dalian
University of Technology in 2017. He is currently
pursuing the Ph.D. degree in computer science and
technology, Jilin University. His research interests
include the wireless communications and optimizations.
 
\end{IEEEbiography}

\begin{IEEEbiography}[{\includegraphics[width=1in,height=1.23in,clip,keepaspectratio]{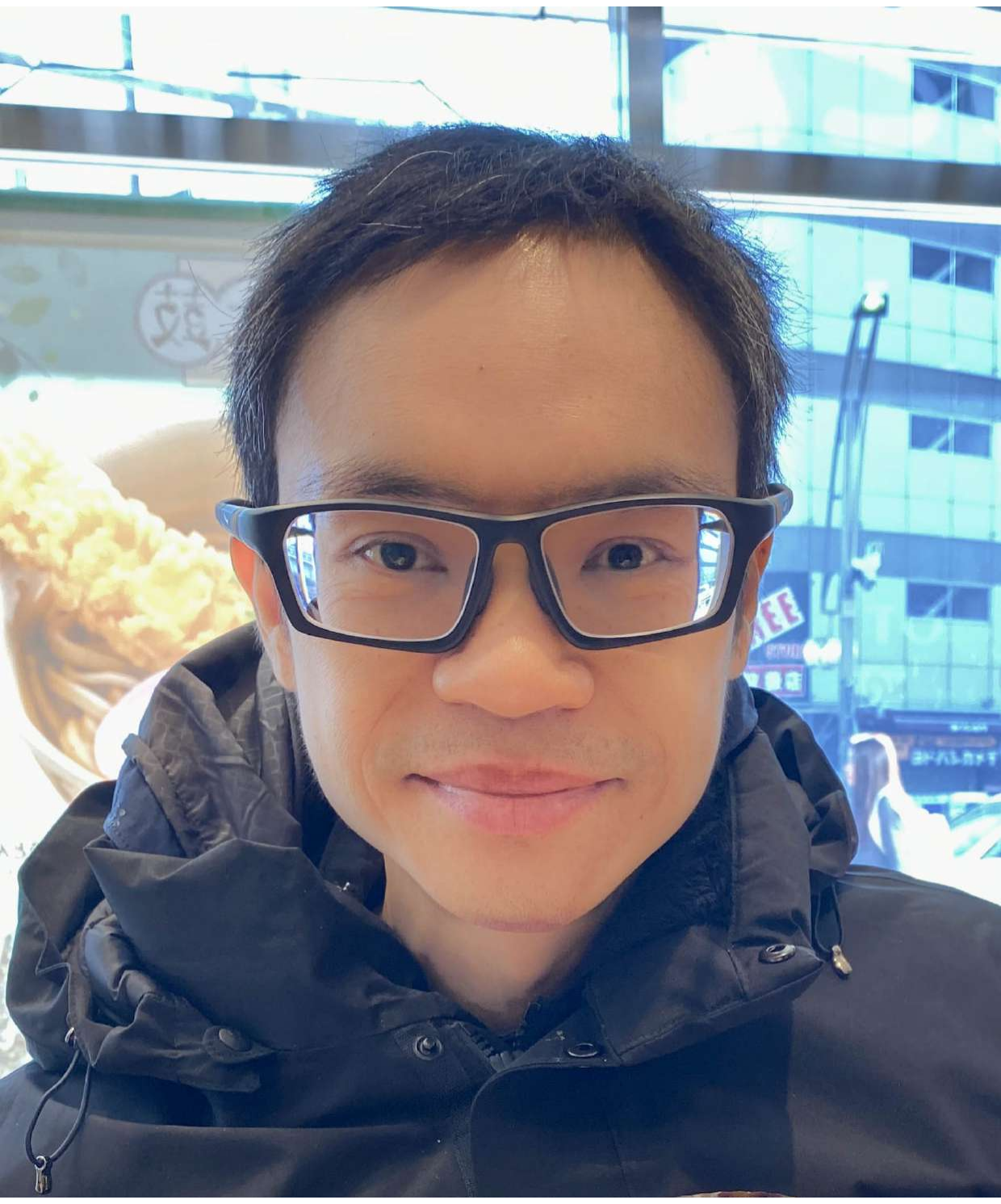}}]{Dusit Niyato} (Fellow, IEEE) received the B.Eng. degree from the King Mongkuts Institute of Technology Ladkrabang (KMITL), Thailand, in 1999, and the Ph.D. degree in electrical and computer engineering from the University of Manitoba, Canada, in 2008. He is currently a Professor with the School of Computer Science and Engineering, Nanyang Technological University, Singapore. His research interests include the Internet of Things (IoT), machine learning, and incentive mechanism design.
\end{IEEEbiography}

\begin{IEEEbiography}[{\includegraphics[width=1in,height=1.23in,clip,keepaspectratio]{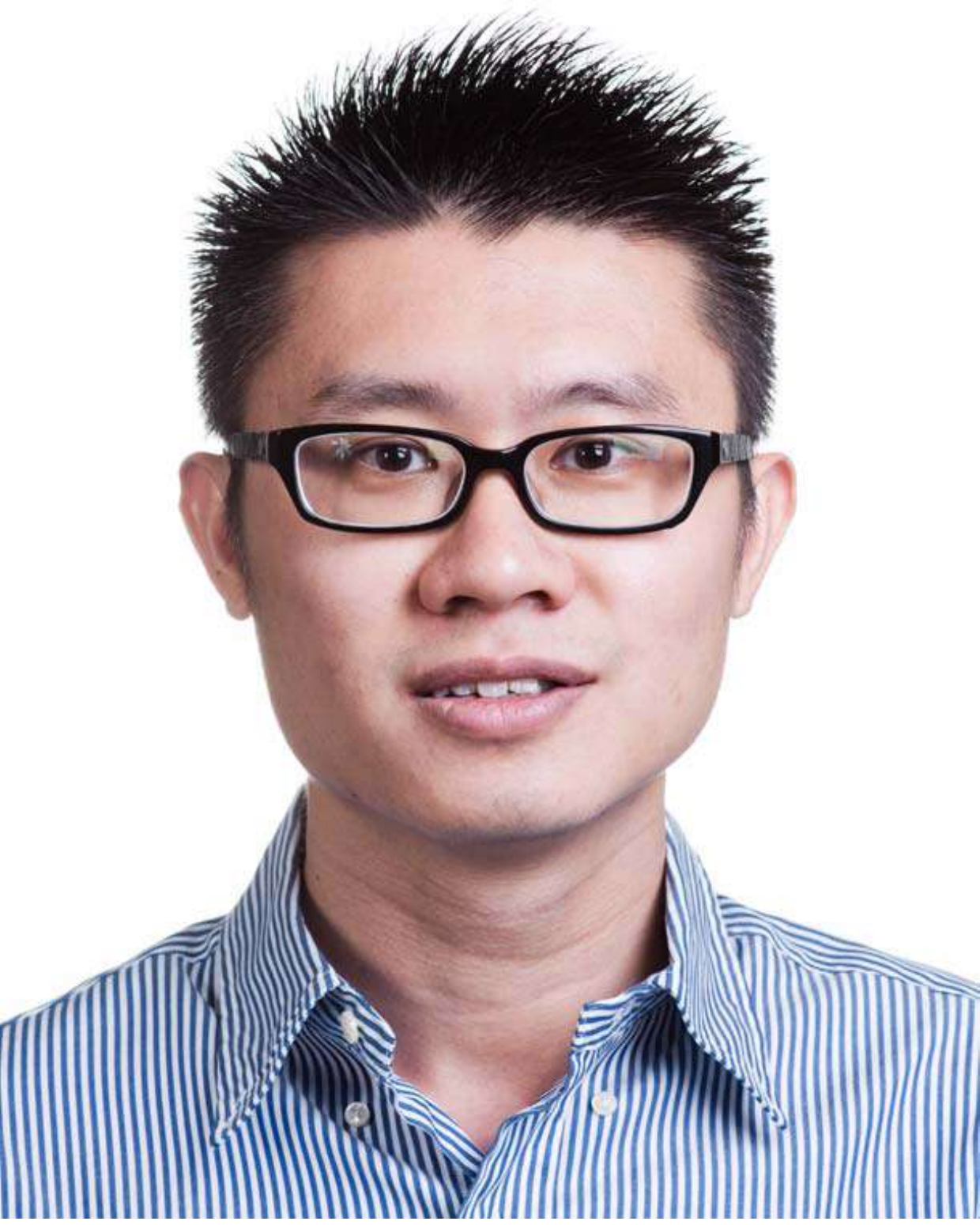}}]{Chau Yuen} (Fellow, IEEE)  received the B.Eng. and Ph.D. degrees in information and communication from Nanyang Technological University, Singapore, in 2000 and 2004, respectively. He was a Post-Doctoral Fellow at the Lucent Technologies Bell Laboratories, Murray Hill, in 2005. From 2006 to 2010, he was at the Institute for Infocomm Research (12R), Singapore. Since 2010, he has been with the Singapore University of Technology and Design.
\end{IEEEbiography}

\begin{IEEEbiography}[{\includegraphics[width=1in,height=1.23in,clip,keepaspectratio]{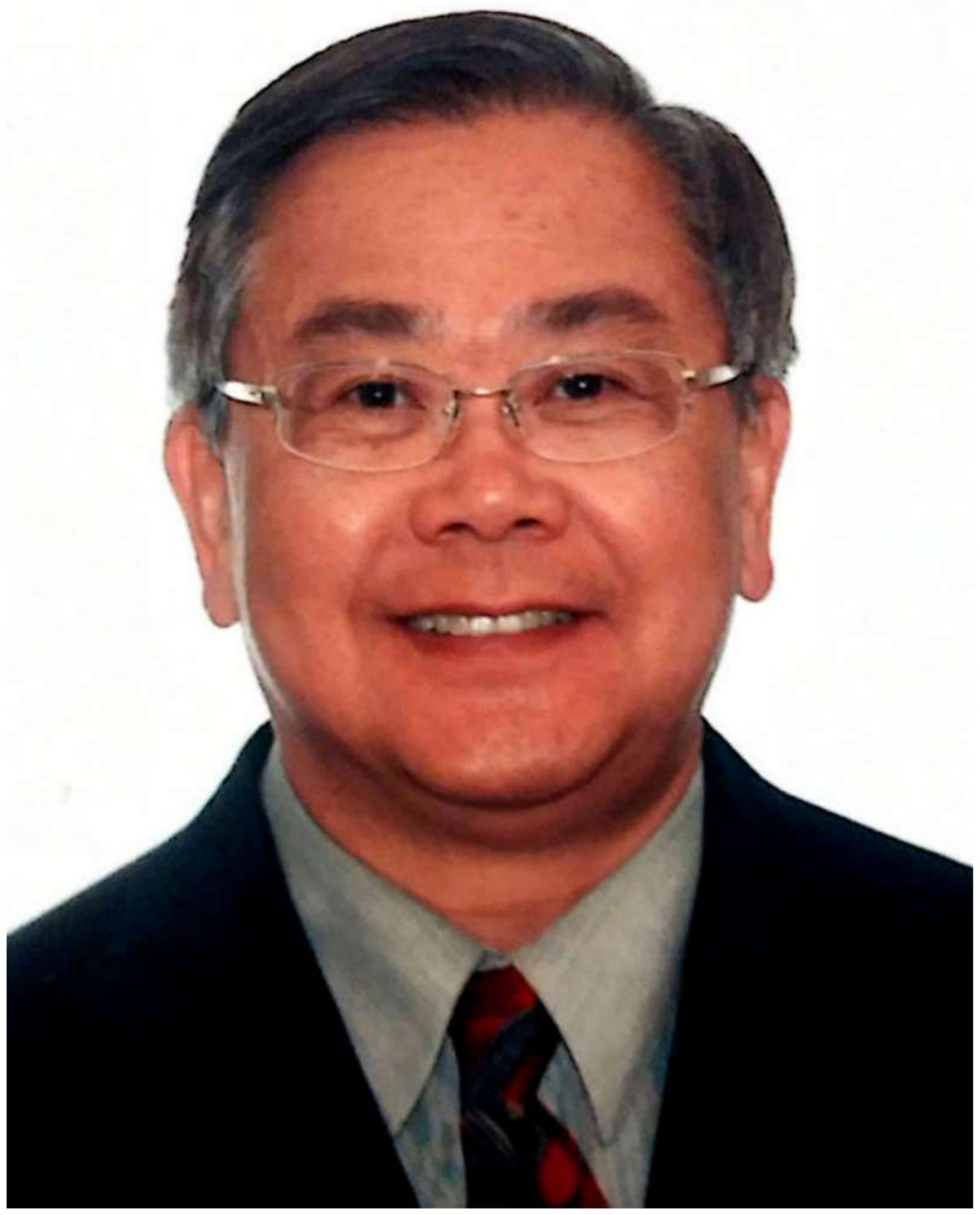}}]{Victor C. M. Leung} (Life Fellow, IEEE) is a Distinguished Professor of computer science and software engineering with Shenzhen University,
China. He is also an Emeritus Professor of electrial and computer engineering and the Director of the Laboratory for Wireless Networks and Mobile Systems at the University of British Columbia (UBC). His research is in the broad areas of wireless networks and mobile systems. 
\end{IEEEbiography}

\end{document}